\documentclass[11pt]{article}
\input{preamble}

\title{On the Equivalence between Neyman Orthogonality \\and Pathwise Differentiability}
\author{Yuxi Chen$^{1}$, Edward H. Kennedy$^{1}$, and Sivaraman Balakrishnan$^{1,2}$\\\\
$^{1}$Department of Statistics \& Data Science\\
$^{2}$Machine Learning Department\\
Carnegie Mellon University\\\\
\texttt{\{eric, edward, siva\}@stat.cmu.edu}
}
\date{}

\begin{document}
\maketitle

\begin{abstract}
It has been frequently observed that Neyman orthogonality, the central device underlying double/debiased machine learning \citep{chernozhukov_doubledebiased_2018}, and pathwise differentiability, a cornerstone concept from semiparametric theory, often lead to the same debiased estimators in practice. Despite the widespread adoption of both ideas, the precise nature of this equivalence has remained elusive, with the two concepts having been developed in largely separate traditions. In this work, we revisit the semiparametric framework of \citet{van_der_laan_unified_2003} and identify an implicit regularity assumption on the relationship between target and nuisance parameters---a local product structure---that allows us to establish a formal equivalence between Neyman orthogonality and pathwise differentiability. We also show that the two directions of this equivalence impose fundamentally different structural requirements. Finally, we illustrate the theory through three detailed examples of estimating the average treatment effect and expected density in a nonparametric model, as well as the slope in a partially linear model. This helps clarify the relationship between these two foundational frameworks and provides a useful reference for practitioners working at their intersection.

\end{abstract}

\section{Introduction}
\label{sec:intro}

In recent years, the double/debiased machine learning (DML) framework of \citet{chernozhukov_doubledebiased_2018} has become a standard tool in modern causal inference for estimating low-dimensional parameters in the presence of high-dimensional nuisance functions. The central feature of DML is that the estimating function satisfies \emph{Neyman orthogonality}: an estimating function $m(Z;\beta,\eta)$ is Neyman orthogonal if the G\^ateaux derivative of the expected estimating function
with respect to the nuisance parameter $\eta$, evaluated at the true parameter values $(\beta_0, \eta_0)$, vanishes in all admissible perturbation directions. This first-order insensitivity to the nuisance ensures that bias from estimating $\eta_0$ enters only at second order, enabling the use of flexible machine learning estimators for nuisance functions while preserving desirable properties of the target estimator. We refer the reader to \citet{chernozhukov_doubledebiased_2018} for a thorough treatment of these statistical consequences.

It has long been observed that Neyman orthogonal estimating functions coincide, in essentially every example of interest, with influence functions of \emph{pathwise differentiable} functionals from classical semiparametric theory, which underpin the construction of efficient estimators \citep{newey_asymptotic_1994, bickel_efficient_1998, van_der_vaart_asymptotic_1998, van_der_laan_unified_2003, tsiatis_semiparametric_2006}. For example, the augmented inverse probability weighted estimator for the average treatment effect arises naturally both as a one-step correction built from the efficient influence function and as the solution to a Neyman orthogonal moment condition. 

However, the two concepts have largely been developed and invoked in separate traditions. A general characterization of the relationship between pathwise differentiability and Neyman orthogonality, and in particular of what structural conditions each direction of the implication requires, does not appear to have been explicitly formulated in the literature. We hope to close this gap by formalizing the equivalence, and by clarifying the structural and regularity conditions that underpin each direction of the implication. For simplicity, we restrict our attention to scalar-valued functionals, although the results extend generally to vector-valued scenarios.

Establishing the equivalence requires bridging two seemingly distinct viewpoints. Pathwise differentiability is formulated geometrically, characterizing the first-order behavior of a functional along smooth perturbations of the data-generating distribution without reference to any explicit nuisance parameterization. Neyman orthogonality, by contrast, is defined analytically through derivatives of an expected estimating function with respect to an explicitly parameterized nuisance. Relating the two turns out to require constructing smooth perturbations of the distribution that move one parameter while holding the other fixed. A natural candidate for guaranteeing that such perturbations exist is the notion of local variation independence, which requires that the attainable parameter set contain a product neighborhood of \((\beta_0, \eta_0)\). However, this condition is purely set-theoretic and does not ensure that the independently varied parameter values are connected by submodels regular enough to differentiate along. We formalize the missing regularity as a \emph{local product structure} (Assumption~\ref{assumption:product-structure}), which requires that coordinate perturbations not merely exist as points in the model but form regular submodels through \(P_0\). This condition underlies the classical framework of \citet{van_der_laan_unified_2003}, where it is implicitly invoked but not separately identified. We make this explicit and discuss its role in their proofs in Appendix~\ref{appendix:product_structure}.

Equipped with this assumption, we establish the equivalence between Neyman orthogonality and pathwise differentiability. The forward direction (Theorem~\ref{thm:forward}) shows that a Neyman orthogonal estimating function with a nondegenerate Jacobian induces an influence function, and hence pathwise differentiability, without requiring any variation independence or product structure. The reverse direction (Theorem~\ref{thm:reverse}) shows that a mean-zero estimating function whose value at the truth is an influence function is automatically Neyman orthogonal. This direction does require local product structure in order to identify coordinate submodels that perturb $\beta$ and $\eta$ independently.

\section{Background}
\label{sec:setup}

We work on a measurable space \((\cZ, \cA)\) and fix a \(\sigma\)-finite measure \(\nu\) such that every \(P \in \cP\) is dominated by \(\nu\). We denote the density of \(P \in \cP\) by \(p = dP/d\nu\) and fix \(P_0 \in \cP\) with density \(p_0.\) We write \(\bE_0[\cdot] \equiv \bE_{P_0}[\cdot]\). For \(P, Q \in \cP\) with densities \(p, q\), the total variation and Hellinger distances are taken to be
\[
\mathrm{TV}(P,Q) := \frac{1}{2}\int |p-q|\,d\nu,
\qquad
H(P,Q) := \left(\frac{1}{2}\int\bigl(\sqrt p-\sqrt q\bigr)^2\,d\nu\right)^{1/2}.
\]
Next, let 
\[
L_2(P_0) = \left\{f \, : \, \bE_0\bigl[f^{2}\bigr] < \infty\right\}, \quad L_2^0(P_0) = \left\{f \in L_2(P_0) \, : \, \bE_0[f] = 0\right\}, \]
and
\[L_\infty(P_0) = \left\{f : \|f\|_\infty := \operatorname*{ess\,sup}_{P_0} |f| < \infty\right\}.
\]
To define local perturbations at \(P_0\), we consider paths through \(P_0\) inside the model \(\cP\). The appropriate regularity condition on such paths is quadratic-mean differentiability \citep{van_der_vaart_asymptotic_1998}.

\subsection{Regular Submodels and Scores}

\begin{definition}[Regular (QMD) submodel and score]
\label{def:regular_qmd_submodel}
    A regular parametric submodel through \(P_0\) is an indexed family \(\{P_t: t \in (-\epsilon, \epsilon)\} \subset \cP\) with \(P_{t = 0} \equiv P_0\) such that 
\begin{enumerate}
    \item \(P_t \ll \nu\) with density \(p_t = dP_t / d\nu\).
    \item The map \(t \mapsto P_t\) is differentiable in quadratic mean (QMD) at 0: there exists \(s \in L_2^0(P_0)\) such that \[\int \left(\frac{\sqrt{p_t} - \sqrt{p_0}}{t} - \frac{1}{2}s\sqrt{p_0}\right)^2\,d\nu \to 0 \quad \text{as \(t \to 0\)}.\]
    The function \(s\) is the score of the submodel at 0.
\end{enumerate}
\end{definition}

When it is helpful to indicate the score of a submodel, we write \(P_{t, s}\) for a regular submodel through \(P_0\) with score \(s\). We also use \(t \mapsto P_{t, s}\) and \(\{P_{t, s}\}\) interchangeably to refer to the submodel. 

One may observe that different submodels can share the same score. The score determines the first-order behavior of the submodel as a probability measure. It does not, however, by itself determine the derivative along the submodel of an arbitrary functional of \(P\), unless that functional is pathwise differentiable. It is in this sense that it is natural to work not with individual submodels but with their scores, which we collect into a single space.

Let \(\cS \subseteq L_2^0(P_0)\) be the set of scores of all regular submodels through \(P_0\).

\begin{definition}[Tangent space]
\label{def:tangent_space}
The (full) tangent space is
\[\cT := \overline{\mathrm{span}(\cS)}^{L_2(P_0)} \subset L_2^0(P_0).\]
\end{definition}

The tangent space is defined as the closed span of the scores, but it remains to show that scores can be constructed in a controlled way. One simple and standard construction is the linear tilt, where one perturbs $p_0$ by a multiplicative factor $1 + tg$ for a bounded, mean-zero function $g$, producing a regular submodel with score exactly $g$. 

\begin{lemma}[Linear tilt submodel is QMD with score \(g\)]
\label{lem:linear_tilt_submodel}
Let \(g \in L_\infty (P_0)\) with \(\bE_0[g] = 0\). Let \(M := \|g\|_\infty\). For \(|t| < 1/M\), define
\[p_t(z) := p_0(z)\{1 + tg(z)\}.\]
Then \(p_t \ge 0\) \(\nu\)-a.s., \(\int p_t \, d\nu = 1\), and the resulting submodel \(\{P_t: |t| < 1/M\}\) is regular (QMD) at \(0\) with score \(s \equiv g\).
(Proof in Appendix~\ref{appendix:proof_linear_tilt_submodel}.)
\end{lemma}

It should be noted that these submodels are not necessarily intended as realistic data-generating mechanisms but rather as analytical tools for assessing the local geometry of the model. Indeed, in the nonparametric model, linear tilts alone suffice to saturate the tangent space.

\begin{corollary}[Saturation in the nonparametric model]
\label{cor:saturation}
    Suppose \(\cP\) is the full nonparametric model (all densities \(p\) w.r.t. \(\nu\)). Then
    \[\cT = L_2^0(P_0).\]
\end{corollary}

\begin{proof}
    By Lemma~\ref{lem:linear_tilt_submodel}, every bounded mean-zero \(g\) is a score so \(L_\infty (P_0) \cap L_2^0(P_0) \subset \cS\). Since bounded functions are dense in \(L_2(P_0)\), it follows that \(L_\infty (P_0) \cap L_2^0(P_0)\) is dense in \(L_2^0(P_0)\). Taking the closed linear span of these scores yields \(\cT = L_2^0(P_0)\).
\end{proof}

\subsubsection{Differentiating Expectations along Regular Submodels}
\label{subsubsec:differentiate_expectations_regular_submodels}

Deriving the central results of this note requires differentiating expectations of the form \(\bE_{P_{t, s}}[f(Z)]\) along regular submodels, where the integrand \(f\) itself may also depend on \(t\). The first result below handles the case for a fixed integrand, and the second extends to integrands that vary along the submodel, which arises naturally when the integrand depends on parameters that move with \(P_{t, s}\).

\begin{lemma}[Differentiation of expectations for a fixed \(f\)]
\label{lem:differentiation_fixed_f}

Let \(t \mapsto P_{t, s}\) be a regular (QMD) submodel through \(P_0\) with bounded score \(s\). If \(f: \cZ \to \RR\) is bounded and measurable, then 
\[\frac{d}{dt}\bE_{P_{t, s}}[f(Z)] \bigg|_{t = 0} = \bE_0[f(Z) \, s(Z)].\]
(Proof in Appendix~\ref{appendix:proof_differentiation_fixed_f}.)
\end{lemma}

\begin{lemma}[Differentiation of expectations for varying \(f_t\)]
\label{lem:differentiation_varying_f}
    Let \(t \mapsto P_{t, s}\) be a regular (QMD) submodel through \(P_0\) with bounded score \(s\). Let \(f_t : \cZ \to \RR\) be measurable for each \(t\), with \(f_0\) bounded. Suppose
    \begin{enumerate}
        \item There exists \(\dot f_0 \in L_2(P_0)\) such that 
        \[\bE_0 \left[\left(\frac{f_t - f_0}{t} - \dot f_0\right)^2\right] \to 0 \quad \text{as \(t \to 0\)}.\]
        \item There exists \(\delta > 0, C < \infty\) such that 
        \[\sup_{|
        t| < \delta} \bE_{P_{t, s}} \left[\left(\frac{f_t - f_0}{t}\right)^2\right] \le C.\]
    \end{enumerate}
    Then \(t \mapsto \bE_{P_{t, s}}[f_t]\) is differentiable at \(0\) and 
    \[\frac{d}{dt}\bE_{P_{t, s}}[f_t] \bigg|_{t = 0} = \bE_0[f_0s] + \bE_0\bigl[\dot f_0\bigr].\](Proof in Appendix~\ref{appendix:proof_differentiation_varying}.)
\end{lemma}

\subsubsection{Nuisance Scores, Nuisance Tangent Space, and Pathwise Derivatives}
\label{subsubsec:nuisance}

The tools developed in Section~\ref{subsubsec:differentiate_expectations_regular_submodels} allow us to differentiate expectations along regular submodels, but do not yet distinguish between perturbations that change the parameter of interest and those that do not. To clarify this distinction, we define nuisance scores, the nuisance tangent space, and influence functions following \citet{van_der_laan_unified_2003}, and show that influence functions are orthogonal to the nuisance tangent space.

Let $\beta : \cP \to \RR$ be the target parameter of interest with $\beta_0 := \beta(P_0)$.

\begin{definition}[Nuisance scores and nuisance tangent space]\label{def:nuisance}
    Assume that for every regular submodel $t\mapsto P_{t, s}$ through \(P_0\), the derivative $\frac{d}{dt}\beta(P_{t, s})|_{t=0}$ exists. Define the nuisance score set
  \[
    \cS_{\mathrm{nuis}}
    :=
    \left\{
      s\in\cS:\, \exists\, \text{a regular submodel } t\mapsto P_{t, s} \text{ with score } s
      \text{ such that }
      \frac{d}{dt}\beta(P_{t, s})\bigg|_{t=0}=0
    \right\}.
  \]
  Define the nuisance tangent space
  \[
    \Lambda := \overline{\operatorname{span}(\cS_{\mathrm{nuis}})}^{\,L_2(P_0)}
    \subset \cT.
  \]
\end{definition}

It is worth noting that nuisance scores are defined without reference to any explicit nuisance parameterization. Concretely, a score $s$ is a nuisance if there exists a regular submodel with score $s$ along which $\beta$ is locally constant to first order. Following the discussion after Definition~\ref{def:regular_qmd_submodel}, the score need not by itself determine $\frac{d}{dt}\beta(P_{t, s})|_{t=0}$, as two regular submodels can share the same score while yielding different derivatives of $\beta$, which is why the definition quantifies over the existence of such a submodel. Furthermore, $\Lambda$ is a closed linear subspace of $\cT$ generated by score directions that admit regular submodels along which $\beta$ is locally constant to first order.

\begin{definition}[Pathwise differentiability and influence functions]\label{def:pathwise}
  We say $\beta$ is pathwise differentiable at $P_0$ if there exists $\varphi \in L_2^0(P_0)$ such that for every regular submodel with score $s$,
  \[
    \frac{d}{dt}\beta(P_{t,s})\bigg|_{t=0} = \bE_0[\varphi(Z ; P_0) \, s(Z)].
  \]
  Any such $\varphi$ is called an influence function of $\beta$ at $P_0$ or a gradient of the pathwise derivative.
\end{definition}
Here, \(\varphi(Z; P_0)\) indicates that \(\varphi\) is a functional of \(P_0\) evaluated at the data point \(Z\). Since \(P_0\) is fixed throughout, we write simply \(\varphi(Z)\) hereafter. 

Note that if $\beta$ is pathwise differentiable at $P_0$, then the derivative depends only on the score, in which case the condition in Definition~\ref{def:nuisance} is equivalent to requiring that $\beta$ does not change to first order along \emph{any} regular submodel with score $s$.

\begin{remark}[Uniqueness of the influence function]\label{remark:uniqueness}
  In general, the influence function need not be unique. The pathwise derivative condition only probes $\varphi$ through inner products with scores $s \in \cT$, so adding any $h \in \cT^\perp$
  to $\varphi$ produces another valid influence function. Only the projection onto $\cT$ is identified by the pathwise derivative. This projection is called the efficient influence function and
  is the unique influence function lying in $\cT$. In the nonparametric model, $\cT = L_2^0(P_0)$ by Corollary~\ref{cor:saturation}, so $\cT^\perp = \{0\}$ and the influence function is unique. For nonparametric models, the influence function and the efficient influence function coincide.
\end{remark}

\begin{lemma}[Influence functions are orthogonal to $\Lambda$]
\label{lem:ortho_nuisance}
  If $\beta$ is pathwise differentiable with influence function $\varphi$ (Definition~\ref{def:pathwise}), then
  \[
    \bE_0[\varphi(Z)\, s(Z)] = 0 \quad \forall s \in \Lambda.
  \]
\end{lemma}

\begin{proof}
  Let $s \in \cS_{\mathrm{nuis}}$. By Definition~\ref{def:nuisance}, there exists a regular submodel with score $s$ along which $\frac{d}{dt}\beta(P_{t,s})|_{t=0} = 0$. By pathwise differentiability,
  \[
    0 = \frac{d}{dt}\beta(P_{t,s})\bigg|_{t=0} = \bE_0[\varphi(Z)\, s(Z)].
  \]
  Since the map $s \mapsto \bE_0[\varphi s]$ is a continuous linear functional on $L_2(P_0)$, the equality extends from $\cS_{\mathrm{nuis}}$ to its closed linear span $\Lambda$.
\end{proof}

Lemma~\ref{lem:ortho_nuisance} says that the influence function is orthogonal to every direction in the nuisance tangent space. This can be considered as an analogue of Neyman orthogonality, which requires that the expected estimating function be insensitive to perturbations of the nuisance parameter, but formulated without reference to any explicit parameterization. Establishing a formal equivalence between these two formulations, as we do in Section~\ref{sec:equivalence}, will rely on the product structure developed in the next section to identify nuisance perturbations with nuisance scores in \(\Lambda\).

\subsection{Estimating Functions and Neyman Orthogonality }
\label{subsec:estimating}

To formulate Neyman orthogonality, we will need to work with estimating functions of the form \(m(Z; \beta, \eta)\) that depend explicitly on both a target parameter \(\beta\) and a nuisance parameter \(\eta\). This requires us to move beyond the framework of Section~\ref{subsubsec:nuisance}, where nuisance scores were defined without reference to any explicit parameterization, and specify concrete functionals on the model that assume the roles of the target and nuisance. Once such a parameterization is in place, it is natural to ask what structure the relationship between \(\beta\) and \(\eta\) must possess for the two viewpoints to agree. Pathwise differentiability is defined through scores alone and makes no reference to how the nuisance is parameterized, while Neyman orthogonality depends explicitly on the functional form of \(\beta\) and \(\eta\). As we show below, connecting these two viewpoints requires the ability to construct submodels that move one coordinate while holding the other fixed.

As before, let
\[\beta : \cP \to \RR, \qquad \eta : \cP \to \cH\]
be functionals on the model, where $\cH \subset \cV$ is a subset of a normed vector space with the norm denoted by \(\|\cdot\|_\cV\). We let $\beta_0 := \beta(P_0)$ and $\eta_0 := \eta(P_0)$. For any pair of $(\beta, \eta)$ in the attainable set $\Theta := \{(\beta(P), \eta(P)) : P \in \cP\}$, we write \(P_{\beta,\eta}\) for a distribution in \(\cP\) with \(\beta(P_{\beta,\eta}) = \beta\) and \(\eta(P_{\beta,\eta}) = \eta\), so that for any \(P \in \cP\), the expectation \(\bE_{P}[f(Z; \beta(P), \eta(P))]\) can be written \(\bE_{P_{\beta,\eta}}[f(Z; \beta, \eta)]\) with \((\beta, \eta) = (\beta(P), \eta(P))\). Finally, let \[\dot{\cH} := \left \{h \in \cV: \exists \epsilon > 0 \text{ such that  } \eta_0 + th \in \cH \text{ for all } |t| < \epsilon \right\}\] denote the set of admissible perturbation directions at $\eta_0$. 

\subsubsection{Local Product Structure}
\label{subsec:local_product_structure}

To apply the differentiation results of Section~\ref{subsubsec:differentiate_expectations_regular_submodels}, we require an additional local product structure assumption, which ensures the existence of regular (QMD) submodels along each coordinate, that is, submodels that perturb one of \(\beta\) or \(\eta\) while holding the other fixed. Note that any regular submodel \(t \mapsto P_{t,s}\) through \(P_0\) induces a coordinate path \(t \mapsto (\beta_{t,s},\, \eta_{t,s}) := (\beta(P_{t,s}),\, \eta(P_{t,s}))\). The following assumption requires that this coordinate path can be controlled independently in each component.

\begin{assumption}[Local product structure]
\label{assumption:product-structure}
The following first-order coordinate conditions hold:
\begin{enumerate}
  \item \textbf{\(\beta\)-coordinate submodel.} There exists a regular (QMD) submodel \(t \mapsto P_t \in \cP\) through \(P_0\) along which the induced coordinate path is differentiable at \(t = 0\) with
  \[\frac{d}{dt}\beta(P_t)\bigg|_{t=0} = 1 \qquad \text{and} \qquad \frac{d}{dt} \eta(P_t)\bigg|_{t = 0} = 0.\]
  \item \textbf{\(\eta\)-coordinate submodel.} For every admissible nuisance perturbation direction \(h \in \dot{\cH}\), there exists a regular (QMD) submodel \(t \mapsto P_t \in \cP\) through \(P_0\) along which the induced coordinate path is differentiable at \(t = 0\) with
  \[\frac{d}{dt}\beta(P_t)\bigg|_{t=0} = 0 \qquad \text{and} \qquad \frac{d}{dt} \eta(P_t)\bigg|_{t = 0} = h.\]
\end{enumerate}
\end{assumption}

This formalizes a condition implicit in the framework of \citet[p.~56]{van_der_laan_unified_2003}, where the model is written as \(\{F_{\mu,\eta}\}\) with \(\mu\) and \(\eta\) ``independently varying,'' and submodels varying only the nuisance parameter are used to generate the nuisance tangent space. Note that Assumption~\ref{assumption:product-structure} requires only first-order control, where the derivatives of \(\beta(P_t)\) and \(\eta(P_t)\) at \(t = 0\) are prescribed, but the paths need not satisfy \(\beta(P_t) = \beta_0 + t\) or \(\eta(P_t) = \eta_0 + th\) exactly for \(t \neq 0\). 

We discuss the relationship between Assumption~\ref{assumption:product-structure} and the notion of local variation independence, as well as the role of product structure in the proof of Lemma~1.3 of \citet{van_der_laan_unified_2003}, in Appendix~\ref{appendix:product_structure}.

\subsubsection{Neyman Orthogonality}
\label{subsubsection:neyman_ortho}

Next, let \(m : \cZ \times \RR \times \cH \to \RR\) be such that \(z \mapsto m(z; \beta, \eta)\) is \(\cA\)-measurable for each 
\((\beta, \eta)\). The function \(m\) plays the role of an estimating function, encoding a moment condition whose solution at the true nuisance value identifies \(\beta_0\), while the explicit dependence on \(\eta\) reflects the presence of nuisance quantities that need to be estimated.

\begin{definition}[Correct local specification]
\label{def:correct_spec}
We say that $m$ is correctly specified in a neighborhood of $(\beta_0, \eta_0)$ if for all $P \in \cP$ with $(\beta(P), \eta(P))$ in a neighborhood of $(\beta_0, \eta_0)$,
\[
  \bE_P[m(Z;\, \beta(P),\, \eta(P))] = 0.
\]
\end{definition}
Correct specification ensures that \(\beta_0\) solves the moment condition at the true nuisance, but does not constrain how the expected estimating function varies with \(\eta\) near \(\eta_0\). Neyman orthogonality strengthens this by requiring that this variation vanish to first order, so that small errors in \(\eta\) do not propagate to estimation of \(\beta\).

\begin{definition}[Neyman orthogonality]
\label{def:neyman_orthogonality}
Assume the map \(\eta \mapsto \bE_0[m(Z; \beta_0, \eta)]\) is G\^ateaux differentiable at \(\eta_0\) along directions \(h \in \dot{\cH}\). We say \(m\) is Neyman orthogonal at \((\beta_0, \eta_0)\) if
\[\frac{\partial}{\partial \eta}\bE_0[m(Z; \beta_0, \eta)]
\bigg|_{\eta = \eta_0}[h] = 0 \quad \forall 
h \in \dot{\cH}.\]
\end{definition}

\begin{remark}
\label{remark:neyman_connection}
The G\^ateaux derivative in Definition~\ref{def:neyman_orthogonality} is computed under the fixed measure \(P_0\) with \(\beta_0\) held fixed. The map \(\eta \mapsto \bE_0[m(Z; \beta_0, \eta)]\) is defined for any \(\eta \in \cH\) for which the integral exists, without requiring that \((\beta_0, \eta)\) correspond to a distribution in the model \(\cP\). In particular, no variation independence 
is needed to formulate Neyman orthogonality. The role of Assumption~\ref{assumption:product-structure} is instead to establish that Neyman orthogonality holds for influence functions.
\end{remark}

\subsection{The \texorpdfstring{$L_2$}{L2} Chain Rule along Coordinate Paths}
\label{subsec:l2_chain_rule}

To connect estimating functions with pathwise differentiability, we also need to differentiate the estimating function 
\(m(Z; \beta, \eta)\) along the coordinate path induced by a regular submodel. The next two assumptions regulate the behavior of this coordinate path and of the estimating function along it.

\begin{assumption}[Coordinate smoothness along a submodel]
\label{assumption:coordinate_smoothness}
For a given regular (QMD) submodel \(t \mapsto P_{t,s}\) through 
\(P_0\) with score \(s\), the induced coordinate path satisfies:
\begin{enumerate}
    \item \(t \mapsto \beta_{t,s} := \beta(P_{t,s})\) is 
    differentiable at \(0\): 
    \((\beta_{t,s} - \beta_0)/t \to \dot\beta_{0,s} \in \RR\).
    \item \(t \mapsto \eta_{t,s} := \eta(P_{t,s})\) is 
    differentiable at \(0\) in \(\cV\): 
    \(\|(\eta_{t,s} - \eta_0)/t - \dot\eta_{0,s}\|_{\cV} \to 0\) 
    for some \(\dot\eta_{0,s} \in \dot{\cH}\).
\end{enumerate}
In particular, 
\(\dot\beta_{0,s} = \frac{d}{dt}\beta(P_{t,s})|_{t=0}\).
\end{assumption}

\begin{assumption}[Fr\'echet differentiability of \(m\) in 
\(L_2(P_0)\)]
\label{assumption:frechet_m}
The map \((\beta, \eta) \mapsto m(\cdot\,; \beta, \eta) \in 
L_2(P_0)\) is Fr\'echet differentiable at \((\beta_0, \eta_0)\). 
That is, there exist bounded linear maps
\[D_\beta m_0 : \RR \to L_2(P_0), \qquad D_\eta m_0 : \cV \to 
L_2(P_0)\]
such that
\[\|m(\cdot\,; \beta, \eta) - m(\cdot\,; \beta_0, \eta_0) - 
D_\beta m_0(\beta - \beta_0) - D_\eta m_0(\eta - \eta_0)
\|_{L_2(P_0)} = o\!\left(|\beta - \beta_0| + \|\eta - \eta_0
\|_{\cV}\right).\]
We write \(\partial_\beta m(Z; \beta_0, \eta_0) := D_\beta 
m_0(1)(Z)\) and \(\partial_\eta m(Z; \beta_0, \eta_0)[h] := 
D_\eta m_0(h)(Z)\).
\end{assumption}

\begin{lemma}[\(L_2\) chain rule]
\label{lem:l2_chain_rule}
Under Assumptions~\ref{assumption:coordinate_smoothness} 
and~\ref{assumption:frechet_m}, define 
\(f_{t,s}(Z) := m(Z; \beta_{t,s}, \eta_{t,s})\) and
\[\dot f_{0,s}(Z) := \partial_\beta m(Z; \beta_0, \eta_0)\,
\dot\beta_{0,s} + \partial_\eta m(Z; \beta_0, \eta_0)
[\dot\eta_{0,s}].\]
Then \((f_{t,s} - f_0)/t \to \dot f_{0,s}\) in \(L_2(P_0)\).
(Proof in Appendix~\ref{appendix:proof_chain_rule}.)
\end{lemma}

\section{Equivalence Between Neyman Orthogonality and Pathwise 
Differentiability}\label{sec:equivalence}

We now establish the relationship between Neyman orthogonality and pathwise differentiability. The forward direction (Section~\ref{subsec:forward}) demonstrates that a Neyman orthogonal estimating function with nondegenerate Jacobian induces an influence function, and hence pathwise differentiability. The reverse direction (Section~\ref{subsec:reverse}) shows that if a correctly specified estimating function evaluates to an influence function at the truth, then it must be Neyman orthogonal and its sensitivity to the target parameter is fully calibrated by the influence function representation. Here, we require local product structure in order to specialize to coordinate submodels that perturb $\beta$ and $\eta$ independently.

The proofs of the two directions differ regarding their structural requirements. The forward direction requires that the induced coordinate paths be smooth along a dense class of regular submodels and that the target functional be locally Lipschitz in Hellinger distance, whereas the reverse direction requires the local product structure of Assumption~\ref{assumption:product-structure} in order to construct submodels that perturb \(\beta\) and \(\eta\) independently.

\subsection{Neyman Orthogonality Implies Pathwise 
Differentiability}
\label{subsec:forward}

Fix an estimating function \(m : \cZ \times \RR \times \cH \to \RR\). Correct specification ensures \(\bE_{P_{t,s}}[m(Z; \beta_{t,s}, \eta_{t,s})] = 0\) identically along any regular submodel, so the derivative of this constant function vanishes. Expanding the derivative via Lemma~\ref{lem:differentiation_varying_f} and the \(L_2\) chain rule (Lemma~\ref{lem:l2_chain_rule}), and then invoking Neyman orthogonality to eliminate the nuisance contribution, yields a representation of \(\dot\beta_{0,s}\) as an inner product with the score, which is exactly pathwise differentiability.

\begin{assumption}[Correct specification]
\label{assumption:correct_spec}
The estimating function \(m\) is correctly specified at \((\beta_0, \eta_0)\) in the sense of Definition~\ref{def:correct_spec}.
\end{assumption} 

\begin{assumption}[Coordinate smoothness along a dense class of submodels]
\label{assumption:coord_smooth_all}
There exists a set of scores \(S \subset L_\infty(P_0) \cap L_2^0(P_0)\) whose \(L_2(P_0)\)-closure is equal to \(\cT\) such that for each \(s \in S\), there exists a regular submodel \(t \mapsto P_{t,s}\) through \(P_0\) with score \(s\) along which the induced coordinate path
\[t \mapsto (\beta_{t,s}, \eta_{t,s}) = (\beta(P_{t,s}), 
\eta(P_{t,s}))\]
satisfies Assumption~\ref{assumption:coordinate_smoothness}.
\end{assumption}

\begin{assumption}[Fr\'echet differentiability of \(m\)]
\label{assumption:frechet_all}
The map \((\beta, \eta) \mapsto m(\cdot\,; \beta, \eta) \in 
L_2(P_0)\) satisfies 
Assumption~\ref{assumption:frechet_m}.
\end{assumption}

\begin{assumption}[Regularity along submodels]
\label{assumption:regularity_submodels}
For each \(s \in S\) and the corresponding submodel \(t \mapsto P_{t, s}\) from Assumption~\ref{assumption:coord_smooth_all}, the function \(f_{t,s}(Z) := m(Z; \beta_{t,s}, \eta_{t,s})\) satisfies the conditions of Lemma~\ref{lem:differentiation_varying_f}.
\end{assumption}

\begin{assumption}[Nondegenerate Jacobian]
\label{assumption:jacobian}
\[G := \bE_0[\partial_\beta m(Z; \beta_0, \eta_0)] \neq 0.\]
\end{assumption}

\begin{assumption}[Neyman orthogonality]
\label{assumption:neyman}
For all \(h \in \dot{\cH}\),
\[\frac{\partial}{\partial \eta}\bE_0[m(Z; \beta_0, \eta)]
\bigg|_{\eta = \eta_0}[h] = 0.\]
\end{assumption}

\begin{assumption}[Hellinger Lipschitz]
\label{assumption:hellinger}
    There exist \(c, \delta > 0\) such that 
    \[|\beta(P_1) - \beta(P_2)| \le cH(P_1, P_2) \quad \forall P_1, P_2 \in \cP  \text{ with \(H(P_i, P_0) \le \delta\)}.\]
\end{assumption}

Assumptions~\ref{assumption:correct_spec} and~\ref{assumption:neyman} are the two standard requirements on the estimating function introduced in Section~\ref{subsubsection:neyman_ortho}. Assumptions~\ref{assumption:coord_smooth_all}
through~\ref{assumption:regularity_submodels} ensure that the differentiation machinery of Section~\ref{sec:setup} applies along a dense class of regular submodels. These amount to differentiability of
the estimating function in its parameters and of the functionals \(\beta\) and \(\eta\) along these submodels, together with boundedness and integrability conditions near the truth. Assumption~\ref{assumption:jacobian} ensures that the rescaling \(\varphi = -G^{-1}m(Z;\beta_0,\eta_0)\) in the conclusion of Theorem~\ref{thm:forward} is well defined. Assumption~\ref{assumption:hellinger} provides the quantitative control needed to extend the pathwise derivative from the dense class of bounded scores to all scores. It bounds how fast $\beta$ can vary relative to the Hellinger distance between distributions, ensuring that replacing an arbitrary regular submodel by one from the dense class with a nearby score incurs a controlled error in the derivative of $\beta$.

\begin{theorem}[Neyman orthogonality implies pathwise 
differentiability]
\label{thm:forward}
Under 
Assumptions~\ref{assumption:correct_spec}--\ref{assumption:hellinger}, 
\(\beta\) is pathwise differentiable 
(Definition~\ref{def:pathwise}) at \(P_0\) with influence 
function
\[\varphi(Z) := -G^{-1}m(Z; \beta_0, \eta_0).\]
\end{theorem}

\begin{proof}
Let \(s \in S\) and let \(t \mapsto P_{t,s}\) be a regular submodel through \(P_0\) with score \(s\) as furnished by Assumption~\ref{assumption:coord_smooth_all}. By Assumption~\ref{assumption:coord_smooth_all}, the induced coordinate path \((\beta_{t,s}, \eta_{t,s})\) lies in the neighborhood of \((\beta_0, \eta_0)\) for small \(t\). Assumption~\ref{assumption:correct_spec} then gives
\[\bE_{P_{t,s}}[m(Z; \beta_{t,s}, \eta_{t,s})] = 0 \quad \text{for all sufficiently small } t.\]
Define \(f_{t,s}(Z) := m(Z; \beta_{t,s}, \eta_{t,s})\) and \(f_0(Z) := m(Z; \beta_0, \eta_0)\). Since \(t \mapsto \bE_{P_{t,s}}[f_{t,s}]\) is identically zero,
\[\frac{d}{dt}\bE_{P_{t,s}}[f_{t,s}]\bigg|_{t=0} = 0.\]
We apply Lemma~\ref{lem:differentiation_varying_f} to the function \(f_{t,s}\), which is valid by Assumption~\ref{assumption:regularity_submodels}. By the \(L_2\) chain rule (Lemma~\ref{lem:l2_chain_rule}), which applies under Assumptions~\ref{assumption:coord_smooth_all} and~\ref{assumption:frechet_all}, the quotient \((f_{t,s} - f_0)/t\) converges in \(L_2(P_0)\) to
\[\dot f_{0,s}(Z) = \partial_\beta m(Z; \beta_0, \eta_0)\,
\dot\beta_{0,s} + \partial_\eta m(Z; \beta_0, \eta_0)
[\dot\eta_{0,s}].\]
Lemma~\ref{lem:differentiation_varying_f} thus gives
\[0 = \bE_0[f_0(Z)\, s(Z)] + \bE_0[\dot f_{0,s}(Z)].\]
Substituting the expression for \(\dot f_{0,s}\) and using linearity of expectation,
\begin{equation}
\label{eq:intermediate_identity}
    0 = \bE_0[m(Z; \beta_0, \eta_0)\, s(Z)] + 
\bE_0[\partial_\beta m(Z; \beta_0, \eta_0)]\,\dot\beta_{0,s} + 
\bE_0[\partial_\eta m(Z; \beta_0, \eta_0)[\dot\eta_{0,s}]].
\end{equation}
Now, \(\dot\eta_{0,s} \in \dot{\cH}\) by Assumption~\ref{assumption:coord_smooth_all}, and Fr\'echet differentiability (Assumption~\ref{assumption:frechet_all}) permits the interchange of derivative and expectation, so that
\[\bE_0[\partial_\eta m(Z; \beta_0, \eta_0)[h]] = 
\frac{\partial}{\partial \eta}\bE_0[m(Z; \beta_0, \eta)]
\bigg|_{\eta = \eta_0}[h],\]
which vanishes by Neyman orthogonality (Assumption~\ref{assumption:neyman}). Recalling \(G = \bE_0[\partial_\beta m(Z; \beta_0, \eta_0)]\), we are left with
\[0 = \bE_0[m(Z; \beta_0, \eta_0)\, s(Z)] + G\,\dot\beta_{0,s}.\]
Since \(G \neq 0\) by Assumption~\ref{assumption:jacobian},
\[\dot\beta_{0,s} = -G^{-1}\bE_0[m(Z; \beta_0, \eta_0)\, s(Z)] 
= \bE_0[\varphi(Z)\, s(Z)],\]
where \(\varphi(Z) = -G^{-1}m(Z; \beta_0, \eta_0)\). We note that \(\varphi \in L_2^0(P_0)\). Lemma~\ref{lem:differentiation_varying_f}, invoked via Assumption~\ref{assumption:regularity_submodels}, requires \(f_0 = m(Z; \beta_0, \eta_0)\) to be bounded and measurable. Since this is a property of \(f_0\) alone and does not depend on the choice of submodel, \(\varphi = -G^{-1}f_0\) is bounded and measurable, hence in \(L_2(P_0)\). Mean zero follows from correct specification at the truth (Assumption~\ref{assumption:correct_spec}).

It remains to extend the conclusion to all regular submodels. To start, let \(t \mapsto P_{t, s'}\) be an arbitrary regular submodel through \(P_0\) with score \(s' \in \cS\) and fix \(\epsilon > 0\). Since \(S\) is dense in \(\cT\) by Assumption~\ref{assumption:coord_smooth_all}, there exists \(g \in S\) with 
\[\|s' - g\|_{L_2(P_0)} \le \epsilon. \]
Let \(t \mapsto P_{t, g}\) be the regular submodel with score \(g\) furnished by Assumption~\ref{assumption:coord_smooth_all}. By QMD, we know that \(H(P_{t, s'}, P_0) \to 0\) as \(t \to 0\), and \(H(P_{t, g}, P_0) \to 0\) as \(t \to 0\). Let \(\delta > 0\) be as in Assumption~\ref{assumption:hellinger}. There exists \(t^* > 0\) such that for all \(|t| < t^*\),
\[H(P_{t, s'}, P_0) \le \delta \quad \text{and} \quad H(P_{t, g}, P_0) \le \delta. \]
By Lemma~\ref{lem:hellinger_gap} (Appendix~\ref{subsec:hellinger_gap}), which bounds the Hellinger distance between two regular submodels in terms of the $L_2(P_0)$ distance between their scores,
\[\limsup_{t \to 0}\frac{H(P_{t, s'}, P_{t, g})}{|t|} \le \frac{1}{2\sqrt{2}}\|s' - g\|_{L_2(P_0)} \le \frac{\epsilon}{2\sqrt{2}}.\]
Then it follows that
\[\limsup_{t \to 0}\frac{|\beta(P_{t, s'}) - \beta(P_{t, g})|}{|t|} \le c \cdot \limsup_{t \to 0} \frac{H(P_{t, s'}, P_{t, g})}{|t|} \le \frac{c\epsilon}{2 \sqrt{2}},\]
where the first inequality holds by Assumption~\ref{assumption:hellinger}, since both $P_{t, s'}$ and $P_{t, g}$ lie within Hellinger distance $\delta$ of $P_0$ for $|t| < t^*$. Next, for any \(t \neq 0\) with \(|t| < t^*\), we write
\[\left|\frac{\beta(P_{t, s'}) - \beta_0}{t} - \bE_{0}[\varphi s']\right| \le \underbrace{\frac{|\beta(P_{t, s'}) - \beta(P_{t, g})|}{|t|}}_{(\mathrm{I})} + \underbrace{\left|\frac{\beta(P_{t, g}) - \beta_0}{t} - \bE_{0}[\varphi g]\right|}_{(\mathrm{II})} + \underbrace{|\bE_0[\varphi(s' - g)]|}_{(\mathrm{III})}.\]
For term (I), we know \(\limsup_{t \to 0}(\mathrm{I}) \le \frac{c\epsilon}{2 \sqrt{2}}.\) For term (II), the score \(g\) lies in \(S\), so we know from above that \(\lim_{t \to 0} (\mathrm{II}) = 0\). For term (III), by Cauchy-Schwarz, we have
\[|\bE_0[\varphi(s' - g)]| \le \|\varphi\|_{L_2(P_0)} \cdot\|s' - g\|_{L_2(P_0)} \le \|\varphi\|_{L_2(P_0)} \cdot \epsilon. \]
Combining the three terms, we arrive at
\[\limsup_{t \to 0} \left|\frac{\beta(P_{t, s'}) - \beta_0}{t} - \bE_{0}[\varphi s']\right| \le \epsilon \left(\frac{c}{2\sqrt{2}} + \|\varphi\|_{L_2(P_0)} \right).\]
Since \(\epsilon\) was arbitrary, the left-hand side evaluates to zero. Therefore,
\[\frac{d}{dt}\beta(P_{t, s'})\bigg|_{t = 0} = \bE_0[\varphi(Z)s'(Z)].\] 
Since \(s' \in \cS\) was arbitrary, Definition~\ref{def:pathwise} is satisfied and \(\beta\) is pathwise differentiable at \(P_0\) with influence function \(\varphi\).
\end{proof}

\begin{remark}[Hellinger Lipschitz]
The extension from bounded scores to all scores in the proof of Theorem~\ref{thm:forward} adapts an argument from \citet{luedtke_one-step_2024}, who use a Hellinger Lipschitz condition to establish pathwise differentiability of Hilbert-valued parameters from a score-dense class of submodels (their Lemma~2). The first part of the proof, which establishes the derivative representation on the dense class from Neyman orthogonality, is specific to the present setting.
\end{remark}

\begin{remark}[Neyman orthogonality and efficiency]
Theorem~\ref{thm:forward} establishes that a Neyman orthogonal moment induces an influence function \(\varphi\). However, the induced influence function need not always be the efficient influence function. As noted in Remark~\ref{remark:uniqueness},
\(
\Lambda^\perp \cap \cT = \mathrm{span}(\varphi_{\mathrm{eff}}),
\)
where \(\varphi_{\mathrm{eff}}\) is the unique element lying in \(\Lambda^\perp \cap \cT\), and every influence function admits a decomposition
\(
\varphi = \varphi_{\mathrm{eff}} + h\) with 
\(h \in \cT^\perp.
\)
In the nonparametric model, \(\cT = L_2^0(P_0)\) by Corollary~\ref{cor:saturation}, so \(\cT^\perp = \{0\}\) and the influence function is unique. In semiparametric models, however, \(\cT^\perp\) is generally nontrivial, so a Neyman orthogonal moment may induce an influence function with \(h \neq 0\), yielding a valid but inefficient estimating equation. In conclusion, efficiency does not automatically follow within the class of orthogonal moments, and is attained when the induced influence function equals \(\varphi_{\mathrm{eff}}\), equivalently when its component lying in \(\cT^\perp\) vanishes. We refer the reader to \citet[Chapter~3.2]{bickel_efficient_1998}, \citet[Chapter~1.4]{van_der_laan_unified_2003}, \citet[Chapter~4.4]{tsiatis_semiparametric_2006}, 
\citet[Section~3.3]{kennedy_semiparametric_2016},
\citet[Section~2.2]{chen_overidentification_2018} for a more in-depth treatment.
\end{remark}

\subsection{Pathwise Differentiability Implies Neyman Orthogonality}
\label{subsec:reverse}

We now prove the converse. If \(m\) is a correctly specified
estimating function whose value at the truth is an influence
function, then \(m\) is Neyman orthogonal and its sensitivity to perturbations of \(\beta\) is pinned at unit rate by the pathwise derivative. Unlike the forward direction, this
requires the local product structure of
Assumption~\ref{assumption:product-structure} in order to
specialize Equation~\ref{eq:intermediate_identity} from the proof of
Theorem~\ref{thm:forward} to each coordinate axis independently.

\begin{assumption}[Pathwise differentiability and influence function representation]
\label{assumption:pathwise}
The functional \(\beta\) is pathwise differentiable at \(P_0\)
(Definition~\ref{def:pathwise}) with influence function
\(\varphi(Z) \equiv m(Z; \beta_0, \eta_0)\). 
\end{assumption}

\begin{assumption}[Local product structure]
\label{assumption:product_reverse}
Assumption~\ref{assumption:product-structure} holds. We denote
the score of the \(\beta\)-coordinate submodel by \(s_\beta\)
and the score of the \(\eta\)-coordinate submodel in direction
\(h\) by \(s_h\).
\end{assumption}

\begin{assumption}[Regularity along coordinate submodels]
\label{assumption:regularity_coordinate}
For each coordinate submodel \(t \mapsto P_{t,s}\) from
Assumption~\ref{assumption:product_reverse}, the score \(s\) is
bounded and the function
\(f_{t,s}(Z) := m(Z; \beta_{t,s}, \eta_{t,s})\) satisfies the
conditions of Lemma~\ref{lem:differentiation_varying_f}.
\end{assumption}

\begin{theorem}[Pathwise differentiability implies Neyman
orthogonality]
\label{thm:reverse}
Under
Assumptions~\ref{assumption:correct_spec},
\ref{assumption:frechet_all},
and~\ref{assumption:pathwise}--\ref{assumption:regularity_coordinate},
the estimating function \(m\) satisfies:
\begin{enumerate}
    \item \textbf{Neyman orthogonality.} For all
    \(h \in \dot{\cH}\),
    \[\frac{\partial}{\partial \eta}\bE_0[m(Z; \beta_0, \eta)]
    \bigg|_{\eta = \eta_0}[h] = 0.\]
    \item \textbf{\(-1\) normalization.}
    \[G := \bE_0[\partial_\beta m(Z; \beta_0, \eta_0)] = -1.\]
\end{enumerate}
\end{theorem}

\begin{proof}
The coordinate submodels furnished by
Assumption~\ref{assumption:product_reverse} are regular
submodels through \(P_0\), and their induced coordinate paths
are differentiable at \(t = 0\) by construction, so they
satisfy Assumption~\ref{assumption:coordinate_smoothness}.
Together with correct specification
(Assumption~\ref{assumption:correct_spec}), Fr\'echet
differentiability (Assumption~\ref{assumption:frechet_all}),
and the regularity conditions of
Assumption~\ref{assumption:regularity_coordinate}, the
derivation in the proof of Theorem~\ref{thm:forward} leading
to Equation~\ref{eq:intermediate_identity} applies to each coordinate
submodel \(t \mapsto P_{t,s}\) with score \(s\), 
\begin{equation}
\label{eq:intermediate_reverse}
0 = \bE_0[m(Z; \beta_0, \eta_0)\, s(Z)] + G\,\dot\beta_{0,s}
+ \frac{\partial}{\partial \eta}\bE_0[m(Z; \beta_0, \eta)]
\bigg|_{\eta = \eta_0}\bigl[\dot\eta_{0,s}\bigr].
\end{equation}
By the influence function representation
(Assumption~\ref{assumption:pathwise}), the first term equals
\(\dot\beta_{0,s} = \frac{d}{dt}\beta(P_{t,s})|_{t=0}\),
so~\eqref{eq:intermediate_reverse} becomes
\begin{equation}
\label{eq:master_identity}
(1 + G)\,\dot\beta_{0,s} +
\frac{\partial}{\partial \eta}
\bE_0[m(Z; \beta_0, \eta)]\bigg|_{\eta = \eta_0}
\bigl[\dot\eta_{0,s}\bigr] = 0.
\end{equation}
We now specialize~\eqref{eq:master_identity} to each coordinate
submodel.

\textbf{Part 1.}
Fix \(h \in \dot{\cH}\) and take \(t \mapsto P_{t, s_h}\) to
be the \(\eta\)-coordinate submodel from
Assumption~\ref{assumption:product_reverse}. By
Assumption~\ref{assumption:product_reverse},
\(\dot\beta_{0, s_h} = 0\) and
\(\dot\eta_{0, s_h} = h\).
Substituting into~\eqref{eq:master_identity},
\[\frac{\partial}{\partial \eta}\bE_0[m(Z; \beta_0, \eta)]
\bigg|_{\eta = \eta_0}[h] = 0.\]
Since \(h \in \dot{\cH}\) was arbitrary, Neyman orthogonality
holds.

\textbf{Part 2.}
Take \(t \mapsto P_{t, s_\beta}\) to be the \(\beta\)-coordinate submodel from Assumption~\ref{assumption:product_reverse}. By Assumption~\ref{assumption:product_reverse}, \(\dot\beta_{0, s_\beta} = 1\) and \(\dot\eta_{0, s_\beta} = 0\). Since the G\^ateaux derivative
in~\eqref{eq:master_identity} is evaluated at direction \(\dot\eta_{0, s_\beta} = 0\), the numerator of the defining difference quotient vanishes identically, which leaves
\[(1 + G) \cdot 1 = 0,\]
hence \(G = -1\).
\end{proof}

\begin{remark}[Structural comparison with the forward direction]
\label{remark:structural_comparison}
The forward and reverse directions share the same intermediate identity~\eqref{eq:intermediate_identity}, but differ in what is known and what is derived. In the forward direction, Neyman orthogonality eliminates the nuisance term, and the resulting
inner-product representation \(\dot\beta_{0,s} = \bE_0[\varphi\, s]\) for every score
\(s \in \cS\) yields pathwise differentiability. In the reverse direction, the influence function representation converts the first term into \(\dot\beta_{0,s}\), and product structure
allows one to specialize the resulting identity~\eqref{eq:master_identity} to each coordinate axis
independently, yielding Neyman orthogonality and \(G = -1\).

The two directions also place different requirements on the submodels. In Theorem~\ref{thm:forward},
the coordinate path \((\beta_{t,s}, \eta_{t,s})\) arises from evaluating the functionals \(\beta\) and \(\eta\) along regular submodels from the dense class in Assumption~\ref{assumption:coord_smooth_all}. In Theorem~\ref{thm:reverse}, we must construct submodels with prescribed first-order coordinate behavior, one along which
\(\dot\beta_{0,s} = 1, \dot\eta_{0,s} = 0\) and one with \(\dot\beta_{0,s} = 0, \dot\eta_{0,s} = h\).
\end{remark}

\begin{remark}[The \(-1\) normalization]
\label{remark:normalization}
The \(-1\) normalization follows naturally as a structural consequence of pathwise differentiability and the
coordinate geometry of the model. Along the
\(\beta\)-coordinate submodel
\(t \mapsto P_{t, s_\beta}\), the parameter \(\beta\) increases
at unit rate by construction, and the influence function
representation gives
\(\bE_0[m(Z; \beta_0, \eta_0)\, s_\beta(Z)] = 1\), and \eqref{eq:master_identity} forces
\(1 + G = 0\). A first-order Taylor expansion gives
\[\bE_0[m(Z; \beta, \eta_0)] \approx
\bE_0[m(Z; \beta_0, \eta_0)] + (-1) \cdot (\beta - \beta_0)
= -(\beta - \beta_0),\]
so that \(\bE_0[m(Z; \beta, \eta_0)] = 0\) has the unique
local solution \(\beta = \beta_0\), as desired. This also sheds
light on a familiar pattern in semiparametric inference where
many influence functions take the form
\(\varphi(Z) = (\text{data-dependent term}) - \beta_0\). The
normalization requires \(\beta\) to enter the expected
estimating function with first-order sensitivity exactly
\(-1\), which is realized by subtracting off \(\beta\).
\end{remark}

\begin{remark}[When local product structure fails]
\label{remark:failure_local_product}
When \(\beta = g(\eta)\) for a Fr\'echet differentiable functional \(g: \cH \to \RR\), the target parameter carries no degrees of freedom beyond those already encoded in the nuisance. By the chain rule, any regular submodel with \(\dot \eta_{0, s} = 0\) satisfies \(\dot \beta_{0, s} = Dg(\eta_0)[0] = 0\), so no \(\beta\)-coordinate submodel where \(\dot \beta_{0, s} = 1\) can exist. Similarly, any submodel with \(\dot \eta_{0, s} = h\) satisfies \(\dot \beta_{0, s} = Dg(\eta_0)[h]\), which is generally nonzero, so no \(\eta\)-coordinate submodel with \(\dot \beta_{0, s} = 0\) can exist. Thus, Assumption~\ref{assumption:product-structure} fails and Theorem~\ref{thm:reverse} does not apply. The following proposition shows that a clean characterization is nevertheless available.
\end{remark}

\begin{proposition}[Neyman orthogonality without local product structure]
\label{prop:neyman_wo_lps}
Suppose that \(\beta = g(\eta)\) for a Fr\'echet differentiable function \(g : \cH \to \RR\) whose derivative does not vanish on the admissible nuisance directions, i.e., \(Dg(\eta_0)[h^*] \neq 0\) for some \(h^* \in \dot{\cH}\). Suppose also that the estimating function \(m\) is correctly specified at \((\beta_0, \eta_0)\) (Assumption~\ref{assumption:correct_spec}), is Fr\'echet differentiable (Assumption~\ref{assumption:frechet_all}), and satisfies the influence function representation (Assumption~\ref{assumption:pathwise}). Finally, suppose that for every \(h \in \dot \cH\), there exists a regular submodel \(t \mapsto P_{t, s}\) through \(P_0\) with bounded score \(s\) such that \(t \mapsto \eta(P_{t, s})\) is differentiable at \(0\) in \(\cV\) with \(\dot \eta_{0, s} = h\) and \(f_{t, s}(Z) := m(Z; \beta_{t, s}, \eta_{t, s})\) satisfies the conditions of Lemma~\ref{lem:differentiation_varying_f}. Then
\[m \text{ is Neyman orthogonal } \quad \Longleftrightarrow \quad G := \bE_0[\partial_\beta m(Z; \beta_0, \eta_0)] = -1. \]
\end{proposition}

\begin{proof}
    Fix \(h \in \dot \cH\) and let \(t \mapsto P_{t, s}\) be a regular submodel with bounded score \(s\) and \(\dot \eta_{0, s} = h\), which exists by assumption. Since \(\beta = g(\eta)\) and \(g\) is Fr\'echet differentiable at \(\eta_0\), the chain rule gives
    \[\dot \beta_{0, s} = \frac{d}{dt}g(\eta_{t, s}) \bigg|_{t = 0} = Dg(\eta_0)\bigl[\dot \eta_{0, s}\bigr] = Dg(\eta_0)[h],\]
    hence the induced coordinate path \(t \mapsto (\beta_{t, s}, \eta_{t, s})\) is differentiable at \(t = 0\). Similar to the derivation in the proof of Theorem~\ref{thm:reverse}, we can show that Equation~\ref{eq:intermediate_identity} applies to the chosen submodel. Therefore, it follows that
    \[0 = \bE_0[m(Z; \beta_0, \eta_0)\, s(Z)] + G\,\dot\beta_{0,s} + \frac{\partial}{\partial \eta}\bE_0[m(Z; \beta_0, \eta)] \bigg|_{\eta = \eta_0}\bigl[\dot\eta_{0,s}\bigr].\]
    Since \(m(Z; \beta_0, \eta_0) \equiv \varphi(Z)\) is an influence function, it follows \(\bE_0[m(Z; \beta_0, \eta_0) \, s(Z)] = \dot \beta_{0, s}\). Substituting into the above and using \(\dot \beta_{0, s} = Dg(\eta_0)[h]\) and \(\dot \eta_{0, s} = h\), we obtain
    \begin{equation}
    \label{eq:master_no_lps}
        (1 + G)Dg(\eta_0)[h] + \frac{\partial}{\partial \eta} \bE_0[m(Z; \beta_0, \eta)]\bigg|_{\eta = \eta_0}[h] = 0.
    \end{equation}
    Since \(h \in \dot \cH\) was arbitrary, \eqref{eq:master_no_lps} holds for all \(h \in \dot \cH\). 
    
    \((\Longleftarrow)\) If \(G = -1\), then \(1 + G = 0\) and \eqref{eq:master_no_lps} gives %
    \[\frac{\partial}{\partial \eta}\bE_0[m(Z; \beta_0, \eta)]\bigg|_{\eta = \eta_0}[h] = 0 \qquad \forall h \in \dot \cH,\]
    which is Neyman orthogonality.

    \((\Longrightarrow)\) If \(m\) is Neyman orthogonal, then the second term in \eqref{eq:master_no_lps} vanishes for all \(h \in \dot \cH\), which leaves
    \[(1 + G) Dg(\eta_0)[h] = 0 \qquad \forall h \in \dot \cH.\]
    By hypothesis, there exists \(h^* \in \dot \cH\) with \(Dg(\eta_0)[h^*] \neq 0\). Evaluating at \(h = h^*\) gives \((1 + G) = 0\), i.e., \(G = -1\).
\end{proof}

To illustrate that the conditions of Theorems~\ref{thm:forward} and~\ref{thm:reverse} can be verified in standard settings, we work through two examples in detail. In Appendix~\ref{appendix:ate}, we consider estimating the average treatment effect in the nonparametric model, constructing the respective coordinate submodels explicitly and checking each assumption. In Appendix~\ref{appendix:plm}, we turn to estimating the slope in the partially linear model, where the semiparametric structure restricts the tangent space and generic linear tilts leave the model. The forward direction is verified for the classical residual-on-residual moment, which is Neyman orthogonal but semiparametrically inefficient, and the reverse direction starts from the efficient influence function and recovers Neyman orthogonality via explicit coordinate submodels witnessing local product structure.

Finally, in Appendix~\ref{appendix:squared_density}, we consider the expected density \(\beta(P) = \int p^2 \, d\nu\), where the target is a known functional of the nuisance \(\eta = p\) and local product structure fails to hold. Theorem~\ref{thm:forward} applies without modification where the estimating function \(m(z; \beta, p) = 2p(z) - \int p^2 \, d\nu - \beta\) is Neyman orthogonal and induces pathwise differentiability with influence function \(\varphi(z) = 2(p_0(z) - \beta_0)\). For the reverse direction, Theorem~\ref{thm:reverse} cannot be applied since no coordinate submodels exist, but Proposition~\ref{prop:neyman_wo_lps} recovers their equivalence.

\section{Discussion}
\label{sec:discussion}

In this paper, we have established a precise equivalence between Neyman orthogonality and pathwise differentiability in nonparametric models, building on the foundational semiparametric theory of \citet{bickel_efficient_1998, van_der_laan_unified_2003, tsiatis_semiparametric_2006}, and connecting it to the modern double/debiased machine learning framework of \citet{chernozhukov_doubledebiased_2018}. Our forward theorem shows that under mild conditions, Neyman orthogonality implies pathwise differentiability, and our converse shows that the reverse implication also holds, but requires the additional geometric condition of local product structure.

Several directions remain open for further investigation. Most importantly, the regularity conditions we impose, notably the existence of coordinate submodels witnessing local product structure, can be nontrivial to verify in complex semiparametric problems, such as those involving constrained nuisance spaces or functionals defined through implicit equations. This being said, the conditions we require are mild, amounting to smoothness of the estimating function and the ability to perturb the target and nuisance parameters independently, and we expect the equivalence to hold broadly in the semiparametric settings most commonly encountered in practice. Relaxing these conditions, extending the equivalence to settings with non-smooth functionals, and developing systematic tools for constructing coordinate submodels in applied problems would be natural next steps.

\section*{Acknowledgements}
Y.C.\ thanks Yanlin Qu, Hongjian Wang, Heyuan Yao, and Weihan Zhang for helpful discussions. The authors also thank Vasilis Syrgkanis for raising the question of orthogonal but inefficient moments in partially linear models, and Xiaohong Chen for pointing out the related results in \citet{chen_overidentification_2018}.

\bibliographystyle{unsrtnat}
\bibliography{references}

\appendix

\section{Proofs of Lemmas}
\label{appendix:proofs of auxiliary results}

\subsection{Proof of Lemma~\ref{lem:linear_tilt_submodel}}
\label{appendix:proof_linear_tilt_submodel}

We first verify that \(p_t\) is a density. Since \(|g| \le M\) \(P_0\)-a.s., for \(|t| < 1/M\) it follows that \(1 + tg(z) > 1 - |t|M > 0\) \(P_0\)-a.s., so \(p_t = p_0(1 + tg) \ge 0\) \(\nu\)-a.s. Also 
\[\int p_t \, d\nu = \int p_0(1 + tg)\,d\nu = 1 + t\int g \,dP_0 = 1 + t\bE_0[g] = 1.\]
We next show that it satisfies the QMD expansion. Write \(\sqrt{p_t} = \sqrt{p_0}\sqrt{1 + tg}\). Define
\[r(u) := \sqrt{1 + u} - 1 - \frac{1}{2}u, \quad u \in (-1, 1).\]
Then \(r(0) = r'(0) = 0\). Since \(\sqrt{1 + u}\) has bounded second derivative on \([-1/2, 1/2]\), there exists \(C < \infty\) such that \(|r(u)| \le Cu^2\) for \(|u| \le 1/2.\) For \(|t| \le 1/(2M)\) we have \(|tg| \le 1/2\), hence
\[\sqrt{1 + tg} = 1 + \frac{t}{2}g + r(tg).\]
Therefore
\[\frac{\sqrt{p_t} - \sqrt{p_0}}{t} - \frac{1}{2}g \sqrt{p_0} = \sqrt{p_0} \cdot \frac{r(tg)}{t}.\]
Using \(|r(tg)| \le Ct^2 g^2\), we have \(\left|\frac{r(tg)}{t}\right| \le C|t|g^2 \le C|t|M^2\). Hence
\[\int \left(\frac{\sqrt{p_t} - \sqrt{p_0}}{t} - \frac{1}{2}g\sqrt{p_0}\right)^2 d\nu \le \int p_0 \bigl(C|t|M^2\bigr)^2 \, d\nu = C^2 t^2 M^4 \to 0 \quad \text{as \(t \to 0\)}\]
and the path is QMD with score \(s \equiv g\).

\subsection{Proof of Lemma~\ref{lem:differentiation_fixed_f}}
\label{appendix:proof_differentiation_fixed_f}

Write \(P_t \equiv P_{t, s}\) and \(p_t = dP_t / d\nu\) throughout. We have \(\bE_{P_t}[f] = \int f p_t \, d\nu\). Then
\[\bE_{P_t}[f] - \bE_0[f] = \int f(p_t - p_0) \, d\nu = \int f \bigl(\sqrt{p_t} - \sqrt{p_0}\bigr) \bigl(\sqrt{p_t} + \sqrt{p_0}\bigr) \, d\nu.\]
Dividing by \(t\),
\[\frac{\bE_{P_t}[f] - \bE_0[f]}{t} = \int f\left(\frac{\sqrt{p_t} - \sqrt{p_0}}{t} \right)\bigl(\sqrt{p_t} + \sqrt{p_0}\bigr) \,d\nu.\]
Let 
\[\Delta_t := \frac{\sqrt{p_t} - \sqrt{p_0}}{t} - \frac{1}{2}s\sqrt{p_0}.\]
By QMD (Definition~\ref{def:regular_qmd_submodel}), it follows \(\|\Delta_t\|_{L_2(\nu)} \to 0\). Decompose
\[\frac{\bE_{P_t}[f] - \bE_0[f]}{t} = I_{t, 1} + I_{t, 2}\]
where 
\[I_{t, 1} := \int f \left(\frac{1}{2} s\sqrt{p_0}\right)\bigl(\sqrt{p_t} + \sqrt{p_0}\bigr) \,d\nu, \quad I_{t, 2} := \int f\Delta_t \bigl(\sqrt{p_t} + \sqrt{p_0}\bigr) \,d\nu.\]
\textbf{Term \(I_{t, 1}\)}. By QMD (Definition~\ref{def:regular_qmd_submodel}) and the triangle inequality, \(\sqrt{p_t} \to \sqrt{p_0}\) in \(L_2(\nu)\). To see this,
there exists a function \(g = \frac{1}{2}s\sqrt{p_0}\) such that
\[\left\|\frac{\sqrt{p_t} - \sqrt{p_0}}{t} - g\right\|_{L_2(\nu)} \to 0,\]
which is equivalent to saying that for any \(\epsilon > 0\), there exists \(\delta > 0\) such that for \(|t| < \delta\),
\[\left\|\frac{\sqrt{p_t} - \sqrt{p_0}}{t} - g\right\|_{L_2(\nu)}< \epsilon.\]
By the triangle inequality for \(|t| < \delta\), 
\[\left\|\frac{\sqrt{p_t} - \sqrt{p_0}}{t}\right\|_{L_2(\nu)} \le \left\|\frac{\sqrt{p_t} - \sqrt{p_0}}{t} - g\right\|_{L_2(\nu)} + \|g\|_{L_2(\nu)} < \epsilon + \|g\|_{L_2(\nu)}\]
where the right-hand side does not depend on \(t\). Multiplying both sides by \(|t|\) and taking the limit as \(t \to 0\) yields the result. 

Thus, \(\sqrt{p_t} + \sqrt{p_0} \to 2 \sqrt{p_0}\) in \(L_2(\nu)\). Since \(fs\sqrt{p_0} \in L_2(\nu)\) as \[\int f^2 s^2 p_0 \, d\nu = \bE_0\bigl[f^2 s^2\bigr] < M^2 \bE_0\bigl[f^2\bigr] < \infty\]
where \(M := \|s\|_\infty\), we have
\[I_{t, 1} \to \int f\left(\frac{1}{2} s\sqrt{p_0}\right)2 \sqrt{p_0} \,d\nu = \int fs p_0 \,d\nu = \bE_0[fs].\]
\textbf{Term \(I_{t, 2}\).}
By Cauchy-Schwarz, 
\[|I_{t, 2}| \le \left\|f\bigl(\sqrt{p_t} + \sqrt{p_0}\bigr)\right\|_{L_2(\nu)} \cdot \|\Delta_t\|_{L_2(\nu)}.\]
We already have that \(\|\Delta_t\|_2 \to 0\). It remains to show \(\bigl\|f\bigl(\sqrt{p_t} + \sqrt{p_0}\bigr)\bigr\|_{L_2(\nu)}\) is bounded for small \(t\). Write
\[\left\|f\bigl(\sqrt{p_t} + \sqrt{p_0}\bigr)\right\|_{L_2(\nu)}^2 = \int f^2 \bigl(\sqrt{p_t} + \sqrt{p_0}\bigr)^2 \,d\nu \le 2 \int f^2(p_t + p_0)\,d\nu = 2 \left[\bE_{P_t}\bigl[f^2\bigr] + \bE_0\bigl[f^2\bigr]\right]\]
Let \(B:= \sup_{z \in \cZ}|f(z)| < \infty\). Under QMD, we know \(P_t \to P_0\) in Hellinger distance, and
\[\left|\bE_{P_t}\bigl[f^2\bigr] - \bE_0\bigl[f^2\bigr]\right| = \left| \int f^2(p_t - p_0) \,d\nu\right| \le B^2 \int |p_t - p_0|\,d\nu \le B^2 \cdot 2 \bigl\|\sqrt{p_t} - \sqrt{p_0}\bigr\|_{L_2(\nu)} \to 0\]
where the first inequality follows by the boundedness assumption of \(f\), and the second follows from \(\int |p_t - p_0|\,d\nu = 2\mathrm{TV}(P_t, P_0) \le 2\sqrt{2}\,H(P_t, P_0)\). Hence \(\bE_{P_t}\bigl[f^2\bigr] \to \bE_0\bigl[f^2\bigr]\) as \(t \to 0\) and \(I_{t, 2} \to 0\).

Combining the results above yields the desired derivative. 

\subsection{Proof of Lemma~\ref{lem:differentiation_varying_f}}
\label{appendix:proof_differentiation_varying}

Write \(P_t \equiv P_{t, s}\) and \(p_t = dP_t / d\nu\) throughout. We have
\[\bE_{P_t}[f_t] - \bE_0[f_0] = \left[\bE_{P_t}[f_0] - \bE_0[f_0]\right] + \bE_{P_t}[f_t - f_0].\]
Dividing by \(t\),
\[\frac{\bE_{P_t}[f_t] - \bE_0[f_0]}{t} = \frac{\bE_{P_t}[f_0] - \bE_0[f_0]}{t} + \bE_{P_t} \left[\frac{f_t - f_0}{t}\right].\]
By Lemma~\ref{lem:differentiation_fixed_f}, the first term converges to \(\bE_0[f_0 s]\). 

For the second term, let \(g_t := (f_t - f_0)/t\). By Assumption (1), \(g_t \to \dot f_0\) in \(L_2(P_0)\) so \(\bE_0[g_t] \to \bE_0\bigl[\dot f_0\bigr]\). It remains to show \(\bE_{P_t}[g_t] - \bE_0[g_t] \to 0\). Write
\[\bE_{P_t}[g_t] - \bE_0[g_t] = \int g_t(p_t - p_0) \, d\nu = \int g_t\bigl(\sqrt{p_t} - \sqrt{p_0}\bigr)\bigl(\sqrt{p_t} + \sqrt{p_0}\bigr) \, d\nu.\]
By Cauchy-Schwarz, 
\[|\bE_{P_t}[g_t] - \bE_0[g_t]| \le \left\|g_t \bigl(\sqrt{p_t} + \sqrt{p_0}\bigr)\right\|_{L_2(\nu)} \cdot \bigl\|\sqrt{p_t} - \sqrt{p_0}\bigr\|_{L_2(\nu)}.\]
By QMD (Definition~\ref{def:regular_qmd_submodel}), \(\bigl\|\sqrt{p_t} - \sqrt{p_0}\bigr\|_{L_2(\nu)} \to 0\). It suffices to show that \(\left\|g_t\bigl(\sqrt{p_t} + \sqrt{p_0}\bigr)\right\|_{L_2(\nu)}\) is bounded for small \(t\). By \((a + b)^2 \le 2\bigl(a^2 + b^2\bigr)\),
\[\left\|g_t\bigl(\sqrt{p_t} + \sqrt{p_0}\bigr)\right\|_{L_2(\nu)}^2 \le 2 \left[\bE_{P_t}\bigl[g_t^2\bigr] + \bE_0\bigl[g_t^2\bigr]\right]\]
By Assumption (2), \(\bE_{P_t}\bigl[g_t^2\bigr]\) is uniformly bounded for small \(t\). \(\bE_0\bigl[g_t^2\bigr]\) is also bounded since \(g_t \to \dot f_0\) in \(L_2(P_0)\). Hence the right-hand side is bounded and \(\bE_{P_t}[g_t] - \bE_0[g_t] \to 0\). Therefore \(\bE_{P_t}[g_t] \to \bE_0\bigl[\dot f_0\bigr]\). Collecting both terms yields the desired identity. 

\subsection{Proof of Lemma~\ref{lem:l2_chain_rule}}
\label{appendix:proof_chain_rule}

By Fr\'echet differentiability (Assumption~\ref{assumption:frechet_m}),
\[m(\cdot; \beta_{t, s}, \eta_{t, s}) - m(\cdot; \beta_0, \eta_0) = D_\beta m_0(\beta_{t, s} - \beta_0) + D_\eta m_0(\eta_{t, s} - \eta_0) + r_{t, s},\]
where
\(\|r_{t, s}\|_{L_2(P_0)} = o(|\beta_{t, s} - \beta_0| + \|\eta_{t, s} - \eta_0\|_{\cV}).\)
Dividing by \(t\) and subtracting \(\dot f_{0, s}\), 
\[\frac{f_{t, s} - f_0}{t} - \dot f_{0, s} = D_\beta m_0 \left(\frac{\beta_{t, s} - \beta_0}{t} - \dot \beta_{0, s}\right) + D_\eta m_0\left(\frac{\eta_{t, s} - \eta_0}{t} - \dot \eta_{0, s}\right) + \frac{r_{t, s}}{t}.\]
Take \(L_2(P_0)\) norms. By boundedness of \(D_\beta m_0, D_\eta m_0\), there exist constants \(C_\beta, C_\eta\) such that 
\[\|\cdot\|_{L_2} \le C_\beta \left|\frac{\beta_{t, s} - \beta_0}{t} - \dot \beta_{0, s}\right| + C_\eta \left\|\frac{\eta_ {t, s} - \eta_0}{t} - \dot \eta_{0, s}\right\|_{\cV} + \left\|\frac{r_{t, s}}{t}\right\|_{L_2}.\]
The first two terms \(\to 0\) by Assumption~\ref{assumption:coordinate_smoothness}. For the remainder, Assumption~\ref{assumption:coordinate_smoothness} implies \(|\beta_{t, s} - \beta_0| + \|\eta_{t, s} - \eta_0\|_{\cV} = O(|t|)\), so \(\|r_{t, s}/t\|_{L_2} = o(1)\), which proves the claim.

\subsection{Hellinger Gap between Regular Submodels}
\label{subsec:hellinger_gap}

\begin{lemma}
\label{lem:hellinger_gap}

Let \(t \mapsto P_{t, s}\) be a regular submodel with score \(s\) and \(t \mapsto P_{t, g}\) be a regular submodel with score \(g\). Then 
\[\limsup_{t \to 0} \frac{H(P_{t, s}, P_{t, g})}{|t|} \le \frac{1}{2 \sqrt{2}}\|s - g\|_{L_2(P_0)}.\]
\end{lemma}

\begin{proof}
By QMD, we have 
\begin{align*}
    \sqrt{p_{t, s}} &= \sqrt{p_0} \left(1 + \frac{t}{2}s\right) + r_t, \quad \|r_t\|_{L_2(\nu)} = o(|t|),\\
    \sqrt{p_{t, g}} &= \sqrt{p_0} \left(1 + \frac{t}{2}g\right) + \tilde r_t, \quad \bigl\|\tilde r_t\bigr\|_{L_2(\nu)} = o(|t|).
\end{align*}    
Subtracting, 
\[\sqrt{p_{t, s}} - \sqrt{p_{t, g}} = \frac{t}{2}(s - g)\sqrt{p_0} + \bigl(r_t - \tilde r_t\bigr).\]
Taking \(L_2(\nu)\) norms and using the triangle inequality, 
\[\bigl\|\sqrt{p_{t, s}} - \sqrt{p_{t, g}}\bigr\|_{L_2(\nu)} \le \frac{|t|}{2}\left\|(s - g)\sqrt{p_0}\right\|_{L_2(\nu)} + \|r_t\|_{L_2(\nu)} + \bigl\|\tilde r_t\bigr\|_{L_2(\nu)}.\]
Now \(\left\|(s - g)\sqrt{p_0}\right\|^2_{L_2(\nu)} = \int (s - g)^2 p_0 \, d\nu = \bE_0\left[(s - g)^2\right] = \|s - g\|^2_{L_2(P_0)}.\) Dividing by \(|t|\),
\[\frac{\bigl\|\sqrt{p_{t, s}} - \sqrt{p_{t, g}}\bigr\|_{L_2(\nu)}}{|t|} \le \frac{1}{2}\|s - g\|_{L_2(P_0)} + \frac{\|r_t\|_{L_2(\nu)} + \bigl\|\tilde r_t\bigr\|_{L_2(\nu)}}{|t|}.\]
Since \(\|r_t\|_{L_2(\nu)} = o(|t|)\) and \(\bigl\|\tilde r_t\bigr\|_{L_2(\nu)} = o(|t|)\), the second term vanishes as \(t \to 0\). By definition of \(H\), 
\[\limsup_{t \to 0} \frac{H(P_{t, s}, P_{t, g})}{|t|} \le \frac{1}{2\sqrt{2}}\|s - g\|_{L_2(P_0)}.\]
\end{proof}

\section{Local Variation Independence and Local Product Structure}
\label{appendix:product_structure}

Assumption~\ref{assumption:product-structure} requires that, for each coordinate direction, there exists a regular submodel through \(P_0\) along which the induced coordinate path moves one of \(\beta\) or \(\eta\) to first order while holding the other fixed. A natural question is how this relates to the classical notion of local variation independence, which asks that the attainable parameter set contains a product neighborhood of \((\beta_0, \eta_0)\). Local variation independence guarantees that independently varied parameter values exist, but is purely set-theoretic and does not ensure that they are connected by submodels regular enough to differentiate along. In this appendix, we formalize the distinction between these two conditions and examine the role of product structure in the classical results of \citet{van_der_laan_unified_2003}.

\subsection{Local Variation Independence}
\label{subsec:local_variation_independence}

As mentioned in Definition~\ref{def:nuisance}, we are primarily concerned with regular submodels along which \(\beta\) has derivative zero at the truth, while \(\eta\) is free to vary. The obvious question is whether such paths can always be constructed, i.e., whether one can perturb \(\eta\) while holding \(\beta\) fixed. If the chosen nuisance functional \(\eta\) already determines \(\beta\), for instance, if \(\beta = g(\eta)\) for some known map \(g\), then varying \(\eta\) necessarily changes \(\beta\), and the two functionals cannot be perturbed independently. 

\begin{definition}[Local Variation Independence]
\label{def:local_variation_independence}
We say that \(\beta\) and \(\eta\) are locally variation independent at \(P_0\) if there exist neighborhoods \(U \ni \beta_0\) and \(V \ni \eta_0\) such that 
\[U \times V \subseteq \Theta := \{(\beta(P), \eta(P)): P \in \cP\},\]
that is, the attainable parameter set \(\Theta\) contains a product neighborhood of \(\beta_0, \eta_0\).
\end{definition}

In words, near \((\beta_0, \eta_0)\), there is a full interval of \(\beta\)-values and a full neighborhood of \(\eta\)-values such that every combination of the two is realized by some \(P \in \cP\). The consequence is that one can vary \(\beta\) while holding \(\eta\) fixed, and vice versa. That is, for sufficiently small \(t\), the pairs \((\beta_0 + t, \eta_0)\) and \((\beta_0, \eta_0 + th)\) are both attainable, meaning there exist distributions in \(\cP\) realizing those functional values. Without such a product neighborhood, the attainable pairs near \((\beta_0, \eta_0)\) could lie along a lower-dimensional surface, so that changing \(\beta\) might force \(\eta\) to change as well.

Crucially, however, local variation independence is purely a set-theoretic statement about the attainable set \(\Theta\). The condition guarantees that for each small \(t\), there exists at least one distribution \(P \in \cP\) with \((\beta(P), \eta(P))  = (\beta_0 + t, \eta_0)\). However, this is only a pointwise existence guarantee and imposes no regularity on how such choices may depend on \(t\). In particular, local variation independence does not imply that there exists a map \(t \mapsto P_{t} \in \cP\) satisfying \((\beta(P_t), \eta(P_t)) = (\beta_0 + t, \eta_0)\) that is quadratic-mean differentiable at \(t= 0\).

\begin{assumption}[Regular coordinate submodels]
\label{ass:coordinate_qmd}
For every admissible direction \(h \in \dot{\cH}\), the paths \(t \mapsto P_{\beta_0 + t,\, \eta_0}\) and \(t \mapsto P_{\beta_0,\, \eta_0 + th}\) exist in \(\cP\) for sufficiently small \(|t|\) and are regular (QMD) submodels through \(P_0\) at \(t = 0\).
\end{assumption}

\begin{proposition}
\label{prop:lvi_plus_qmd}
If \(\beta\) and \(\eta\) are locally variation independent at \(P_0\) (Definition~\ref{def:local_variation_independence}) and satisfy coordinate QMD smoothness (Assumption~\ref{ass:coordinate_qmd}), then local product structure (Assumption~\ref{assumption:product-structure}) holds.
\end{proposition}

\begin{proof}
Local variation independence provides neighborhoods \(U \ni \beta_0\) and \(V \ni \eta_0\) with \(U \times V \subseteq \Theta\). For small \(|t|\), the parameter values \((\beta_0 + t,\, \eta_0)\) and \((\beta_0,\, \eta_0 + th)\) lie in \(U \times V\) and hence correspond to distributions in \(\cP\). Assumption~\ref{ass:coordinate_qmd} asserts that the resulting paths are differentiable in quadratic mean at \(t = 0\), giving exactly the conditions of Assumption~\ref{assumption:product-structure}.
\end{proof}
We note that Assumption~\ref{assumption:product-structure} is strictly weaker than this combination in two respects, as it requires neither a full product neighborhood in the parameter space nor exact coordinate paths, only regular submodels with the correct first-order coordinate derivatives at \(P_0\). This distinction also applies to the work of \citet[Section~3.4]{bickel_efficient_1998}, which posits a product parameterization that builds in local variation independence by construction, and requires the target submodel to be a regular parametric family. Together, these conditions are strictly stronger than Assumption~\ref{assumption:product-structure}.

\subsection{Revisiting the Gradient Characterization}
\label{subsec:revisiting_vdlr}

As discussed in Section~\ref{subsec:local_product_structure}, the distinction between the set-theoretic content of local variation independence and the analytic content of Assumption~\ref{assumption:product-structure} is subtle, and it is natural to ask whether this distinction matters in practice. We demonstrate that the answer is affirmative by revisiting the classical results of \citet[Section~1.4]{van_der_laan_unified_2003}, which connect influence functions to estimating functions. Their framework contains the essential insight that underpins the equivalence we formalize in Section~\ref{sec:equivalence}. However, the regularity of submodels that perturb \(\beta\) and \(\eta\) independently, which we have isolated as Assumption~\ref{assumption:product-structure}, plays an important role in their argument that was not separately identified. Making this explicit is the purpose of the present subsection.

We focus on two results from \citet{van_der_laan_unified_2003}: their Lemma~1.2, which characterizes gradients through the derivative of an expected estimating function along arbitrary submodels, and Lemma~1.3, which establishes that the derivative of the expected estimating function with respect to \(\beta\) at fixed \(\eta_0\) equals \(-1\). This latter result is the key step that links influence functions to estimating functions and underpins the construction of efficient estimators via solving moment conditions. We will show that the proof of Lemma~1.3 contains an implicit step, replacing the varying nuisance \(\eta(P_{t,s})\) by the fixed value \(\eta_0\) inside a derivative, that requires the nuisance tangent space to capture all nuisance directions, which in turn requires the local product structure of Assumption~\ref{assumption:product-structure}.

For the reader's convenience, we state the relevant results in our notation. The correspondence with \citet{van_der_laan_unified_2003} is:
\[P_0 \leftrightarrow F_X, \beta \leftrightarrow \mu, \eta \leftrightarrow \rho, \bE_0 \leftrightarrow \bE_{F_X}, \{P_{t, s}\} \leftrightarrow\{F_{\epsilon, s}\}, \varphi^* \leftrightarrow S_{\mathbf{eff}}^{*F}, \Lambda \leftrightarrow T_{\mathrm{nuis}}^F, \cP \leftrightarrow \cM^F.\]

\paragraph{Setup.} 
The framework of \citet{van_der_laan_unified_2003} posits a class of estimating functions indexed by an abstract label \(k\), mapping each distribution in the model to a mean-zero function of the data. The key structural requirement is that these estimating functions, evaluated at the true parameter values, span  the orthogonal complement \(\Lambda^\perp\) of the nuisance tangent space. Since influence functions are orthogonal to \(\Lambda\) by Lemma~\ref{lem:ortho_nuisance}, this ensures that every candidate influence function is representable as an estimating function, and combined with unbiasedness along submodels, allows one to recover the inner-product characterization linking estimating functions to gradients (Lemma~\ref{lem:gradient_characterization}). We collect the precise conditions as follows.

\begin{assumption}[Estimating function representation]
\label{assumption:ef_representation}
Suppose there exists an abstract index set \(\cK\) and a mapping \((k, \beta, \eta) \mapsto D_k(\cdot \mid \beta, \eta)\) from \(\cK \times \Theta\) into functions of \(Z\) such that:
\begin{enumerate}
    \item \textbf{Unbiased estimating function.} \(\bE_{P}[D_k(Z \mid \beta(P), \eta(P))] = 0\) for all \(P \in \cP\) and all \(k \in \cK\).
    \item \textbf{Richness.} The index set \(\cK\) is rich enough that, at \(P_0\),
    \[\Lambda^\perp = \{D_k(\cdot \mid \beta_0, \eta_0) : k \in \cK(P_0)\},\]
    where \(\cK(P_0) \subseteq \cK\) is the index set at \(P_0\) and \(\Lambda^\perp\) denotes the orthogonal complement of the nuisance tangent space \(\Lambda\) (Definition~\ref{def:nuisance}) inside \(L_2^0(P_0)\).
    \item \textbf{Continuity along submodels.} For all \(k \in \cK(P_0)\) and each regular submodel \(\{P_{t, s}\}\) with score \(s \in \cS\),
    \[\|D_k(\cdot \mid \beta(P_{t, s}), \eta(P_{t, s})) - D_k(\cdot \mid \beta_0, \eta_0)\|_{L_2(P_0)} \to 0 \quad \text{as } t \to 0.\]
    \item \textbf{Pathwise differentiability.} \(\beta\) is pathwise differentiable at \(P_0\) with efficient influence function \(\varphi^*\), and \(\bigl\langle \varphi^* \bigr\rangle \subset \cS\), where \(\bigl\langle \varphi^* \bigr\rangle\) denotes the one-dimensional span of \(\varphi^*\).
    \item \textbf{Uniform boundedness.} For all \(k \in \cK(P_0)\), there exist \(C < \infty\) and a neighborhood \(N\) of \((\beta_0, \eta_0)\) such that \[\sup_{z \in \cZ,\, (\beta, \eta) \in N} |D_k(z \mid \beta, \eta)| \le C.\]
\end{enumerate}
\end{assumption}

\paragraph{Gradient characterization (Lemma~1.2 of \citet{van_der_laan_unified_2003}).}
The first result characterizes which estimating functions are gradients. The idea is as follows: using the unbiasedness condition (i) of Assumption~\ref{assumption:ef_representation}, the expectation of \(D_k\) under \(P_{t,s}\) vanishes identically along any regular submodel. Differentiating this identity at $t = 0$ recovers an 
inner-product representation that determines when an estimating function is an influence function.
\begin{lemma}[Gradient characterization; Lemma~1.2 of \citet{van_der_laan_unified_2003}]
\label{lem:gradient_characterization}
Under Assumption~\ref{assumption:ef_representation}, define
\[f_k(s) := \frac{d}{dt}\bE_0[D_k(Z \mid \beta(P_{t, s}), \eta(P_{t, s}))] \bigg|_{t = 0}.\]
Then an element \(D = D_k(\cdot \mid \beta_0, \eta_0) \in \Lambda^\perp\) for \(k \in \cK(P_0)\) is a gradient if and only if
\[f_k(s) = \begin{cases} 0 & \text{if } s \in \cS_{\mathrm{nuis}}, \\ -\frac{d}{dt}\beta(P_{t, s})|_{t = 0} & \text{if } s \in \bigl\langle \varphi^* \bigr\rangle. \end{cases}\]
\end{lemma}

\begin{proof}
By Assumption~\ref{assumption:ef_representation} (i), \[\bE_{P_{t, s}}[D_k(Z \mid \beta(P_{t, s}), \eta(P_{t, s}))] = 0\] for all sufficiently small \(t\). Combined with \(\bE_0[D_k(Z \mid \beta_0, \eta_0)] = 0\), we can write
\begin{align*}
\frac{1}{t}\bE_0\left[D_k\left(Z \mid \beta(P_{t, s}), \eta(P_{t, s})\right)\right] 
&= \frac{1}{t}\left\{\bE_0\left[D_k\left(Z \mid \beta(P_{t, s}), \eta(P_{t, s})\right)\right] - \bE_{P_{t, s}}\left[D_k\left(Z \mid \beta(P_{t, s}), \eta(P_{t, s})\right)\right]\right\}  \\
&= \int D_k\left(z \mid \beta(P_{t, s}), \eta(P_{t, s})\right) \frac{dP_0 - dP_{t,s}}{t}(z) 
\end{align*}
Define \(g_t  := D_k\left(z \mid \beta(P_{t, s}), \eta(P_{t, s})\right)\) and  \(g_0  := D_k\left(z \mid \beta_0, \eta_0\right)\). Writing \(dP_0 = p_0\,d\nu\) and \(dP_{t,s} = p_t\,d\nu\),
\begin{align*}
  \int g_t \cdot \frac{p_0 - p_t}{t}\, d\nu
  &= -\int g_t \cdot \frac{\sqrt{p_t} - \sqrt{p_0}}{t}\,
     \bigl(\sqrt{p_t} + \sqrt{p_0}\bigr)\, d\nu \\
  &= -\int g_0 \cdot \frac{\sqrt{p_t} - \sqrt{p_0}}{t}\,
     \bigl(\sqrt{p_t} + \sqrt{p_0}\bigr)\, d\nu
     - \int (g_t - g_0) \cdot \frac{\sqrt{p_t} - \sqrt{p_0}}{t}\,
     \bigl(\sqrt{p_t} + \sqrt{p_0}\bigr)\, d\nu. 
\end{align*}
The first integral converges to $\bE_0[g_0s]$ by the same argument as in the proof of Lemma~\ref{lem:differentiation_fixed_f}. For the second integral, Cauchy--Schwarz gives
\[\left|\int (g_t - g_0) \cdot \frac{\sqrt{p_t} - \sqrt{p_0}}{t}\,
     \bigl(\sqrt{p_t} + \sqrt{p_0}\bigr)\, d\nu\right|
  \;\le\; \left\|(g_t - g_0)\bigl(\sqrt{p_t} + \sqrt{p_0}\bigr)\right\|_{L_2(\nu)}
     \cdot \left\|\frac{\sqrt{p_t} - \sqrt{p_0}}{t}\right\|_{L_2(\nu)},\]
where the second factor is bounded by QMD. For the first factor,
\begin{align*}
  \left\|(g_t - g_0)\bigl(\sqrt{p_t} + \sqrt{p_0}\bigr)\right\|_{L_2(\nu)}^2
  &= \int (g_t - g_0)^2\,\bigl(\sqrt{p_t} + \sqrt{p_0}\bigr)^2\, d\nu \\
  &\le 2\int (g_t - g_0)^2\,(p_t + p_0)\, d\nu \\
  &= 2\left(\bE_{P_t}\left[(g_t - g_0)^{2}\right] + \bE_0\left[(g_t - g_0)^{2}\right]\right).
\end{align*}
The term \(\bE_0\left[(g_t - g_0)^{2}\right] \to 0\) by Assumption~\ref{assumption:ef_representation} (iii). For \(\bE_{P_t}\left[(g_t - g_0)^{2}\right]\), write
\begin{align*}
  \bE_{P_t}\left[(g_t - g_0)^{2}\right]
  &= \bE_0\left[(g_t - g_0)^{2}\right] + \int (g_t - g_0)^2\,(p_t - p_0)\, d\nu \\
  &\le \bE_0\left[(g_t - g_0)^{2}\right] + 4C^2 \int |p_t - p_0|\, d\nu,
\end{align*}
where the second inequality uses Assumption~\ref{assumption:ef_representation} (v), and \(\int |p_t - p_0|\, d\nu \to 0\) again by QMD. Thus,
\[f_k(s) = -\bE_0[D_k(Z \mid \beta_0, \eta_0) \cdot s(Z)] = -\langle D_k(\cdot \mid \beta_0, \eta_0),\, s\rangle_{P_0}.\]
By definition, \(D_k\) is a gradient if and only if the inner product equals zero for all \(s \in \cS_{\mathrm{nuis}}\) and equals \(\frac{d}{dt}\beta(P_{t, s})|_{t = 0}\) for \(s \in \langle \varphi^* \rangle\). This is equivalent to the stated conditions on \(f_k\).
\end{proof}

\paragraph{The negative identity (Lemma~1.3 of 
\citet{van_der_laan_unified_2003}).} The second result builds on Lemma~\ref{lem:gradient_characterization} to establish that if $D_k(\cdot \mid \beta_0, \eta_0)$ is an influence function, then the partial derivative of its expectation $\bE_0[D_k(Z \mid \beta, \eta_0)]$ with respect to $\beta$ at $\beta_0$ equals $-1$. This central result links influence functions to estimating functions and allows efficient estimators to be obtained by solving moment conditions. Below, we reproduce the argument of \citet{van_der_laan_unified_2003} essentially unchanged, with the addition of making explicit a regularity condition it leaves implicit. The argument differentiates the expected estimating function with the nuisance held fixed. In their formulation this step is immediate, since the target and nuisance parameters vary independently thus perturbations of the nuisance alone are available from the outset. What the step requires in addition, however, is less the independence of the two parameters than the regularity of the submodels that realize these perturbations---the local product structure of Assumption~\ref{assumption:product-structure}, which their treatment uses without separately identifying it. Accordingly, we present the proof as given by \citet{van_der_laan_unified_2003} and then proceed to clarify the precise point at which it is warranted.

\begin{lemma}[Negative identity; Lemma~1.3 of \citet{van_der_laan_unified_2003}]
\label{lem:negative_identity}
In addition to Assumption~\ref{assumption:ef_representation}, assume that \(\beta\) and \(\eta\) are locally variation independent at \(P_0\) (Definition~\ref{def:local_variation_independence}), that \(\beta \mapsto \bE_0[D_k(Z \mid \beta, \eta_0)]\) is differentiable at \(\beta_0\) with nonzero derivative for all \(k \in \cK(P_0)\), and that \(\bE_0\left[\varphi^*(Z)^{2}\right] > 0\). If \(D_k(\cdot \mid \beta_0, \eta_0)\) is a gradient, then
\[\frac{d}{d\beta}\bE_0[D_k(Z \mid \beta, \eta_0)]\bigg|_{\beta = \beta_0} = -1.\]
\end{lemma}

\begin{proof}[Proof (as given by \citet{van_der_laan_unified_2003})]
Let \(s \in \cS\) be a scalar multiple of \(\varphi^*\), say \(s = c\varphi^*\) for some \(c \neq 0\). Since \(\varphi^* \in \cS\) by Assumption~\ref{assumption:ef_representation} (iv), \(s\) is the score of some regular submodel \(\{P_{t, s}\}\) through \(P_0\). Define
\[h_{2,s}(t) := \beta(P_{t, s}), \qquad h_1(\beta) := \bE_0[D_k(Z \mid \beta, \eta_0)].\]
The map \(t \mapsto \bE_0[D_k(Z \mid \beta(P_{t,s}), \eta_0)]\) is the composition \(h_1(h_{2,s}(t))\). As in the proof of Lemma~\ref{lem:gradient_characterization},
\begin{equation}
\label{eq:key_derivative}
\frac{d}{dt}h_1(h_{2,s}(t))\bigg|_{t = 0} = -\frac{d}{dt}\beta(P_{t, s})\bigg|_{t = 0}.
\end{equation}
By the chain rule, the left-hand side equals \(h_1'(\beta_0) \cdot h_{2,s}'(0)\). Pathwise differentiability gives
\[h_{2,s}'(0) = \frac{d}{dt}\beta(P_{t,s})\bigg|_{t=0} = \bE_0[\varphi^* s] = c\,\bE_0\left[(\varphi^*)^{2}\right] \neq 0.\]
So \(h_1'(\beta_0) \cdot h_{2,s}'(0) = -h_{2,s}'(0)\). Since \(h_{2,s}'(0) \neq 0\), it follows that \(h_1'(\beta_0) = -1\).
\end{proof}

\paragraph{The role of regularity.} The proof invokes ``as in the proof of Lemma~\ref{lem:gradient_characterization}'' to claim~\eqref{eq:key_derivative}, i.e.,
\begin{equation}
\label{eq:frozen_eta}
\frac{d}{dt}\bE_0[D_k(Z \mid \beta(P_{t, s}), \eta_0)]\bigg|_{t = 0} = -\frac{d}{dt}\beta(P_{t, s})\bigg|_{t = 0}.
\end{equation}
However, Lemma~\ref{lem:gradient_characterization} actually established
\begin{equation}
\label{eq:moving_eta}
\frac{d}{dt}\bE_0[D_k(Z \mid \beta(P_{t, s}), \eta(P_{t, s}))]\bigg|_{t = 0} = -\frac{d}{dt}\beta(P_{t, s})\bigg|_{t = 0},
\end{equation}
where \(\eta(P_{t,s})\) varies with \(t\). For~\eqref{eq:frozen_eta} to follow from~\eqref{eq:moving_eta}, one must show that replacing \(\eta(P_{t,s})\) by the fixed value \(\eta_0\) does not affect the derivative, i.e., that
\begin{equation}
\label{eq:nuisance_insensitivity}
\frac{\partial}{\partial \eta}\bE_0[D_k(Z \mid \beta_0, \eta)]\bigg|_{\eta = \eta_0}[h] = 0 \quad \forall h \in \dot{\cH}.
\end{equation}
To see why~\eqref{eq:nuisance_insensitivity} is needed, suppose that the map \((\beta, \eta) \mapsto \bE_0[D_k(Z \mid \beta, 
\eta)]\) is Fr\'echet differentiable at \((\beta_0, \eta_0)\). The chain rule decomposes~\eqref{eq:moving_eta} as 
\begin{align*}
&\frac{d}{dt}\bE_0\!\left[D_k\!\left(Z \mid \beta(P_{t,s}), \eta(P_{t,s})\right)\right]\bigg|_{t=0}\\
= &\frac{d}{dt}\bE_0\!\left[D_k\!\left(Z \mid \beta(P_{t,s}), \eta_0\right)\right]\bigg|_{t=0} + \frac{\partial}{\partial \eta}\bE_0\!\left[D_k\!\left(Z \mid \beta_0, \eta\right)\right]\bigg|_{\eta = \eta_0}\!\left[\frac{d}{dt}\eta(P_{t,s})\bigg|_{t=0}\right]
\end{align*}
so~\eqref{eq:frozen_eta} follows from~\eqref{eq:moving_eta} if 
and only if the second term vanishes. 
\eqref{eq:nuisance_insensitivity} guarantees this by requiring that the nuisance derivative of the expected  estimating function vanishes in every direction 
\(h \in \dot{\cH}\).

We now show that establishing~\eqref{eq:nuisance_insensitivity} 
requires Assumption~\ref{assumption:product-structure}. Apply 
Lemma~\ref{lem:gradient_characterization} to a nuisance score 
\(s_{\mathrm{nuis}} \in \cS_{\mathrm{nuis}}\). Since 
\(D_k(\cdot \mid \beta_0, \eta_0)\) is an influence function, 
\(f_k(s_{\mathrm{nuis}}) = 0\), i.e.,
\[\frac{d}{dt}\bE_0\!\left[D_k\!\left(Z \mid \beta(P_{t, 
s_{\mathrm{nuis}}}), \eta(P_{t, s_{\mathrm{nuis}}})\right)
\right]\bigg|_{t = 0} = 0.\]
Assuming again Fr\'echet differentiability, the chain rule gives
\[\frac{\partial}{\partial \beta}\bE_0\!\left[D_k\!\left(Z \mid \beta, \eta_0\right)\right]\bigg|_{\beta = \beta_0} \cdot \frac{d}{dt}\beta(P_{t, s_{\mathrm{nuis}}})\bigg|_{t=0} + \frac{\partial}{\partial \eta}\bE_0\!\left[D_k\!\left(Z \mid \beta_0,\eta\right)\right]\bigg|_{\eta = \eta_0}\!\left[\frac{d}{dt}\eta(P_{t, s_{\mathrm{nuis}}})\bigg|_{t=0}\right] = 0.\]
Since \(s_{\mathrm{nuis}}\) is a nuisance score, 
\(\frac{d}{dt}\beta(P_{t, s_{\mathrm{nuis}}})|_{t=0} = 0\), 
so the first term vanishes and we obtain
\[\frac{\partial}{\partial \eta}\bE_0\!\left[D_k\!\left(Z \mid \beta_0, \eta\right)\right]\bigg|_{\eta = \eta_0}\!\left[\frac{d}{dt}\eta(P_{t, s_{\mathrm{nuis}}})\bigg|_{t=0}\right] = 0.\]
This establishes~\eqref{eq:nuisance_insensitivity} only for  those directions \(h \in \dot{\cH}\) that arise as nuisance  derivatives of submodels in \(\cS_{\mathrm{nuis}}\). A priori, these nuisance derivatives populate some subset of \(\dot{\cH}\), but there is no reason this subset should exhaust \(\dot{\cH}\). 
Assumption~\ref{assumption:product-structure} closes the 
remaining gap by furnishing for each \(h \in \dot{\cH}\) a regular submodel with \(\frac{d}{dt}\beta(P_t)|_{t=0} = 0\) and \(\frac{d}{dt}\eta(P_t)|_{t=0} = h\). Since \(\frac{d}{dt}\beta(P_t)|_{t=0} = 0\), the score of this submodel is a nuisance score, and its nuisance derivative at \(t = 0\) is exactly \(h\). The argument above then yields~\eqref{eq:nuisance_insensitivity} for this \(h\). Since \(h \in \dot{\cH}\) was arbitrary, \eqref{eq:nuisance_insensitivity} holds in full generality.

With~\eqref{eq:nuisance_insensitivity} in hand, the passage from~\eqref{eq:moving_eta} to~\eqref{eq:frozen_eta} immediately follows. For any score \(s = c\varphi^*\), the same chain-rule 
decomposition used above gives
\begin{align*}
&\frac{d}{dt}\bE_0\!\left[D_k\!\left(Z \mid \beta(P_{t,s}), \eta(P_{t,s})\right)\right]\bigg|_{t=0} \\
= &\frac{\partial}{\partial \beta}\bE_0\!\left[D_k\!\left(Z \mid \beta, \eta_0\right)\right]\bigg|_{\beta = \beta_0} \cdot \frac{d}{dt}\beta(P_{t,s})\bigg|_{t=0} + \frac{\partial}{\partial \eta}\bE_0\!\left[D_k\!\left(Z \mid \beta_0, \eta\right)\right]\bigg|_{\eta = \eta_0}\!\left[\frac{d}{dt}\eta(P_{t,s})\bigg|_{t=0}\right] \\
= &\frac{\partial}{\partial \beta}\bE_0\!\left[D_k\!\left(Z \mid \beta, \eta_0\right)\right]\bigg|_{\beta = \beta_0} \cdot \frac{d}{dt}\beta(P_{t,s})\bigg|_{t=0} \\
= &\frac{d}{dt}\bE_0\!\left[D_k\!\left(Z \mid \beta(P_{t,s}), \eta_0\right)\right]\bigg|_{t=0},
\end{align*}
where the second equality 
follows from~\eqref{eq:nuisance_insensitivity}, and the proof 
of the negative identity then proceeds as written.

\begin{remark}[Fr\'echet differentiability]
\label{remark:frechet}
The chain-rule decompositions above require Fr\'echet differentiability of the map \((\beta, \eta) \mapsto \bE_0[D_k(Z \mid \beta, \eta)]\) at \((\beta_0, \eta_0)\),  which is not explicitly stated in Lemma~1.3 of \citet{van_der_laan_unified_2003}. The paragraph immediately preceding Lemma~1.3 in their exposition, however, suggests that smoothness conditions should be jointly imposed on \(\beta\) and \(\eta\).
\end{remark}

\begin{remark}[Boundedness and the score definition]
\label{remark:score_definition}
The reader may notice Assumption~\ref{assumption:ef_representation} (v) imposes
uniform boundedness on the estimating functions, a condition
not present in the corresponding result of \citet{van_der_laan_unified_2003}. This difference traces to the definition of the score, where \citet{van_der_laan_unified_2003} define the score as the $L_2(P_0)$ limit of the density ratio $(p_t/p_0 - 1)/t$, which is strictly stronger than the quadratic mean differentiability (QMD) formulation of \citet{van_der_vaart_asymptotic_1998} adopted herein. Under their definition, the convergence in the proof of Lemma~\ref{lem:gradient_characterization} follows from Cauchy--Schwarz in $L_2(P_0)$ alone. Under QMD, the same step requires decomposing through $\bigl(\sqrt{p_t} - \sqrt{p_0}\bigr)\bigl(\sqrt{p_t} + \sqrt{p_0}\bigr)$, and bounding the resulting cross term requires introducing the uniform boundedness condition. It should be noted that the uniform boundedness condition is an artifact of the QMD formulation and not a structural requirement of the arguments. We adopt QMD throughout to maintain a single consistent convention, and the distinction between local product structure and variation independence arises independently of which score formulation is adopted.
\end{remark}

\section{Examples}

\subsection{Average Treatment Effect}
\label{appendix:ate}

We first illustrate the equivalence results of Section~\ref{sec:equivalence} through a detailed worked example on the average treatment effect. For each direction of the equivalence, we verify every assumption and construct the required objects explicitly.

\paragraph{Setup.}
Let \(Z = (Y, X, A)\) with confounders \(X\), binary treatment
\(A \in \{0, 1\}\), and outcome \(Y\). We assume that the standard causal assumptions of consistency, positivity, and no unmeasured confounding hold. We work in the nonparametric model \(\cP\) consisting of all densities \(p\) with respect to a \(\sigma\)-finite dominating measure \(\nu\) that satisfy the regularity conditions (R1)--(R2) below. We fix \(P_0 \in \cP\) and define the nuisance quantities
\[
  \mu_a(x) := \bE_0[Y \mid X = x,\, A = a], \qquad
  \pi(x) := P_0(A = 1 \mid X = x),
\]
the treatment effect function \(\tau(x) := \mu_1(x) - \mu_0(x)\), and the conditional outcome variance \(\sigma_a^2(x) := \mathrm{Var}_0(Y \mid X = x,\, A = a)\). The target and nuisance functionals are
\[
  \beta(P) := \bE_P\left[\mu_1^P(X) - \mu_0^P(X)\right], \qquad
  \eta(P) := \bigl(\mu_1^P,\, \mu_0^P,\, \pi^P\bigr),
\]
with \(\beta_0 := \beta(P_0)\) and \(\eta_0 := (\mu_1,\, \mu_0,\, \pi)\).

\paragraph{Regularity conditions.}
We further impose the following conditions, which ensure that the constructed submodels are well-behaved. Note that positivity already appeared as an identification assumption.
\begin{enumerate}
  \item[(R1)] \textbf{Positivity.} There exists \(\varepsilon > 0\)
    such that \(\pi^P(x) \in [\varepsilon,\, 1 - \varepsilon]\) for all \(x\) and \(P \in \cP\).
  \item[(R2)] \textbf{Bounded outcomes.} There exists
    \(C_Y < \infty\) such that \(|y| \le C_Y\) for all \(y\).
  \item[(R3)] \textbf{Positive conditional variance.}
    \(\sigma_a^2(x) \ge \sigma^2 > 0\) for all \(x\) and \(a = 0, 1\).
  \item[(R4)] \textbf{Treatment effect heterogeneity.}
    \(\mathrm{Var}_0(\tau(X)) > 0\).
\end{enumerate}

We also assume an interior positivity margin at \(P_0\): there exists
\(\varepsilon' > \varepsilon\) such that
\(\pi(x) \in [\varepsilon', 1 - \varepsilon']\) for all \(x\). This ensures that for any bounded mean-zero
\(g\), the linear tilt \(P_t\) with density \(p_0(1 + tg)\)
remains in \(\cP\) for sufficiently small \(|t|\), so the tangent
space at \(P_0\) is \(\cT = L_2^0(P_0)\) by the same argument as
Corollary~\ref{cor:saturation}.

Finally, we take the ambient normed space to be
\(\cV := L_\infty(P_{0,X})^3\) with the product supremum norm,
and the nuisance parameter set to be
\[
  \cH := \left\{(\mu_1, \mu_0, \pi) \in \cV \,:\,
  \operatorname{ess\,inf}_{P_{0,X}} \pi > 0 \;\text{ and }\;
  \operatorname{ess\,inf}_{P_{0,X}} (1 - \pi) > 0\right\}.
\]
Since \(\pi(x) \in [\varepsilon, 1 - \varepsilon]\) for all \(x\) 
by~(R1) and \(|\mu_a(x)| \le C_Y\) for all \(x\)\ by~(R2), it follows \(\eta_0 \in \cH\). Moreover, since \(\cH\) is open in \(\cV\), the admissible perturbation space is
\(\dot{\cH} = \cV\).

\paragraph{Estimating function and influence function.} Define the estimating function
\begin{equation}
    m(Z;\, \beta,\, \eta) := \frac{A}{\pi(X)}(Y - \mu_1(X)) - \frac{1 - A}{1 - \pi(X)}(Y - \mu_0(X)) + \mu_1(X) - \mu_0(X) - \beta, 
    \label{eq:ate_m}
\end{equation}
and the influence function at the truth,
\begin{equation}
    \varphi(Z) := m(Z;\, \beta_0,\, \eta_0) = \frac{A}{\pi(X)}(Y - \mu_1(X)) - \frac{1 - A}{1 - \pi(X)}(Y - \mu_0(X)) + \tau(X) - \beta_0.
    \label{eq:ate_phi}
\end{equation}

\subsubsection{Forward direction}
\label{subsubsec:ate_forward}

We verify Assumptions~\ref{assumption:correct_spec}--\ref{assumption:hellinger} and apply Theorem~\ref{thm:forward} to conclude that \(\beta\) is pathwise differentiable with influence function \(\varphi(Z) = -G^{-1}m(Z;\, \beta_0,\, \eta_0)\).

\textbf{Assumption~\ref{assumption:correct_spec}.}
Let \(P\) be any distribution with \(\beta(P) = \beta\) and \(\eta(P) = (\mu_1, \mu_0, \pi)\). We show \(\bE_P[m(Z;\, \beta,\, \eta)] = 0\). By the tower property, conditioning first on \(X\) and then on \((X, A)\), and using the definition \(\mu_1(x) = \bE_P[Y \mid X = x,\, A = 1]\),
\begin{align*}
  \bE_P\left[\frac{A}{\pi(X)}(Y - \mu_1(X))\right] &= \bE_P\left[\bE_P\left[\frac{A}{\pi(X)}(Y - \mu_1(X))\middle| X\right]\right]\\
  &= \bE_P\left[\frac{\pi(X)}{\pi(X)} \cdot \bE_P[Y - \mu_1(X) \mid X,\, A = 1]\right] = 0,
\end{align*}
where the second equality uses \(\bE_P[A \cdot f(Z) \mid X] = \pi(X) \cdot \bE_P[f(Z) \mid X,\, A = 1]\). The second IPW term vanishes identically by the same argument with \(a = 0\). The remaining terms contribute \(\bE_P[\mu_1(X) - \mu_0(X)] - \beta = \beta - \beta = 0\). 

\textbf{Assumption~\ref{assumption:jacobian}.}
Since \(m\) is linear in \(\beta\) with coefficient \(-1\), we have \(\partial_\beta m(Z;\, \beta_0,\, \eta_0) = -1\) identically, so
\[
  G := \bE_0[\partial_\beta m(Z;\, \beta_0,\, \eta_0)] = -1 \neq 0. 
\]

\textbf{Assumption~\ref{assumption:neyman}.}
We verify that the G\^ateaux derivative of \(\eta \mapsto \bE_0[m(Z;\, \beta_0,\, \eta)]\) vanishes at \(\eta_0\) in each coordinate direction of \(\dot{\cH}\). Since \(\eta = (\mu_1, \mu_0, \pi)\) and the admissible perturbation space \(\dot{\cH}\) is a product, linearity allows us to check each component separately. Recall that throughout, the expectation \(\bE_0\) is taken under the fixed measure \(P_0\) and only the function arguments inside \(m\) are being varied.

\textbf{Perturbation \(\mu_1 \to \mu_1 + th_1\).}
Substituting \(\mu_1 + th_1\) into~\eqref{eq:ate_m} with \(\beta = \beta_0\) and \((\mu_0, \pi)\) held at their true values, the only terms affected are the first IPW term \[\frac{A}{\pi(X)}(Y - \mu_1(X) - th_1(X))\] and the outcome regression \[\mu_1(X) + th_1(X) - \mu_0(X).\] Taking the expectation under \(P_0\) and differentiating at \(t = 0\):
\[
  \frac{d}{dt}\bE_0[m(Z; \beta_0, (\mu_1 + th_1, \mu_0,\pi))]\bigg|_{t=0} = \bE_0\left[-\frac{A h_1(X)}{\pi(X)} + h_1(X)\right] = \bE_0\left[h_1(X)\left(1 - \frac{A}{\pi(X)}\right)\right].
\]
Conditioning on \(X\) and using \(\bE_0[A \mid X] = \pi(X)\):
\[
  \bE_0\left[1 - \frac{A}{\pi(X)} \middle| X\right] = 1 - \frac{\pi(X)}{\pi(X)} = 0.
\]
By the tower property, the derivative vanishes for all \(h_1 \in L_\infty(P_{0,X})\). The perturbation \(\mu_0 \to \mu_0 + th_0\) follows by an identical argument. 

\textbf{Perturbation \(\pi \to \pi + th_\pi\).}
Substituting \(\pi + th_\pi\) affects only the denominators of the two IPW terms. Since the outcome regression \(\mu_1(X) - \mu_0(X) - \beta_0\) does not involve \(\pi\), we differentiate only the IPW terms. Using \[\frac{d}{dt}\frac{1}{\pi + th_\pi}\bigg|_{t=0} = -h_\pi / \pi^2 \quad \text{and} \quad \frac{d}{dt}\frac{1}{1-\pi-th_\pi}\bigg|_{t=0} = h_\pi/(1-\pi)^2,\] we can write
\[
  \frac{d}{dt}\bE_0[m(Z;\beta_0,(\mu_1,\mu_0,\pi + th_\pi))]\bigg|_{t=0} = \bE_0\left[-\frac{A\, h_\pi(X)}{\pi(X)^2}(Y - \mu_1(X)) - \frac{(1-A)\, h_\pi(X)}{(1-\pi(X))^2}(Y - \mu_0(X))\right].
\]
For the first term, we condition on \(X\):
\[
  \bE_0\left[\frac{A\,(Y-\mu_1(X))}{\pi(X)^2}\middle| X\right] = \frac{1}{\pi(X)^2}\,\bE_0[A(Y-\mu_1(X)) \mid X] = \frac{\pi(X)}{\pi(X)^2}\, \bE_0[Y - \mu_1(X) \mid X,\, A=1] = 0,
\]
where we used \(\bE_0[A \cdot f(Z) \mid X] = \pi(X)\, \bE_0[f(Z) \mid X,\, A=1]\) and the definition of \(\mu_1\). The second term vanishes identically by the same argument with \(a = 0\).

\textbf{Assumptions~\ref{assumption:coord_smooth_all}--\ref{assumption:regularity_submodels}.}
Under (R1)--(R2), the map \((\beta,\, \eta) \mapsto m(\cdot\,;\, \beta,\, \eta) \in L_2(P_0)\) is Fr\'echet differentiable at \((\beta_0,\, \eta_0)\). The partial derivatives computed above are bounded linear maps into \(L_2(P_0)\), with boundedness following from \(\pi \ge \varepsilon\) and \(|Y| \le C_Y\), which ensure all IPW-weighted terms lie in \(L_\infty(P_0)\). We take \(S = L_\infty(P_0) \cap L_2^0(P_0)\), which is dense in \(\cT = L_2^0(P_0)\) by the same argument as Corollary~\ref{cor:saturation}, and for each \(s \in S\) we use the linear tilt submodel from Lemma~\ref{lem:linear_tilt_submodel}. The induced coordinate paths \(t \mapsto (\beta_{t,s},\, \eta_{t,s})\) are differentiable at \(t = 0\), which follows from the explicit derivative formulas
\begin{align*}
  \frac{\partial}{\partial t}\,\mu_{a,t}(x)\bigg|_{t=0} &= \bE_0[(Y - \mu_a(X))\, s(Z) \mid X = x,\, A = a],\\
  \frac{\partial}{\partial t}\,\pi_{t}(x)\bigg|_{t=0} &= \bE_0[(A - \pi(X))\, s(Z) \mid X = x], 
\end{align*}
which are derived in the pathwise differentiability verification below via the quotient rule. The uniform second moment bound of Lemma~\ref{lem:differentiation_varying_f} holds since for linear tilt submodels with bounded scores, the nuisance difference quotients \((\mu_{a,t} - \mu_a)/t\) and \((\pi_t - \pi)/t\) admit closed-form expressions via the change-of-measure identity \(\mu_{a,t}(x) = \bE_0[Y(1+ts) \mid X{=}x, A{=}a] / \bE_0[(1+ts) \mid X{=}x, A{=}a]\) and similarly for \(\pi_t\), which are uniformly bounded in \(x\) for small \(t\) under (R1)--(R2). Combined with the boundedness of \((\beta_t - \beta_0)/t\) from coordinate smoothness, this yields a uniform \(L_\infty\) bound on the full difference quotient \((f_{t,s} - f_0)/t\), which dominates the \(L_2(P_{t,s})\) norm for any \(t\).

\textbf{Assumption~\ref{assumption:hellinger}.}
We show that \[|\beta(P_1) - \beta(P_2)| \le cH(P_1, P_2) \quad \text{ for all } P_1, P_2 \in \cP, \text{ where }c = 4\sqrt{2}C_Y(1 + 1/\varepsilon).\]
To start, write $p_j = dP_j/d\nu$ and recall $\tau^{P_j}(x) = \mu_1^{P_j}(x) - \mu_0^{P_j}(x)$ and $\beta(P_j) = \int \tau^{P_j}(x)\,dP_{j,X}(x)$.
We decompose
\[
\beta(P_1) - \beta(P_2)
= \underbrace{\int \left(\tau^{P_1}(x) - \tau^{P_2}(x)\right)\,dP_{1,X}(x)}_{(\mathrm{I})}
\;+\; \underbrace{\int \tau^{P_2}(x)\,d(P_{1,X} - P_{2,X})(x)}_{(\mathrm{II})}.
\]
By~(R2), $|\tau^{P_2}(x)| \le 2C_Y$, so
\[
|\mathrm{II}| \le 2C_Y \int |p_{1,X}(x) - p_{2,X}(x)|\,d\nu_X
= 4C_Y\,\mathrm{TV}(P_{1,X}, P_{2,X}).
\]
Since the marginal density is obtained by integrating out $(y,a)$,
\[
|p_{1,X}(x) - p_{2,X}(x)|
= \left|\sum_{a \in \{0,1\}} \int (p_1 - p_2)(y,x,a)\,d\nu_Y\right|
\le \sum_{a \in \{0,1\}} \int |p_1 - p_2|(y,x,a)\,d\nu_Y,
\]
where the inequality holds by the triangle inequality. Integrating over $\nu_X$, we obtain
\begin{align*}
    \mathrm{TV}(P_{1,X}, P_{2,X})
&= \frac{1}{2}\int |p_{1,X}(x) - p_{2,X}(x)|\,d\nu_X\\
&\le \frac{1}{2}\sum_{a \in \{0,1\}} \iint |p_1 - p_2|(y,x,a)\,d\nu_Y\,d\nu_X\\
&= \frac{1}{2}\int |p_1 - p_2|\,d\nu\\
&= \mathrm{TV}(P_1, P_2),
\end{align*}
hence $|\mathrm{II}| \le 4C_Y\,\mathrm{TV}(P_1, P_2)$. Next, by the triangle inequality,
\[
|\tau^{P_1}(x) - \tau^{P_2}(x)|
\le |\mu_1^{P_1}(x) - \mu_1^{P_2}(x)| + |\mu_0^{P_1}(x) - \mu_0^{P_2}(x)|,
\]
so it suffices to bound each $\int |\mu_a^{P_1}(x) - \mu_a^{P_2}(x)|\,dP_{1,X}(x)$ separately. Fix $a \in \{0,1\}$. By definition, $\mu_a^{P_j}(x) = \int y\,dP_j(y \mid x, a)$, so $\int (y - \mu_a^{P_2}(x))\,dP_2(y \mid x, a) = 0$. It follows that for $P_{1,X}$-a.s.\ $x$,
\begin{align*}
\mu_a^{P_1}(x) - \mu_a^{P_2}(x)
&= \int y\,dP_1(y \mid x, a) - \mu_a^{P_2}(x) \\
&= \int (y - \mu_a^{P_2}(x))\,dP_1(y \mid x, a) \\
&= \int (y - \mu_a^{P_2}(x))\,dP_1(y \mid x, a)
   - \int (y - \mu_a^{P_2}(x))\,dP_2(y \mid x, a) \\
&= \frac{\int (y - \mu_a^{P_2}(x))(p_1 - p_2)(y,x,a)\,d\nu_Y}{p_1(x,a)},
\end{align*}
where the last equality writes $dP_j(y \mid x, a) = p_j(y,x,a)\,d\nu_Y / p_j(x,a)$. By~(R2), $|y - \mu_a^{P_2}(x)| \le 2C_Y$, so
\[
|\mu_a^{P_1}(x) - \mu_a^{P_2}(x)|
\le \frac{2C_Y}{p_1(x,a)} \int |p_1 - p_2|(y,x,a)\,d\nu_Y.
\]
By~(R1), $p_1(x,a) \ge \varepsilon\,p_{1,X}(x)$, since $p_1(x,a) = \pi_a^{P_1}(x)\,p_{1,X}(x)$ and $\pi_a^{P_1}(x) \ge \varepsilon$. Multiplying both sides by $p_{1,X}(x)$:
\[
|\mu_a^{P_1}(x) - \mu_a^{P_2}(x)|\,p_{1,X}(x)
\le \frac{2C_Y}{\varepsilon} \int |p_1 - p_2|(y,x,a)\,d\nu_Y.
\]
Integrating over $\nu_X$ and summing over $a \in \{0,1\}$,
\[
\mathrm{I}
\le \sum_{a=0}^{1} \int |\mu_a^{P_1}(x) - \mu_a^{P_2}(x)|\,dP_{1,X}(x)
\le \frac{2C_Y}{\varepsilon}
    \sum_{a=0}^{1} \iint |p_1 - p_2|(y,x,a)\,d\nu_Y\,d\nu_X
= \frac{4C_Y}{\varepsilon}\,\mathrm{TV}(P_1, P_2).
\]
Combining the above, we arrive at
\begin{align*}
|\beta(P_1) - \beta(P_2)|
&\le \left(4C_Y + \frac{4C_Y}{\varepsilon}\right)\mathrm{TV}(P_1, P_2)\\
&= 4C_Y\!\left(1 + \frac{1}{\varepsilon}\right)\mathrm{TV}(P_1, P_2) \\
&\le 4\sqrt{2}C_Y \left(1 + \frac{1}{\varepsilon}\right) H(P_1, P_2),
\end{align*}
where the last step uses \(\mathrm{TV}(P_1, P_2)
\le \sqrt{2}H(P_1, P_2).\) Therefore, Assumption~\ref{assumption:hellinger} holds with $c = 4\sqrt{2}C_Y(1 + 1/\varepsilon)$ and any $\delta > 0$.

Since the assumptions of Theorem~\ref{thm:forward} hold, we conclude that \(\beta\) is pathwise differentiable at \(P_0\) with influence function
\[
  \varphi(Z) = -G^{-1}\, m(Z;\, \beta_0,\, \eta_0) = -(-1)^{-1}\, m(Z;\, \beta_0,\, \eta_0) = m(Z;\, \beta_0,\, \eta_0).
\]

\subsubsection{Reverse direction}
\label{subsubsec:ate_reverse}

We verify Assumptions~\ref{assumption:pathwise}--\ref{assumption:regularity_coordinate} of Theorem~\ref{thm:reverse}, as Assumptions~\ref{assumption:correct_spec} and~\ref{assumption:frechet_all} have already been verified in the forward direction.

\textbf{Assumption~\ref{assumption:pathwise}.}
\label{subsubsec:ate_pathwise}

We show that for every linear tilt submodel \(p_t(z) = p_0(z)(1 + t g(z))\) with \(\bE_0[g] = 0\) and \(\|g\|_\infty \le M\), whose score is \(s \equiv g\) by Lemma~\ref{lem:linear_tilt_submodel},
\[\frac{d}{dt}\beta(P_t)\bigg|_{t=0} = \bE_0[\varphi(Z)\, g(Z)].\]
We first establish this identity for all bounded mean-zero scores below, and then extend the conclusion to all regular submodels via the approximation step used in the argument of Theorem~\ref{thm:forward}.

The derivative of \(\beta(P_t) = \bE_{P_t}[\mu_{1,t}(X) - \mu_{0,t}(X)]\) decomposes by the product rule into three terms:
\[\frac{d}{dt}\,\beta(P_t)\bigg|_{t=0} = \underbrace{\int_\cX \frac{\partial}{\partial t}\,\mu_{1,t}(x)\bigg|_{t=0} dP_0(x)}_{\text{(I)}} - \underbrace{\int_\cX \frac{\partial}{\partial t}\,\mu_{0,t}(x)\bigg|_{t=0} dP_0(x)}_{\text{(II)}} + \underbrace{\int_\cX \tau(x)\frac{\partial}{\partial t}\,dP_{t,X}(x)\bigg|_{t=0}}_{\text{(III)}}.\]
For the first two terms, 
\[\mu_{a,t}(x) = \int_{\cY} y\, dP_t(y,x,a) \big/ \int_{\cY} dP_t(y,x,a)\] 
where \(dP_t = (1 + t g)\, dP_0\). Define \[N_{\mu_a}(t) := \int_{\cY} y\, dP_0(y,x,a)(1 + t g(y,x,a)), \quad \text{and} \quad D_{\mu_a}(t) := \int_{\cY} dP_0(y,x,a)(1 + t g(y,x,a)).\] 
Then it follows
\begin{align*}
  N_{\mu_a}(0) &= \mu_a(x)\, dP_0(x,a), &N_{\mu_a}'(0) &= \bE_0[Yg(Z) \mid X=x, A=a]\, dP_0(x,a),\\
  D_{\mu_a}(0) &= dP_0(x,a), &D_{\mu_a}'(0) &= \bE_0[g(Z) \mid X=x, A=a]\, dP_0(x,a).
\end{align*}
By the quotient rule, we obtain
\begin{align}
\label{eq:mu_a_deriv}
  \frac{\partial}{\partial t}\,\mu_{a,t}(x)\bigg|_{t=0} 
  &= \biggl[N_{\mu_a}'(0)D_{\mu_a}(0) - N_{\mu_a}(0)D_{\mu_a}'(0)\biggr]\bigg/{D_{\mu_a}(0)^2} \nonumber\\
  &= \bE_0[Yg(Z) \mid X=x, A=a] - \mu_a(x)\,\bE_0[g(Z) \mid X=x, A=a] \nonumber \\
  &= \bE_0[(Y - \mu_a(X))\, g(Z) \mid X=x, A=a].
\end{align}

For $\pi_t(x) = p_t(x, A=1)/p_t(x)$, set
\[
  N_{\pi}(t) := \int_{\cY} dP_0(y, x, 1)(1 + t\,g(y,x,1)), \qquad
  D_{\pi}(t) := \int_{\cY} \sum_{a} dP_0(y, x, a)(1 + t\,g(y,x,a)).
\]
Then it follows
\begin{align*}
  &N_{\pi}(0) = dP_0(x, A{=}1), &N_{\pi}'(0) &= \bE_0[g(Z) \mid X = x,\, A = 1]\,dP_0(x, A{=}1),\\
  &D_{\pi}(0) = dP_0(x), &D_{\pi}'(0) &= \bE_0[g(Z) \mid X = x]\,dP_0(x).
\end{align*}
Again by the quotient rule, we obtain
\begin{align}
  \left.\frac{\partial}{\partial t}\,\pi_t(x)\right|_{t=0}
  &= \biggl[N_{\pi}'(0)D_{\pi}(0) - N_{\pi}(0)D_{\pi}'(0)\biggr]\bigg/{D_{\pi}(0)^2} \nonumber\\
  &= \pi(x)\left(\bE_0[g(Z) \mid X = x,\, A = 1] - \bE_0[g(Z) \mid X = x]\right) \nonumber\\
  &= \bE_0[(A - \pi(X)) \, g(Z) \mid X = x],
  \label{eq:pi_deriv}
\end{align}
where the second equality follows from
$\bE_0[A \,g(Z) \mid X = x] = \pi(x)\,\bE_0[g(Z) \mid X = x,\, A = 1]$.

We now assemble the terms. For term~(I), recall the identity
\[\bE_0[W \cdot \mathbf{1}\{A=1\} \mid X] = \pi(X) \cdot \bE_0[W \mid X,\, A=1].\]
By the law of total expectation,
\begin{align*}
  \int_\cX \bE_0[(Y - \mu_1(X))\, g(Z) \mid X=x, A=1]\, dP_0(x) &= \bE_0\left[\pi(X) \cdot \bE_0\left[\frac{(Y-\mu_1(X))g(Z)}{\pi(X)} \,\bigg|\, X, A=1\right]\right] \nonumber \\
  &= \bE_0\left[\frac{A\,(Y - \mu_1(X))\, g(Z)}{\pi(X)}\right].  
\end{align*}
By the same reasoning, term~(II) gives \[\bE_0\left[\frac{(1-A)(Y - \mu_0(X))\, g(Z)}{1-\pi(X)}\right].\] For term~(III), by the law of total expectation
\[
  \int_\cX \tau(x)\, \bE_0[g(Z) \mid X=x]\, dP_0(x) = \bE_0[\tau(X)\, g(Z)].
\]
Collecting all three terms shows that $\frac{d}{dt}\beta(P_{t,g})|_{t=0} = \bE_0[\varphi(Z)\, g(Z)]$ for every linear tilt submodel with bounded mean-zero score $g$. We know bounded mean-zero functions are dense in $\cT = L_2^0(P_0)$ by the same argument as Corollary~\ref{cor:saturation}, and the ATE is Hellinger Lipschitz as verified in Section~\ref{subsubsec:ate_forward} for Assumption~\ref{assumption:hellinger}. Therefore, the same three-term approximation argument used in the proof of Theorem~\ref{thm:forward} extends this identity to all regular submodels with score $s' \in \cS$, which establishes pathwise differentiability at $P_0$ with influence function $\varphi$.

\textbf{Assumption~\ref{assumption:product_reverse}.}

To verify this assumption, we construct explicit QMD submodels along each coordinate of the parameter space.

\textbf{\(\beta\)-coordinate submodel.}
We construct a QMD path \(t \mapsto P_t\) with \(\beta(P_t) = \beta_0 + t\) and \(\eta(P_t) = \eta_0\). Define the function
\[g_\beta(x) := \frac{\tau(x) - \beta_0}{\mathrm{Var}_0(\tau(X))},\]
which depends on \(z\) only through \(x\). By (R4), \(\mathrm{Var}_0(\tau(X)) > 0\), and (R2) gives \(\|g_\beta\|_\infty < \infty\). Clearly, \(\bE_0[g_\beta] = 0\). By Lemma~\ref{lem:linear_tilt_submodel}, the linear tilt \(dP_t(z) = (1 + t g_\beta(x))\, dP_0(z)\) defines a regular QMD submodel through \(P_0\) with score \(g_\beta\) for \(|t| < 1/\|g_\beta\|_\infty\). Since \(g_\beta\) depends only on \(x\), the conditional densities are undisturbed by the tilt:
\[
  dP_t(y \mid x,\, a) = \frac{dP_t(y,x,a)}{dP_t(x,a)} = \frac{(1 + t g_\beta(x))\,dP_0(y,x,a)}{(1 + t g_\beta(x))\,dP_0(x,a)} = dP_0(y \mid x,\, a),
\]
so \(\mu_{a,t}(x) = \mu_a(x)\) for all small \(t\). Similarly, \(\pi_t(x) = dP_t(x,1)/dP_t(x) = \pi(x)\) since the \((1 + t g_\beta(x))\) factors cancel in the ratio. Hence \(\eta(P_t) = \eta_0\).

Furthermore, \(\beta\) increases at unit rate:
\begin{align*}
  \beta(P_t) &= \bE_{P_t}[\tau(X)] = \bE_0[\tau(X)(1 + t g_\beta(X))] = \beta_0 + t \cdot \frac{\bE_0[\tau(X)\,(\tau(X) - \beta_0)]}{\mathrm{Var}_0(\tau(X))} = \beta_0 + t,
\end{align*}
since \(\bE_0[\tau(X)(\tau(X) - \beta_0)] = \mathrm{Var}_0(\tau(X))\).

\textbf{\(\eta\)-coordinate submodels.}

For each admissible direction \(h = (h_1, h_0, h_\pi) \in \dot{\cH}\), we construct a regular (QMD) submodel through \(P_0\) satisfying \(\dot\beta_{0,s_h} = 0\) and \(\dot\eta_{0,s_h} = h\). Since Assumption~\ref{assumption:product-structure} requires only first-order coordinate control, a linear tilt submodel suffices.

Define the perturbation functions
\[
  g_a(y, x) := \frac{h_a(x)\,(y - \mu_a(x))}{\sigma_a^2(x)}, \quad
  g_{h_\pi}(x, a) := \frac{h_\pi(x)\,(a - \pi(x))}{\pi(x)(1 - \pi(x))}, \quad
  g_\beta(x) := \frac{\tau(x) - \beta_0}{\mathrm{Var}_0(\tau(X))},
\]
and the score
\begin{equation}\label{eq:eta_score}
  s_h(z) := g_a(y, x) + g_{h_\pi}(x, a) + \alpha_0\, g_\beta(x), \qquad \alpha_0 := -\bE_0[h_1(X) - h_0(X)].
\end{equation}
Under (R1)--(R3), each summand is bounded. Indeed, \(|g_a| \le \|h_a\|_\infty \cdot 2C_Y / \sigma^2\), \(|g_{h_\pi}| \le \|h_\pi\|_\infty / (\varepsilon(1 - \varepsilon))\), and \(|g_\beta| \le 4C_Y / \mathrm{Var}_0(\tau(X))\). Each summand also has mean zero. For the outcome perturbation, the law of iterated expectation and \(\bE_0[Y - \mu_a(X) \mid X, A] = 0\) give \(\bE_0[g_a] = 0\). For the propensity perturbation, \(\bE_0[A - \pi(X) \mid X] = 0\) gives \(\bE_0[g_{h_\pi}] = 0\). Finally, \(\bE_0[g_\beta] = 0\) by construction. Hence \(s_h\) is bounded and mean-zero, and by Lemma~\ref{lem:linear_tilt_submodel}, the linear tilt
\[
  p_t(z) := p_0(z)(1 + t\, s_h(z))
\]
defines a regular QMD submodel through \(P_0\) with score \(s_h\) for \(|t| < 1/\|s_h\|_\infty\).

We now verify the first-order coordinate derivatives using the quotient rule formulas~\eqref{eq:mu_a_deriv} and~\eqref{eq:pi_deriv}, applied to the linear tilt with score \(s_h\).

\textbf{Derivative of \(\mu_a\).}
By~\eqref{eq:mu_a_deriv},
\[
  \dot\mu_{a,0}(x) = \bE_0[(Y - \mu_a(X))\, s_h(Z) \mid X = x,\, A = a].
\]
We expand \(s_h = g_a + g_{h_\pi} + \alpha_0\, g_\beta\) and compute each contribution separately. For the outcome perturbation,
\begin{align*}
  \bE_0[(Y - \mu_a(X))\, g_a(Y, X) \mid X = x,\, A = a]
  &= \frac{h_a(x)}{\sigma_a^2(x)}\, \bE_0\left[(Y - \mu_a(X))^{2} \mid X = x,\, A = a\right] \\
  &= \frac{h_a(x)}{\sigma_a^2(x)} \cdot \sigma_a^2(x) = h_a(x).
\end{align*}
For the propensity perturbation, since \(g_{h_\pi}(x, a)\) does not depend on \(y\), it factors out of the conditional expectation and the remaining factor \(\bE_0[Y - \mu_a(X) \mid X = x,\, A = a]\) vanishes by definition of \(\mu_a\). The same reasoning applies to \(\alpha_0\, g_\beta(x)\), which also does not depend on \(y\). That is,
\begin{align*}
  \bE_0[(Y - \mu_a(X))\, g_{h_\pi}(X, A) \mid X = x,\, A = a] &= g_{h_\pi}(x, a) \cdot \bE_0[Y - \mu_a(X) \mid X = x,\, A = a] = 0, \\
  \bE_0[(Y - \mu_a(X))\, \alpha_0\, g_\beta(X) \mid X = x,\, A = a] &= \alpha_0\, g_\beta(x) \cdot \bE_0[Y - \mu_a(X) \mid X = x,\, A = a] = 0.
\end{align*}
Combining the three contributions gives \(\dot\mu_{a,0}(x) = h_a(x)\).

\textbf{Derivative of \(\pi\).}
By~\eqref{eq:pi_deriv},
\[
  \dot\pi_0(x) = \bE_0[(A - \pi(X))\, s_h(Z) \mid X = x].
\]
For the outcome perturbation, we condition on \(A\) and use the conditional mean-zero property of \(g_a\) to obtain
\begin{align*}
  \bE_0[(A - \pi(X))\, g_a(Y, X) \mid X = x]
  &= \sum_{a' \in \{0,1\}} P_0(A = a' \mid x)\,(a' - \pi(x))\, \bE_0[g_{a'}(Y, X) \mid X = x,\, A = a'].
\end{align*}
Each inner expectation evaluates to
\[
  \bE_0[g_{a'}(Y, X) \mid X = x,\, A = a'] = \frac{h_{a'}(x)}{\sigma_{a'}^2(x)}\, \bE_0[Y - \mu_{a'}(X) \mid X = x,\, A = a'] = 0,
\]
so the entire sum vanishes. For the propensity perturbation, recalling that \(g_{h_\pi}(x, a) = h_\pi(x)(a - \pi(x)) / [\pi(x)(1-\pi(x))]\),
\begin{align*}
  \bE_0[(A - \pi(X))\, g_{h_\pi}(X, A) \mid X = x]
  &= \frac{h_\pi(x)}{\pi(x)(1-\pi(x))}\, \bE_0\left[(A - \pi(X))^{2} \mid X = x\right] \\
  &= \frac{h_\pi(x)}{\pi(x)(1-\pi(x))} \cdot \pi(x)(1-\pi(x)) = h_\pi(x).
\end{align*}
For the marginal correction, since \(\alpha_0\, g_\beta(x)\) does not depend on \(a\), it factors out and the remaining expectation vanishes:
\[
  \bE_0[(A - \pi(X))\, \alpha_0\, g_\beta(X) \mid X = x] = \alpha_0\, g_\beta(x) \cdot \bE_0[A - \pi(X) \mid X = x] = 0.
\]
Combining the three contributions gives \(\dot\pi_0(x) = h_\pi(x)\), and therefore \(\dot\eta_{0, s_h} = (h_1, h_0, h_\pi) = h\).

\textbf{Derivative of \(\beta\).}
By Assumption~\ref{assumption:pathwise},
\[
  \dot\beta_{0, s_h} = \bE_0[\varphi(Z)\, s_h(Z)].
\]
Expanding \(\varphi\) from~\eqref{eq:ate_phi} and using linearity of expectation, this becomes
\begin{align}
\label{eq:beta_dot_expansion}
  \bE_0[\varphi(Z)\; s_h(Z)]
  = & \, \bE_0\!\left[\frac{A\,(Y - \mu_1(X))\, s_h(Z)}{\pi(X)}\right]
  - \bE_0\!\left[\frac{(1-A)(Y - \mu_0(X))\, s_h(Z)}{1-\pi(X)}\right] \nonumber \\
  + & \; \bE_0[\tau(X)\, s_h(Z)]
  - \beta_0\, \bE_0[s_h(Z)].
\end{align}
The last term vanishes since \(s_h \in L_2^0(P_0)\). We evaluate the remaining three terms in order.

For the first term, the identity \(\bE_0[A \cdot f(Z) \mid X] = \pi(X)\,\bE_0[f(Z) \mid X, A=1]\) allows us to write
\begin{align*}
  \bE_0\!\left[\frac{A\,(Y - \mu_1(X))\, s_h(Z)}{\pi(X)}\right]
  &= \bE_0\!\left[\bE_0\!\left[\frac{A\,(Y - \mu_1(X))\, s_h(Z)}{\pi(X)} \,\middle|\, X\right]\right] \\
  &= \bE_0\!\left[\bE_0\!\left[(Y - \mu_1(X))\, s_h(Z) \mid X,\, A = 1\right]\right] \\
  &= \bE_0\bigl[\dot\mu_{1,0}(X)\bigr]
  = \bE_0[h_1(X)],
\end{align*}
where the penultimate equality uses~\eqref{eq:mu_a_deriv} and the final equality uses \(\dot\mu_{1,0}(x) = h_1(x)\) as established above. The second term in~\eqref{eq:beta_dot_expansion} follows by the same argument with \(a = 0\), giving \(\bE_0[h_0(X)]\).

For the third term, we apply the tower property to condition on \(X\):
\[
  \bE_0[\tau(X)\, s_h(Z)] = \bE_0\!\left[\tau(X) \cdot \bE_0[s_h(Z) \mid X]\right].
\]
To evaluate the inner conditional expectation, we treat each component of \(s_h\) separately. For \(g_a\), identical to the above, conditioning further on \(A\) gives
\begin{align*}
  \bE_0[g_a(Y, X) \mid X = x]
  &= \sum_{a' \in \{0,1\}} P_0(A = a' \mid x)\, \bE_0[g_{a'}(Y, X) \mid X = x,\, A = a'] \\
  &= \sum_{a' \in \{0, 1\}} P_0(A = a' \mid x) \cdot \frac{h_{a'}(x)}{\sigma_{a'}^2(x)}\, \bE_0[Y - \mu_{a'}(X) \mid X = x,\, A = a'] = 0,
\end{align*}
since the conditional mean of \(Y - \mu_{a'}(X)\) vanishes by definition. For \(g_{h_\pi}\),
\[
  \bE_0[g_{h_\pi}(X, A) \mid X = x] = \frac{h_\pi(x)}{\pi(x)(1-\pi(x))}\, \bE_0[A - \pi(X) \mid X = x] = 0.
\]
Since \(\alpha_0\, g_\beta(x)\) is already a function of \(x\) alone, it passes through the conditional expectation unchanged. Combining these three observations,
\[
  \bE_0[s_h(Z) \mid X = x] = 0 + 0 + \alpha_0\, g_\beta(x) = \alpha_0\, g_\beta(x).
\]
Substituting back and using the identity \(\bE_0[\tau(X)\, g_\beta(X)] = 1\) (established for the \(\beta\)-coordinate submodel),
\[
  \bE_0[\tau(X)\, s_h(Z)] = \bE_0[\tau(X) \cdot \alpha_0\, g_\beta(X)] = \alpha_0 \cdot \bE_0[\tau(X)\, g_\beta(X)] = \alpha_0.
\]
Collecting the three terms of~\eqref{eq:beta_dot_expansion},
\begin{align*}
  \dot\beta_{0, s_h}
  &= \bE_0[h_1(X)] - \bE_0[h_0(X)] + \alpha_0 \\
  &= \bE_0[h_1(X) - h_0(X)] + \left(-\bE_0[h_1(X) - h_0(X)]\right) = 0.
\end{align*}

Combined with the \(\beta\)-coordinate submodel above, we see that Assumption~\ref{assumption:product-structure} holds for the average treatment effect.

\begin{remark}[When no marginal correction is needed]
\label{remark:no_correction}
When \(\bE_0[h_1(X) - h_0(X)] = 0\), which holds for instance when \(h_1 = h_0\) pointwise or whenever the outcome perturbations are mean-balanced across treatment arms, we have \(\alpha_0 = 0\) and the score simplifies to \(s_h = g_a + g_{h_\pi}\). The marginal correction \(\alpha_0\, g_\beta\) is driven entirely by the imbalance \(\bE_0[h_1 - h_0]\) of the outcome perturbations.
\end{remark}

\textbf{Assumption~\ref{assumption:regularity_coordinate}.} Under (R1)--(R2), \(\|\varphi\|_\infty < \infty\) since \(\mu_a\) and \(\pi\) are bounded and \(\pi\) is bounded away from \(0\) and \(1\). The scores \(s_\beta\) and \(s_h\) are bounded by construction, so each coordinate submodel is a linear tilt with bounded score and the verification of Lemma~\ref{lem:differentiation_varying_f} proceeds identically to the forward direction.

Since all assumptions of Theorem~\ref{thm:reverse} hold, with
local product structure
(Assumption~\ref{assumption:product_reverse}) established by
the explicit coordinate submodel constructions above,
Theorem~\ref{thm:reverse} gives that \(m\) is Neyman orthogonal
with \(G = \bE_0[\partial_\beta m(Z;\, \beta_0,\, \eta_0)] = -1\).

\subsection{Partially Linear Model}
\label{appendix:plm}

We now illustrate the equivalence results of Section~\ref{sec:equivalence} through a detailed worked example on estimating the slope in the partially linear model. Unlike the ATE example, the tangent space in the partially linear model is a proper subspace of \(L_2^0(P_0)\), where generic linear tilts of the full density leave the model. Therefore, we construct the needed regular submodels explicitly, characterize the geometry of the tangent space, and verify every assumption for both directions of the equivalence. We recover the known result that the classical residual-on-residual moment is Neyman orthogonal but generally inefficient, with the inefficiency gap attributed to a component lying in \(\cT^\perp\).

\paragraph{Setup.}
We work with the partially linear regression model of \citet{robinson_root-n-consistent_1988}. Let \(Z = (Y, D, X)\), where \(D\) is a scalar treatment or exposure variable and \(X\) is an arbitrary covariate vector. We assume all distributions under consideration are dominated by a fixed \(\sigma\)-finite product measure
\[
\nu = \nu_Y \otimes \nu_D \otimes \nu_X,
\]
with density factorization
\[
p(y, d, x) = p_X(x)\, p_{D \mid X}(d \mid x)\, q(y \mid d, x).
\]
The partially linear model assumes that for each distribution $P$ in the model,
\[
\bE_P[Y \mid D, X] = \ell_P(X) + \theta(P)(D - m_P(X)),
\]
where the target parameter $\theta(P)$ is a scalar, $m_P(X) = \bE_P[D \mid X]$, and $\ell_P(X) = \bE_P[Y \mid X]$. Under the true distribution $P_0$, write $\theta_0 := \theta(P_0)$, $\ell_0 := \ell_{P_0}$, $m_0 := m_{P_0}$, and define
\[
V := D - m_0(X), \qquad U := Y - \ell_0(X) - \theta_0 V.
\]
Then $Y = \ell_0(X) + \theta_0 V + U$. The model restriction also gives $\bE_0[U \mid D, X] = 0$, and the definition of \(m_0\) yields $\bE_0[V \mid X] = 0$. 

\begin{assumption}[Regularity conditions for the PLM]
\label{ass:plm_standing}
There exist finite constants \(C_Y, C_D > 0\) and positive constants \(\underline\sigma^2, \overline\sigma^2, \underline\tau^2, \overline\tau^2\) such that 
\[
|Y| \le C_Y, \qquad |D| \le C_D,
\]
and the conditional variances
\[
\sigma_0^2(D, X) := \bE_0\bigl[U^{2} \mid D, X\bigr], \qquad \tau_0^2(X) := \bE_0\bigl[V^{2} \mid X\bigr]
\]
satisfy
\[
0 < \underline\sigma^2 \le \sigma_0^2(D, X) \le \overline\sigma^2 < \infty, \qquad 0 < \underline\tau^2 \le \tau_0^2(X) \le \overline\tau^2 < \infty.
\]
In particular,
\[
J_0 := \bE_0\bigl[V^{2}\bigr] > 0.
\]
\end{assumption}

The nuisance functionals are
\[
\eta(P) := (\ell_P, m_P), \qquad \eta_0 = (\ell_0, m_0),
\]
and the nuisance parameter set is
\[
\cH := L_\infty(P_{0,X}) \times L_\infty(P_{0,X})
\]
endowed with the product supremum norm. The admissible perturbation space is given by \(\dot \cH = L_\infty(P_{0, X}) \times L_\infty(P_{0, X})\).

\subsubsection{Submodel Constructions}
\label{subsubsec:plm_submodels}

In the partially linear model, a generic linear tilt of the full density can violate the restriction that the conditional mean of \(Y\) is linear in \(D\). Therefore, we construct the needed regular submodels explicitly in the section below. We first state the following lemma, which will be used later.

\begin{lemma}[QMD implies an \(L_1\) first-order density expansion]
\label{lem:plm_qmd_l1}
Suppose \(t \mapsto p_t\) is QMD at \(0\) with score \(s \in L_2^0(P_0)\), so that
\[
\sqrt{p_t} = \sqrt{p_0}\left(1 + \frac{t}{2}s\right) + r_t, \qquad \|r_t\|_{L_2(\nu)} = o(|t|).
\]
Let \(\Lambda_t := p_t / p_0\). Then
\begin{equation}
\label{eq:plm_l1_expansion}
\|\Lambda_t - 1 - ts\|_{L_1(P_0)} = o(|t|).
\end{equation}
Hence, for every bounded measurable \(f\),
\begin{equation}
\label{eq:plm_bounded_diff}
\bE_{P_t}[f(Z)] = \bE_0[f(Z)] + t\,\bE_0[f(Z)s(Z)] + o(|t|).
\end{equation}
\end{lemma}

\begin{proof}
Define $A := \{p_0 > 0\}$ and
\[
u_t := \mathbf{1}_A \frac{r_t}{\sqrt{p_0}}.
\]
Since $\int_A r_t^2\,d\nu \le \|r_t\|_{L_2(\nu)}^2 = o(t^2)$, we have $\|u_t\|_{L_2(P_0)} = o(|t|)$. On $A$, the QMD expansion reads
\[
\sqrt{\frac{p_t}{p_0}} = 1 + \frac{t}{2}s + u_t \qquad P_0\text{-a.s.},
\]
so squaring gives the $P_0$-a.s.\ identity
\[
\Lambda_t = \left(1 + \frac{t}{2}s + u_t\right)^2 = 1 + ts + 2u_t + ts\,u_t + u_t^2 + \frac{t^2}{4}s^2.
\]

By the triangle inequality,
\begin{align*}
\|\Lambda_t - 1 - ts\|_{L_1(P_0)} 
&\le 2\|u_t\|_{L_1(P_0)} + |t|\,\|s\,u_t\|_{L_1(P_0)} + \|u_t^2\|_{L_1(P_0)} + \frac{t^2}{4}\|s^2\|_{L_1(P_0)}\\
&\le 2\|u_t\|_{L_2(P_0)} + |t|\,\|s\|_{L_2(P_0)}\|u_t\|_{L_2(P_0)} + \|u_t\|_{L_2(P_0)}^2 + \frac{t^2}{4}\|s\|_{L_2(P_0)}^2 \\
&= o(|t|),
\end{align*}
which establishes~\eqref{eq:plm_l1_expansion}. Next, we decompose the expectation over $A$ and its complement,
\[
\bE_{P_t}[f] = \int_A f\,p_t\,d\nu + \int_{A^c} f\,p_t\,d\nu.
\]
The first integral equals $\bE_0[f\,\Lambda_t]$, so 
\[
\int_A f\,p_t\,d\nu - \bE_0[f] - t\,\bE_0[fs] = \bE_0\left[f(\Lambda_t - 1 - ts)\right],
\]
whose absolute value is at most $\|f\|_\infty\,\|\Lambda_t - 1 - ts\|_{L_1(P_0)} = o(|t|)$ by~\eqref{eq:plm_l1_expansion}. For the second integral, note that $p_0 = 0$ on $A^c$, so the QMD expansion reduces to $\sqrt{p_t} = r_t$ on $A^c$. Therefore,
\[
\biggl|\int_{A^c} f\,p_t\,d\nu\biggr| \le \|f\|_\infty \int_{A^c} r_t^2\,d\nu \le \|f\|_\infty\,\|r_t\|_{L_2(\nu)}^2 = o(t^2).
\]
Combining the two bounds gives~\eqref{eq:plm_bounded_diff}.
\end{proof}

\paragraph{Perturbing the Marginal Distribution of \(X\).}

\begin{proposition}[Marginal \(X\)-submodels]
\label{prop:plm_X_submodel}
Let \(s_X \in L_\infty(P_{0,X})\) with \(\bE_0[s_X(X)] = 0\). For \(|t| < \|s_X\|_\infty^{-1}/2\), define
\[
p_t(y, d, x) := p_0(y, d \mid x)\, p_{0,X}(x)(1 + ts_X(x)).
\]
Then,
\begin{enumerate}[label=(\roman*)]
\item \(p_t\) is a density and the corresponding distribution remains in the partially linear model.
\item The perturbation preserves the conditional distributions \(p_t(d \mid x) = p_0(d \mid x)\) and \(p_t(y \mid d, x) = p_0(y \mid d, x)\).
\item Consequently,
\[
\bE_{P_t}[D \mid X] = m_0(X), \qquad \bE_{P_t}[Y \mid D, X] = \ell_0(X) + \theta_0(D - m_0(X)).
\]
\item The path is QMD with score \(s_X(X)\).
\end{enumerate}
\end{proposition}

\begin{proof}
The density property immediately follows,
\[
\int p_t\,d\nu = \int p_0(y, d \mid x)\, p_{0,X}(x)(1 + ts_X(x))\,d\nu = 1 + t\,\bE_0[s_X(X)] = 1.
\]
We can similarly compute the conditional distributions. First, we write
\[
p_t(x) = p_{0,X}(x)(1 + ts_X(x)).
\]
Thus,
\[
p_t(y, d \mid x) = \frac{p_0(y, d \mid x)\, p_{0,X}(x)(1 + ts_X(x))}{p_{0,X}(x)(1 + ts_X(x))} = p_0(y, d \mid x).
\]
Taking marginals gives
\[
p_t(d \mid x) = p_0(d \mid x), \qquad p_t(y \mid d, x) = p_0(y \mid d, x).
\]
Therefore, it follows that
\[
\bE_{P_t}[D \mid X] = \bE_0[D \mid X] = m_0(X), \qquad \bE_{P_t}[Y \mid D, X] = \bE_0[Y \mid D, X] = \ell_0(X) + \theta_0(D - m_0(X)).
\]
Finally,
\(
p_t = p_0(1 + ts_X),
\)
so QMD with score \(s_X(X)\) follows from Lemma~\ref{lem:linear_tilt_submodel}.
\end{proof}

\paragraph{Perturbing the Conditional Distribution of \(D \mid X\).}

\begin{proposition}[Conditional \(D \mid X\)-submodels]
\label{prop:plm_D_submodel}
Let \(s_D \in L_\infty(P_{0,DX})\) satisfy
\[
\bE_0[s_D(D, X) \mid X] = 0.
\]
For \(|t| < \|s_D\|_\infty^{-1}/2\), define
\[
p_t(y, d, x) := q_0(y \mid d, x)\, p_{0,D \mid X}(d \mid x)(1 + ts_D(d, x))\, p_{0,X}(x).
\]
Then:
\begin{enumerate}[label=(\roman*)]
\item \(p_t\) is a density and the corresponding distribution remains in the partially linear model.
\item The perturbation preserves both the conditional distribution \(p_t(y \mid d, x) = p_0(y \mid d, x)\) and the marginal distribution \(p_t(x) = p_{0,X}(x)\) exactly.
\item Consequently,
\[
\bE_{P_t}[Y \mid D, X] = \ell_0(X) + \theta_0(D - m_0(X)).
\]
\item The path is QMD with score \(s_D(D, X)\).
\end{enumerate}
Moreover, if \(h_m \in L_\infty(P_{0,X})\) and
\begin{equation}
\label{eq:plm_D_shift_score}
s_D(d, x) := \frac{h_m(x)(d - m_0(x))}{\tau_0^2(x)},
\end{equation}
    then the perturbed conditional mean of \(D\) is given by
\begin{equation}
\label{eq:plm_mt_exact}
m_t(x) := \bE_{P_t}[D \mid X = x] = m_0(x) + th_m(x).
\end{equation}
\end{proposition}

\begin{proof}
Because \(\bE_0[s_D \mid X] = 0\),
\begin{align*}
\int p_t\,d\nu &= \int p_{0,X}(x)\left\{\int p_{0,D \mid X}(d \mid x)(1 + ts_D(d, x))\,d\nu_D\right\}\left\{\int q_0(y \mid d, x)\,d\nu_Y\right\}\,d\nu_X \\
&= \int p_{0,X}(x)(1 + t\,\bE_0[s_D \mid X = x])\,d\nu_X = 1.
\end{align*}
Thus \(p_t\) is a density. Next,
\[
p_t(x) = p_{0,X}(x)\int p_{0,D \mid X}(d \mid x)(1 + ts_D(d, x))\,d\nu_D = p_{0,X}(x).
\]
Also,
\[
p_t(d \mid x) = p_{0,D \mid X}(d \mid x)(1 + ts_D(d, x)).
\]
Since the perturbation factor depends only on \((d, x)\),
\[
p_t(y \mid d, x) = \frac{q_0(y \mid d, x)\, p_{0,D \mid X}(d \mid x)(1 + ts_D(d, x))\, p_{0,X}(x)}{p_{0,D \mid X}(d \mid x)(1 + ts_D(d, x))\, p_{0,X}(x)} = q_0(y \mid d, x).
\]
Hence,
\[
\bE_{P_t}[Y \mid D, X] = \bE_0[Y \mid D, X] = \ell_0(X) + \theta_0(D - m_0(X)).
\]
Finally,
\(
p_t = p_0(1 + ts_D),
\)
so QMD with score \(s_D(D, X)\) follows from Lemma~\ref{lem:linear_tilt_submodel}. For the score \eqref{eq:plm_D_shift_score}, note first that
\[
\bE_0[s_D(D, X) \mid X] = \frac{h_m(X)}{\tau_0^2(X)}\,\bE_0[V \mid X] = 0,
\]
and the construction is valid. Under \(P_t\), we can write
\begin{align*}
m_t(X) 
&= \bE_{P_t}[D \mid X] \\
&= \bE_0[D(1 + ts_D(D, X)) \mid X] \\
&= m_0(X) + t\,\bE_0[D\,s_D(D, X) \mid X] \\
&= m_0(X) + t\,\frac{h_m(X)}{\tau_0^2(X)}\,\bE_0[DV \mid X].
\end{align*}
Since \(D = m_0(X) + V\) and \(\bE_0[V \mid X] = 0\),
\[
\bE_0[DV \mid X] = \bE_0\bigl[V^{2} \mid X\bigr] = \tau_0^2(X),
\]
which yields~\eqref{eq:plm_mt_exact}.
\end{proof}

\paragraph{Perturbing the Conditional Distribution of \(Y \mid D, X\).}

\begin{proposition}[Conditional \(Y \mid D, X\)-submodels]
\label{prop:plm_Y_submodel}
Let \(w \in L_\infty(P_0)\) satisfy
\[
\bE_0[w(Y, D, X) \mid D, X] = 0, \qquad \bE_0[Uw(Y, D, X) \mid D, X] = 0.
\]
For each \(|t| < \varepsilon\), let \(\mu_t(d,x)\) be a bounded measurable function such that:
\begin{enumerate}[label=(\alph*)]
\item \(\mu_0(d, x) = \ell_0(x) + \theta_0(d - m_0(x))\).
\item Each \(\mu_t\) has partially linear form,
\begin{equation}
\label{eq:plm_mu_prescribed}
\mu_t(d, x) = \ell_t(x) + \theta_t(d - m_t(x))
\end{equation}
for some measurable \(\ell_t, m_t, \theta_t\).
\item There exists a bounded measurable \(r(d, x)\) such that
\begin{equation}
\label{eq:plm_mu_derivative}
\|\mu_t - \mu_0 - tr\|_\infty = o(|t|).
\end{equation}
\end{enumerate}
Then there exists \(\varepsilon > 0\) such that for all \(|t| < \varepsilon\), one can choose a bounded measurable function \(\lambda_t(d, x)\) with \(\lambda_0 \equiv 0\) and define
\begin{equation}
\label{eq:plm_Y_submodel_density}
\begin{aligned}
q_t(y \mid d, x) &:= q_0(y \mid d, x)\,\frac{\exp\{\lambda_t(d, x)y + tw(y, d, x)\}}{M_t(d, x)}, \\
M_t(d, x) &:= \bE_0[\exp\{\lambda_t(d, x)Y + tw(Y, D, X)\} \mid D = d,\, X = x],
\end{aligned}
\end{equation}
so that:
\begin{enumerate}[label=(\roman*)]
\item \(q_t(\cdot \mid d, x)\) is a conditional density for every \((d, x)\).
\item The conditional mean equals the prescribed one exactly,
\[
\bE_{P_t}[Y \mid D = d,\, X = x] = \mu_t(d, x) \qquad \text{for all } |t| < \varepsilon.
\]
\item The full path \(p_t(y, d, x) := q_t(y \mid d, x)\, p_{0,D,X}(d, x)\) remains in the partially linear model for every \(|t| < \varepsilon\),
\item The path is QMD with score
\[
s_Y(Y, D, X) = \frac{U\,r(D, X)}{\sigma_0^2(D, X)} + w(Y, D, X).
\]
\end{enumerate}
\end{proposition}

\begin{proof}
Write \(\mu_0(d, x) = \ell_0(x) + \theta_0(d - m_0(x))\) and \(U = Y - \mu_0(D, X)\). For fixed \((d, x)\) define
\[
\widetilde{M}(\lambda, t; d, x) := \bE_0\bigl[e^{\lambda U + tw} \mid D = d,\, X = x\bigr],
\]
so that \[M_t(d, x) = e^{\lambda_t(d,x)\mu_0(d,x)}\,\widetilde{M}(\lambda_t(d,x), t; d, x).\] Also define
\[
F(\lambda, t; d, x) := \frac{\bE_0\bigl[Ye^{\lambda Y + tw} \mid D = d,\, X = x\bigr]}{\bE_0\bigl[e^{\lambda Y + tw} \mid D = d,\, X = x\bigr]} = \mu_0(d, x) + G(\lambda, t; d, x),
\]
where
\[
G(\lambda, t; d, x) := \frac{\bE_0\bigl[Ue^{\lambda U + tw} \mid D = d,\, X = x\bigr]}{\bE_0\bigl[e^{\lambda U + tw} \mid D = d,\, X = x\bigr]}.
\]
Since \(U\) and \(w\) are bounded, all these conditional moment functions are finite and jointly continuous in \((\lambda, t)\), uniformly over \((d, x)\) on compact neighborhoods of \((0, 0)\).

For fixed \((d, x)\) and \(t\), the map \(\lambda \mapsto F(\lambda, t; d, x)\) is differentiable with
\[
\partial_\lambda F(\lambda, t; d, x) = \mathrm{Var}_{\lambda, t}(Y \mid d, x),
\]
where \(\mathrm{Var}_{\lambda, t}\) denotes variance under the tilted conditional distribution proportional to \(e^{\lambda Y + tw}\,q_0(y \mid d, x)\). At \((\lambda, t) = (0, 0)\),
\[
\partial_\lambda F(0, 0; d, x) = \mathrm{Var}_0(Y \mid d, x) = \sigma_0^2(d, x) \ge \underline\sigma^2.
\]
By continuity and the lower variance bound, there exists \(\rho > 0\) such that whenever \(|\lambda|, |t| \le \rho\),
\begin{equation}
\label{eq:plm_slope_bound}
\partial_\lambda F(\lambda, t; d, x) \ge \frac{\underline\sigma^2}{2} \qquad \text{for all } (d, x).
\end{equation}
Hence \(F(\cdot, t; d, x)\) is strictly increasing on \([-\rho, \rho]\). Next,
\[
F(0, t; d, x) - \mu_0(d, x) = \frac{\bE_0\bigl[Ue^{tw} \mid d, x\bigr]}{\bE_0\bigl[e^{tw} \mid d, x\bigr]}.
\]
Because \(\bE_0[U \mid d, x] = 0\), \(\bE_0[Uw \mid d, x] = 0\), and \(U, w\) are bounded, a second-order Taylor expansion gives a constant \(C_1 < \infty\) such that, uniformly in \((d, x)\),
\begin{equation}
\label{eq:plm_Fzero_bound}
|F(0, t; d, x) - \mu_0(d, x)| \le C_1 t^2.
\end{equation}
Also, by~\eqref{eq:plm_mu_derivative}, there exists \(C_\mu < \infty\) such that
\[
\|\mu_t - \mu_0\|_\infty \le C_\mu |t|
\]
for all sufficiently small \(t\). Since \(\rho\) is now fixed, we choose \(\varepsilon \le \rho\) small enough such that for \(|t| \le \varepsilon\),
\[
C_1 t^2 + C_\mu |t| \le \frac{\rho\,\underline\sigma^2}{4}.
\]
Then, by the mean value theorem and~\eqref{eq:plm_slope_bound},
\[
F(\rho, t; d, x) \ge F(0, t; d, x) + \frac{\rho\,\underline\sigma^2}{2} \ge \mu_0(d, x) - C_1 t^2 + \frac{\rho\,\underline\sigma^2}{2} \ge \mu_t(d, x),
\]
and similarly \(F(-\rho, t; d, x) \le \mu_t(d, x)\). By continuity and strict monotonicity, for each \((d, x)\) there exists a unique \(\lambda_t(d, x) \in [-\rho, \rho]\) satisfying
\begin{equation}
\label{eq:plm_lambda_def}
F(\lambda_t(d, x), t; d, x) = \mu_t(d, x).
\end{equation}
The map \((d, x) \mapsto \lambda_t(d, x)\) is measurable because
\[
\{(d, x) : \lambda_t(d, x) > c\} = \{(d, x) : F(c, t; d, x) < \mu_t(d, x)\}
\]
for every real \(c\), and the right-hand side is measurable.

By construction, \(q_t(\cdot \mid d, x)\) integrates to one. Also,
\[
\bE_{P_t}[Y \mid D = d,\, X = x] = F(\lambda_t(d, x), t; d, x) = \mu_t(d, x)
\]
by~\eqref{eq:plm_lambda_def}. Since every \(\mu_t\) has the partially linear form~\eqref{eq:plm_mu_prescribed}, the path
\[
p_t(y, d, x) := q_t(y \mid d, x)\, p_{0,D,X}(d, x)
\]
lies inside the partially linear model for all sufficiently small \(t\).

Next, we apply the mean value theorem to the function \(\lambda \mapsto F(\lambda, t; d, x)\) between \(0\) and \(\lambda_t(d, x)\):
\[
F(\lambda_t, t; d, x) - F(0, t; d, x) = \partial_\lambda F(\widetilde\lambda_t(d, x), t; d, x)\,\lambda_t(d, x)
\]
for some \(\widetilde\lambda_t(d, x)\) between \(0\) and \(\lambda_t(d, x)\). Therefore,
\[
\mu_t(d, x) - \mu_0(d, x) = \partial_\lambda F(\widetilde\lambda_t(d, x), t; d, x)\,\lambda_t(d, x) + (F(0, t; d, x) - \mu_0(d, x)).
\]
By~\eqref{eq:plm_Fzero_bound},~\eqref{eq:plm_mu_derivative}, and continuity of \(\partial_\lambda F\),
\[
\partial_\lambda F(\widetilde\lambda_t(d, x), t; d, x) \to \sigma_0^2(d, x) \qquad \text{uniformly},
\]
so using~\eqref{eq:plm_mu_derivative} again gives the uniform expansion
\begin{equation}
\label{eq:plm_lambda_expansion}
\left\|\lambda_t - t\,\frac{r}{\sigma_0^2}\right\|_\infty = o(|t|).
\end{equation}

From~\eqref{eq:plm_Y_submodel_density},
\[
\log \frac{q_t(y \mid d, x)}{q_0(y \mid d, x)} = \lambda_t(d, x)\,y + tw(y, d, x) - \log M_t(d, x).
\]
Since \(Y = \mu_0(d, x) + U\) and \(M_t = e^{\lambda_t \mu_0}\,\widetilde{M}(\lambda_t, t)\),
\[
\log \frac{q_t}{q_0} = \lambda_t U + tw - \log \widetilde{M}(\lambda_t, t).
\]
We know that \(\bE_0[U \mid d, x] = 0\) and \(\bE_0[w \mid d, x] = 0\), while \(\lambda_t = O(|t|)\) and \(w\) is bounded. Hence
\[
\widetilde{M}(\lambda_t, t) = 1 + O\bigl(t^2\bigr) \qquad \text{uniformly in } (d, x),
\]
so
\[\log \frac{q_t}{q_0} = \lambda_t U + tw + O\bigl(t^2\bigr) = t\left(\frac{Ur}{\sigma_0^2} + w\right) + o(|t|)\]
uniformly, where the last step uses~\eqref{eq:plm_lambda_expansion}. Finally, define
\[
c(Y, D, X) := \frac{Ur(D, X)}{\sigma_0^2(D, X)} + w(Y, D, X).
\]
Since \(c\) is bounded, the Taylor bound \(\bigl|e^{u/2} - 1 - u/2\bigr| \le Cu^2\) for small \(u\) implies that
\[
\sqrt{\frac{q_t}{q_0}} = 1 + \frac{t}{2}c + o(|t|)
\]
uniformly. Therefore,
\[
\left\|\sqrt{q_t} - \sqrt{q_0} - \frac{t}{2}c\sqrt{q_0}\right\|_{L_2(\nu_Y \otimes P_{0,DX})} = o(|t|),
\]
which gives QMD with score \(c\).
\end{proof}

\paragraph{Assembling Bounded-score Submodels.}

\begin{proposition}[Bounded score class]
\label{prop:plm_bounded_submodels}
Let bounded measurable functions \(s_X, s_D, a, w\) and a scalar \(b \in \RR\) satisfy
\[
\bE_0[s_X(X)] = 0, \qquad \bE_0[s_D(D, X) \mid X] = 0, \qquad \bE_0[w \mid D, X] = 0, \qquad \bE_0[Uw \mid D, X] = 0.
\]
Define
\begin{equation}
\label{eq:plm_k_def}
k(X) := \bE_0[Vs_D(D, X) \mid X].
\end{equation}
Let
\[
\theta_t := \theta_0 + tb, \qquad m_t(X) := m_0(X) + tk(X), \qquad \ell_t(X) := \ell_0(X) + t(a(X) + \theta_0 k(X)),
\]
and set
\begin{equation}
\label{eq:plm_mu_combined}
\mu_t(D, X) := \ell_t(X) + \theta_t(D - m_t(X)).
\end{equation}
Then there exists \(\varepsilon > 0\) such that for all \(|t| < \varepsilon\), one can define a distribution \(P_t\) in the partially linear model with density
\[
p_t(y, d, x) = p_{0,X}(x)(1 + ts_X(x)) \cdot p_{0,D \mid X}(d \mid x)(1 + ts_D(d, x)) \cdot q_t(y \mid d, x),
\]
where \(q_t\) is the conditional density supplied by Proposition~\ref{prop:plm_Y_submodel} with mean path~\eqref{eq:plm_mu_combined}. The resulting path is QMD with score
\begin{equation}
\label{eq:plm_bounded_score}
s(Z) = s_X(X) + s_D(D, X) + \frac{U(a(X) + bV)}{\sigma_0^2(D, X)} + w(Y, D, X).
\end{equation}
\end{proposition}

\begin{proof}
First note that, by definition of \(k\),
\[
\bE_0\left[k(X)^2\right] \le \bE_0\left[V^2 s_D(D, X)^2\right] \le \|s_D\|_\infty^2\,\bE_0\bigl[V^2\bigr] < \infty,
\]
and because \(V\) and \(s_D\) are bounded, so is \(k\). The choice~\eqref{eq:plm_k_def} also implies that the \(D \mid X\) perturbation changes the conditional mean of \(D\) to \(m_t(X) = m_0(X) + tk(X)\), where the \(X\)-tilt perturbs only the marginal of \(X\).

Next, the prescribed conditional mean path is partially linear for every \(t\) by construction. Moreover,
\begin{align*}
\mu_t(D, X) &= \ell_0(X) + t(a(X) + \theta_0 k(X)) + (\theta_0 + tb)(V - tk(X)) \\
&= \mu_0(D, X) + t(a(X) + bV) - t^2 bk(X).
\end{align*}
Hence
\[
\|\mu_t - \mu_0 - t(a + bV)\|_\infty \le |t|^2 |b|\,\|k\|_\infty = o(|t|),
\]
so Proposition~\ref{prop:plm_Y_submodel} applies with
\[
r(D, X) = a(X) + bV.
\]
It follows that the \(Y \mid D, X\) perturbation is QMD with score
\[
s_Y(Y, D, X) = \frac{U(a(X) + bV)}{\sigma_0^2(D, X)} + w(Y, D, X)
\]
and that the full conditional mean remains  equal to~\eqref{eq:plm_mu_combined}.

It remains to verify QMD for the full product path. Let
\[
r_{X,t} := \sqrt{1 + ts_X} - 1 - \frac{t}{2}s_X, \qquad r_{D,t} := \sqrt{1 + ts_D} - 1 - \frac{t}{2}s_D,
\]
and let \(r_{Y,t}\) be defined by
\[
\sqrt{q_t/q_0} = 1 + \frac{t}{2}s_Y + r_{Y,t}.
\]
By Lemma~\ref{lem:linear_tilt_submodel},
\[
\|r_{X,t}\|_\infty = O\bigl(t^2\bigr), \qquad \|r_{D,t}\|_\infty = O\bigl(t^2\bigr),
\]
and by Proposition~\ref{prop:plm_Y_submodel},
\[
\|r_{Y,t}\|_\infty = o(|t|).
\]
Therefore,
\[
\sqrt{p_t / p_0} = \left(1 + \frac{t}{2}s_X + r_{X,t}\right)\left(1 + \frac{t}{2}s_D + r_{D,t}\right)\left(1 + \frac{t}{2}s_Y + r_{Y,t}\right).
\]
Expanding the product,
\[
\sqrt{p_t / p_0} = 1 + \frac{t}{2}(s_X + s_D + s_Y) + \rho_t,
\]
where every term in \(\rho_t\) is either one of the remainders \(r_{X,t}, r_{D,t}, r_{Y,t}\) or a product of quantities each of order \(O(|t|)\) or \(o(|t|)\). Since all scores are bounded, it follows that \(\|\rho_t\|_\infty = o(|t|)\), hence also \(\|\rho_t\|_{L_2(P_0)} = o(|t|)\), which gives QMD with score~\eqref{eq:plm_bounded_score}.
\end{proof}

Proposition~\ref{prop:plm_bounded_submodels} constructs, for each valid choice of bounded ingredients, an explicit QMD submodel inside the partially linear model with the corresponding score~\eqref{eq:plm_bounded_score}. The following two lemmas show that the class of scores achievable in this way is dense in the full set of scores permitted by the partially linear constraint, which will allow us to identify the tangent space in Proposition~\ref{prop:plm_tangent_space}.

\begin{lemma}[Density of the bounded score class]
\label{lem:plm_bounded_density}
Let 
\begin{equation}
\label{eq:plm_A_def}
\cA_{\mathrm{plm}} := \{a(X) + bV : a \in L_2(P_{0,X}),\; b \in \RR\} \subset L_2(P_{0,DX}).
\end{equation}
Define
\[
\cT_* := \{s \in L_2^0(P_0) : \bE_0[Us \mid D, X] \in \cA_{\mathrm{plm}}\}.
\]
Let \(\cT_b\) denote the subset of \(\cT_*\) consisting of all scores of the form
\[
s = s_X + s_D + \frac{Uc}{\sigma_0^2} + w, \qquad c = a + bV \in \cA_{\mathrm{plm}},
\]
where \(s_X, s_D, a, w\) are bounded, \(b \in \RR\), and
\[
\bE_0[s_X(X)] = 0, \qquad \bE_0[s_D(D, X) \mid X] = 0, \qquad \bE_0[w \mid D, X] = 0, \qquad \bE_0[Uw \mid D, X] = 0.
\]
Then \(\cT_b\) is dense in \(\cT_*\) in \(L_2(P_0)\).
\end{lemma}

\begin{proof}
Take any \(s \in \cT_*\). We can write
\begin{equation}
\label{eq:plm_general_decomp}
s = s_X + s_D + \frac{Uc}{\sigma_0^2} + w,
\end{equation}
where
\[
s_X := \bE_0[s \mid X], \qquad s_D := \bE_0[s \mid D, X] - \bE_0[s \mid X], \qquad c := \bE_0[Us \mid D, X],
\]
and
\[
w := s - s_X - s_D - \frac{Uc}{\sigma_0^2}.
\]
Then \(c \in \cA_{\mathrm{plm}}\) by assumption, \(\bE_0[s_X] = 0\), \(\bE_0[s_D \mid X] = 0\), and
\[
\bE_0[w \mid D, X] = 0, \qquad \bE_0[Uw \mid D, X] = 0.
\]
Since \(c \in \cA_{\mathrm{plm}}\), there exist \(a \in L_2(P_{0,X})\) and \(b \in \RR\) such that \(c = a + bV\).

We now approximate each term in~\eqref{eq:plm_general_decomp} by bounded terms preserving the defining orthogonality constraints. For the \(X\)-part, define
\[
s_{X,n} := \mathrm{trunc}_n(s_X) - \bE_0[\mathrm{trunc}_n(s_X)],
\]
where \(\mathrm{trunc}_n(u) := \max\{-n, \min\{u, n\}\}\). Then \(s_{X,n}\) is bounded, mean zero, and \(s_{X,n} \to s_X\) in \(L_2\). For the \(D \mid X\)-part, define
\[
\widetilde{s}_{D,n} := \mathrm{trunc}_n(s_D), \qquad s_{D,n} := \widetilde{s}_{D,n} - \bE_0\bigl[\widetilde{s}_{D,n} \mid X\bigr].
\]
Then \(s_{D,n}\) is bounded, satisfies \(\bE_0[s_{D,n} \mid X] = 0\), and, since conditional expectation is an \(L_2\) contraction,
\[
\|s_{D,n} - s_D\|_2 \le \bigl\|\widetilde{s}_{D,n} - s_D\bigr\|_2 + \left\|\bE_0\bigl[\widetilde{s}_{D,n} - s_D \mid X\bigr]\right\|_2 \le 2\bigl\|\widetilde{s}_{D,n} - s_D\bigr\|_2 \to 0.
\]

For the \(c\)-part, define \(a_n := \mathrm{trunc}_n(a)\) and \(c_n := a_n + bV\). Since \(V\) is bounded, \(c_n\) is bounded and \(c_n \to c\) in \(L_2(P_{0,DX})\). Hence
\[
\left\|\frac{Uc_n}{\sigma_0^2} - \frac{Uc}{\sigma_0^2}\right\|_2 \le \frac{\|U\|_\infty}{\underline\sigma^2}\,\|c_n - c\|_2 \to 0.
\]

For the \(w\)-part, define \(\widetilde{w}_n := \mathrm{trunc}_n(w)\) and 
\[
\alpha_n := \bE_0\bigl[\widetilde{w}_n \mid D, X\bigr], \qquad \beta_n := \frac{\bE_0\bigl[U\widetilde{w}_n \mid D, X\bigr]}{\sigma_0^2(D, X)}, \qquad w_n := \widetilde{w}_n - \alpha_n - U\beta_n.
\]
Then \(w_n\) is bounded and satisfies
\[
\bE_0[w_n \mid D, X] = 0, \qquad \bE_0[Uw_n \mid D, X] = 0.
\]
Moreover, because \(\bE_0[w \mid D, X] = 0\) and \(\bE_0[Uw \mid D, X] = 0\),
\[
\alpha_n = \bE_0\bigl[\widetilde{w}_n - w \mid D, X\bigr], \qquad \beta_n = \frac{\bE_0\left[U\bigl(\widetilde{w}_n - w\bigr) \mid D, X\right]}{\sigma_0^2}.
\]
Thus, by conditional Jensen and conditional Cauchy--Schwarz,
\[
\|\alpha_n\|_2 \le \bigl\|\widetilde{w}_n - w\bigr\|_2 \to 0,
\]
and
\begin{align*}
\|U\beta_n\|_2^2 &= \bE_0\!\left[\sigma_0^2(D, X)\,\beta_n(D, X)^2\right] \\
&= \bE_0\!\left[\frac{\bE_0\left[U\bigl(\widetilde{w}_n - w\bigr) \mid D, X\right]^2}{\sigma_0^2(D, X)}\right] \\
&\le \bE_0\left[\bE_0\left[\bigl(\widetilde{w}_n - w\bigr)^2 \mid D, X\right]\right] \\
&= \bigl\|\widetilde{w}_n - w\bigr\|_2^2 \\
&\to 0.
\end{align*}
Hence
\[
\|w_n - w\|_2 \le \bigl\|\widetilde{w}_n - w\bigr\|_2 + \|\alpha_n\|_2 + \|U\beta_n\|_2 \to 0.
\]
Finally, we define
\[
s_n := s_{X,n} + s_{D,n} + \frac{Uc_n}{\sigma_0^2} + w_n.
\]
Since each \(s_n\) lies in \(\cT_b\), and the triangle inequality gives \(\|s_n - s\|_2 \to 0\), it follows that \(\cT_b\) is dense in \(\cT_*\).
\end{proof}

\begin{lemma}[Bounded elements of \(\cA_{\mathrm{plm}}^\perp\) are dense in \(\cA_{\mathrm{plm}}^\perp\)]
\label{lem:plm_Aperp_density}
The bounded elements of
\[
\cA_{\mathrm{plm}}^\perp = \{c(D, X) \in L_2(P_{0,DX}) : \bE_0[c \mid X] = 0,\; \bE_0[cV] = 0\}
\]
are dense in \(\cA_{\mathrm{plm}}^\perp\) in \(L_2(P_{0,DX})\).
\end{lemma}

\begin{proof}
Take any \(c \in \cA_{\mathrm{plm}}^\perp\). Let \(\widetilde{c}_n := \mathrm{trunc}_n(c)\) and define
\[
d_n := \widetilde{c}_n - \bE_0\bigl[\widetilde{c}_n \mid X\bigr] - \frac{\bE_0\bigl[\widetilde{c}_n V\bigr]}{J_0}\,V.
\]
Then each \(d_n\) is bounded. Moreover, we can verify that
\[
\bE_0[d_n \mid X] = 0, \qquad \bE_0[d_n V] = 0,
\]
so \(d_n \in \cA_{\mathrm{plm}}^\perp\). Since \(\widetilde{c}_n \to c\) in \(L_2\) and conditional expectation is an \(L_2\) contraction, we can write
\begin{align*}
\|d_n - c\|_2 
&\le \bigl\|\widetilde{c}_n - c\bigr\|_2 + \left\|\bE_0\bigl[\widetilde{c}_n - c \mid X\bigr]\right\|_2 + \frac{\left|\bE_0\left[\bigl(\widetilde{c}_n - c\bigr)V\right]\right|}{J_0}\,\|V\|_2 \\
&\le 2\bigl\|\widetilde{c}_n - c\bigr\|_2 + \frac{\|V\|_2}{J_0}\,\bigl\|\widetilde{c}_n - c\bigr\|_2\,\|V\|_2 \to 0,
\end{align*}
where the first inequality holds since \(\bE[c \mid X] = \bE_0[cV] = 0\). Thus bounded elements are dense in \(\cA_{\mathrm{plm}}^\perp\).
\end{proof}

\subsubsection{Hellinger-Lipschitz Bound for the Slope}
\label{subsubsec:plm_hellinger}

In this section, we show that the slope in the partially linear model is Hellinger Lipschitz (Assumption~\ref{assumption:hellinger}). Recall that for any distribution \(P\) in the partially linear model, we have
\[
\theta_P := \theta(P), \qquad \ell_P(X) := \bE_P[Y \mid X], \qquad m_P(X) := \bE_P[D \mid X], \qquad V_P := D - m_P(X).
\]
Multiplying \(Y = \ell_P(X) + \theta_P V_P + U_P\) by \(V_P\) on both sides and taking expectations under \(P\) gives
\[
\bE_P[V_P Y] = \bE_P[V_P \ell_P(X)] + \theta_P\,\bE_P\bigl[V_P^2\bigr] + \bE_P[V_P U_P] = \theta_P\,\bE_P\bigl[V_P^2\bigr],
\]
where the final equality holds since \(\bE_P[V_P U_P] = \bE_P[V_P\,\bE_P[U_P \mid D, X]] = 0\) and \(\bE_P[V_P \ell_P(X)] = \bE_P[\ell_P(X) \bE_P[V_P \mid X]] = 0\). 
It follows that 
\[
\theta_P = \frac{N(P)}{J(P)},
\]
where 
\begin{equation}
\label{eq:plm_N_ratio}
N(P) := \bE_P[V_P Y] = \bE_P[DY] - \bE_{P_X}[m_P(X)\ell_P(X)],
\end{equation}
and
\begin{equation}
\label{eq:plm_J_ratio}
J(P) := \bE_P\bigl[V_P^2\bigr] = \bE_P\bigl[D^2\bigr] - \bE_{P_X}\left[m_P(X)^2\right].
\end{equation}

\begin{proposition}[Hellinger-Lipschitz bound of \(\theta\)]
\label{prop:plm_Hellinger_Lipschitz}
Let \(P_0\) satisfy Assumption~\ref{ass:plm_standing}. Define
\begin{equation}
\label{eq:plm_lipschitz_constants}
\delta := \frac{J_0}{16\sqrt{2}\,C_D^2}, \qquad c_\theta := 16\sqrt{2}\,\frac{C_D C_Y}{J_0} + 64\sqrt{2}\,\frac{C_D^3 C_Y}{J_0^2}.
\end{equation}
Then, for any \(P_1, P_2\) in the partially linear model with
\[
H(P_i, P_0) \le \delta \qquad (i = 1, 2),
\]
we have
\begin{equation}
\label{eq:plm_Hellinger_Lipschitz}
|\theta(P_1) - \theta(P_2)| \le c_\theta\, H(P_1, P_2).
\end{equation}
In particular, Assumption~\ref{assumption:hellinger} holds.
\end{proposition}

\begin{proof}
Throughout, write \(p_i\) for the density of \(P_i\), \(p_{i,X}\) for the marginal density of \(X\), and
\[
m_i := m_{P_i}, \qquad \ell_i := \ell_{P_i}, \qquad \theta_i := \theta(P_i), \qquad J_i := J(P_i), \qquad N_i := N(P_i).
\]
To start, we define the \(X\)-indexed linear functionals
\[
g_i(x) := m_i(x)\,p_{i,X}(x) = \int d\, p_i(y, d, x)\,d\nu_Y\,d\nu_D,
\]
\[
h_i(x) := \ell_i(x)\,p_{i,X}(x) = \int y\, p_i(y, d, x)\,d\nu_Y\,d\nu_D.
\]
Since \(|D| \le C_D\) and \(|Y| \le C_Y\) under every distribution,
\begin{equation}
\label{eq:plm_g_h_bounds}
\|g_1 - g_2\|_{L_1(\nu_X)} \le 2C_D\,\mathrm{TV}(P_1, P_2), \qquad \|h_1 - h_2\|_{L_1(\nu_X)} \le 2C_Y\,\mathrm{TV}(P_1, P_2).
\end{equation}
Also,
\begin{equation}
\label{eq:plm_marginal_TV}
\|p_{1,X} - p_{2,X}\|_{L_1(\nu_X)} \le 2\,\mathrm{TV}(P_1, P_2).
\end{equation}
From~\eqref{eq:plm_J_ratio},
\[
|J_1 - J_2| \le \left|\bE_{P_1}\bigl[D^2\bigr] - \bE_{P_2}\bigl[D^2\bigr]\right| + \left|\bE_{P_{1,X}}\bigl[m_1^2\bigr] - \bE_{P_{2,X}}\bigl[m_2^2\bigr]\right|.
\]
The first term is bounded by \(2C_D^2\,\mathrm{TV}(P_1, P_2)\). For the second term, we use the identity
\[
m_1^2 p_{1,X} - m_2^2 p_{2,X} = (m_1 + m_2)(g_1 - g_2) - m_1 m_2(p_{1,X} - p_{2,X}).
\]
Since \(|m_i| \le C_D\),~\eqref{eq:plm_g_h_bounds} and~\eqref{eq:plm_marginal_TV} yield
\[
\bigl\|m_1^2 p_{1,X} - m_2^2 p_{2,X}\bigr\|_{L_1(\nu_X)} \le 2C_D\|g_1 - g_2\|_1 + C_D^2\|p_{1,X} - p_{2,X}\|_1 \le 6C_D^2\,\mathrm{TV}(P_1, P_2).
\]
Therefore,
\begin{equation}
\label{eq:plm_J_TV}
|J_1 - J_2| \le 8C_D^2\,\mathrm{TV}(P_1, P_2).
\end{equation}
Taking \(P_2 = P_0\) and using \(\mathrm{TV}(P_1, P_2) \le \sqrt{2}H(P_1, P_2)\),
\[
|J(P) - J_0| \le 8\sqrt{2}\,C_D^2\,H(P, P_0).
\]
Hence, if \(H(P, P_0) \le \delta\) with \(\delta\) from~\eqref{eq:plm_lipschitz_constants}, then
\begin{equation}
\label{eq:plm_J_lower}
J(P) \ge J_0 - 8\sqrt{2}\,C_D^2\,\delta = \frac{J_0}{2}.
\end{equation}
Thus every distribution in the Hellinger ball of radius \(\delta\) has denominator bounded away from zero.

By~\eqref{eq:plm_N_ratio},
\[
|N_1 - N_2| \le \left|\bE_{P_1}[DY] - \bE_{P_2}[DY]\right| + \left|\bE_{P_{1,X}}[m_1 \ell_1] - \bE_{P_{2,X}}[m_2 \ell_2]\right|.
\]
The first term is bounded by \(2C_D C_Y\,\mathrm{TV}(P_1, P_2)\). For the second term, we use the identity
\[
m_1 \ell_1 p_{1,X} - m_2 \ell_2 p_{2,X} = \ell_1(g_1 - g_2) + m_2(h_1 - h_2) - m_2 \ell_1(p_{1,X} - p_{2,X}).
\]
Since \(|\ell_i| \le C_Y\) and \(|m_i| \le C_D\), equations~\eqref{eq:plm_g_h_bounds} and~\eqref{eq:plm_marginal_TV} imply
\begin{align*}
\|m_1 \ell_1 p_{1,X} - m_2 \ell_2 p_{2,X}\|_{L_1(\nu_X)} &\le C_Y\|g_1 - g_2\|_1 + C_D\|h_1 - h_2\|_1 + C_D C_Y\|p_{1,X} - p_{2,X}\|_1 \\
&\le 6C_D C_Y\,\mathrm{TV}(P_1, P_2).
\end{align*}
Therefore,
\begin{equation}
\label{eq:plm_N_TV}
|N_1 - N_2| \le 8C_D C_Y\,\mathrm{TV}(P_1, P_2).
\end{equation}
Also, since \(|D| \le C_D\), \(|Y| \le C_Y\), \(|m_P| \le C_D\), and \(|\ell_P| \le C_Y\),
\begin{equation}
\label{eq:plm_N_uniform}
|N(P)| \le 2C_D C_Y \qquad \text{for every distribution } P.
\end{equation}

Finally, suppose now that \(H(P_i, P_0) \le \delta\) for \(i = 1, 2\). Then~\eqref{eq:plm_J_lower} gives \(J_1 \ge J_0/2\) and \(J_2 \ge J_0/2\). Therefore,
\begin{align*}
|\theta_1 - \theta_2| &= \left|\frac{N_1}{J_1} - \frac{N_2}{J_2}\right| \\
&\le \frac{|N_1 - N_2|}{J_1} + \frac{|N_2|\,|J_1 - J_2|}{J_1 J_2} \\
&\le \frac{2}{J_0}\,|N_1 - N_2| + \frac{4}{J_0^2}\,|N_2|\,|J_1 - J_2| \\
&\le \left(\frac{16C_D C_Y}{J_0} + \frac{64C_D^3 C_Y}{J_0^2}\right)\mathrm{TV}(P_1, P_2),
\end{align*}
where the last step uses~\eqref{eq:plm_N_TV},~\eqref{eq:plm_J_TV}, and~\eqref{eq:plm_N_uniform}. Again using \(\mathrm{TV}(P_1, P_2) \le \sqrt{2}H(P_1, P_2)\), 
\[
|\theta(P_1) - \theta(P_2)| \le \left(\frac{16C_D C_Y}{J_0} + \frac{64C_D^3 C_Y}{J_0^2}\right) \sqrt{2}\,H(P_1, P_2) = c_\theta\, H(P_1, P_2),
\]
which gives~\eqref{eq:plm_Hellinger_Lipschitz}.
\end{proof}

\subsubsection{Geometry of the Tangent Space}
\label{subsubsec:plm_geometry}

We now characterize the tangent space, the nuisance tangent space, their orthogonal complements, and the full class of influence functions. 

\begin{proposition}[Full tangent space]
\label{prop:plm_tangent_space}
Under Assumption~\ref{ass:plm_standing}, the tangent space at \(P_0\) is
\begin{equation}
\label{eq:plm_tangent_char}
\cT = \left\{s \in L_2^0(P_0) : \bE_0[Us \mid D, X] = a(X) + bV \text{ for some } a \in L_2(P_{0,X}),\; b \in \RR\right\}.
\end{equation}
Moreover, if \(s \in \cT\) and \(\bE_0[Us \mid D, X] = a(X) + bV\), then the pathwise derivative of \(\theta\) is \(\dot\theta_s = b.\)
\end{proposition}

\begin{proof}
Let \(\cT_* := \left\{s \in L_2^0(P_0) : \bE_0[Us \mid D, X] \in \cA_{\mathrm{plm}}\right\}\). We first prove \(\cT \subseteq \cT_*\), then the reverse inclusion. To start, we take any regular submodel \(t \mapsto P_t\) through \(P_0\) with score \(s\), and let
\[
\Lambda_t := \frac{p_t}{p_0}.
\]
Write
\[
\mu_t(D, X) := \bE_{P_t}[Y \mid D, X].
\]
Since \(P_t\) remains in the partially linear model for every \(t\),
\[
\mu_t(D, X) = \ell_t(X) + \theta_t(D - m_t(X)).
\]
Hence, for each fixed \(t \neq 0\),
\begin{equation}
\label{eq:plm_rt_in_A}
r_t(D, X) := \frac{\mu_t(D, X) - \mu_0(D, X)}{t} \in \cA_{\mathrm{plm}}.
\end{equation}
Indeed,
\[
\mu_t - \mu_0 = \{\ell_t - \ell_0 - \theta_t(m_t - m_0)\} + (\theta_t - \theta_0)V,
\]
which is of the form \(a_t(X) + b_t V\).

We now identify the weak limit of \(r_t\). Let \(c(D, X)\) be any bounded
measurable test function, and let \(A := \{p_0 > 0\}\). Next, we can write
\begin{align*}
\bE_{P_t}[cY]
&= \int c\,\mu_t\,p_t\,d\nu\\
&= \bE_0[c\,\mu_t\,\Lambda_t]
   + \int_{A^c} c\,\mu_t\,p_t\,d\nu \\
&= \bE_0[c\,\mu_t]
   + \bE_0[c\,\mu_t(\Lambda_t - 1)]
   + \int_{A^c} c\,\mu_t\,p_t\,d\nu.
\end{align*}
Similarly, 
\begin{align*}
\bE_{P_t}[cY] - \bE_0[cY] 
&= \int cY(p_t - p_0)\,d\nu\\
&= \int_A cY\frac{(p_t - p_0)}{p_0} p_0\,d\nu + \int_{A^c} cY\,p_t\,d\nu\\
&= \bE_0[cY(\Lambda_t - 1)]
  + \int_{A^c} cY\,p_t\,d\nu.
\end{align*}
Therefore, it follows that
\begin{align*}
    \bE_0[c \mu_t] - \bE_0[c \mu_0]
    &= \bE_0[c\mu_t ] - \bE_0[ cY]\\
    &= \left[\bE_{P_t}[cY] -  \bE_0[cY]\right] - \bE_0[c\,\mu_t(\Lambda_t - 1)]
   - \int_{A^c} c\,\mu_t\,p_t\,d\nu \\
    &= \bE_0[cY(\Lambda_t - 1)]
  - \bE_0[c\,\mu_t(\Lambda_t - 1)] + \int_{A^c} c(Y - \mu_t) p_t\,d\nu,
\end{align*}
where the first equality holds by the tower property and the others by substituting the values derived above. 
Dividing by \(t\),
\begin{equation}
\label{eq:plm_crt_diff}
\bE_0[c\,r_t]
= \underbrace{\bE_0\!\left[cY\,\frac{\Lambda_t - 1}{t}\right]}_{\mathrm{(I)}}
  - \underbrace{\bE_0\!\left[c\,\mu_t\,\frac{\Lambda_t - 1}{t}\right]}_{\mathrm{(II)}}
  + \underbrace{\frac{1}{t}\int_{A^c} c(Y - \mu_t)\,p_t\,d\nu}_{\mathrm{(III)}}.
\end{equation}
For term \((\mathrm{III})\), on \(A^c\) we have \(p_0 = 0\), so the QMD expansion gives \(P_t\bigl(A^c\bigr) = o\bigl(t^{2}\bigr)\). Thus,
\[
\left|\frac{1}{t}\int_{A^c} c(Y - \mu_t)\,p_t\,d\nu\right|
\le \frac{2\|c\|_\infty\,C_Y}{|t|}\,P_t\bigl(A^c\bigr)
= \frac{o\bigl(t^{2}\bigr)}{|t|} \to 0.
\]
By Lemma~\ref{lem:plm_qmd_l1}, term \((\mathrm{I})\) converges to \(\bE_0[cYs]\). For  term \((\mathrm{II})\), note that \(|\mu_t| \le C_Y\) uniformly, and the same lemma gives
\[
\left\|\frac{\Lambda_t - 1}{t} - s\right\|_{L_1(P_0)} \to 0.
\]
Hence, it follows that
\[
\bE_0\!\left[c\,\mu_t\,\frac{\Lambda_t - 1}{t}\right] - \bE_0[c\,\mu_t\,s] \to 0.
\]
Also, because \(|\mu_t - \mu_0| \le 2C_Y\),
\[
\|\mu_t - \mu_0\|_2^2 \le 2C_Y\|\mu_t - \mu_0\|_1.
\]
To bound the \(L_1\) norm, define \(F_t(d,x) := \mu_t(d,x)\,p_{t,DX}(d,x) = \int y\,p_t(y,d,x)\,d\nu_Y\), where \(p_{t,DX}\) denotes the \((D,X)\)-marginal density under \(P_t\). The identity
\[
(\mu_t - \mu_0)\,p_{0,DX} = (F_t - F_0) - \mu_t\,(p_{t,DX} - p_{0,DX})
\]
gives, after taking absolute values and integrating,
\[
\|\mu_t - \mu_0\|_{L_1(P_0)} \le \|F_t - F_0\|_{L_1(\nu_{DX})} + C_Y\,\|p_{t,DX} - p_{0,DX}\|_{L_1(\nu_{DX})}.
\]
Since \(|F_t(d,x) - F_0(d,x)| \le C_Y \int |p_t - p_0|\,d\nu_Y\), integrating over \((d,x)\) yields \(\|F_t - F_0\|_{L_1(\nu_{DX})} \le 2C_Y\,\mathrm{TV}(P_t, P_0)\). Marginalizing out \(Y\) also gives \(\|p_{t,DX} - p_{0,DX}\|_{L_1(\nu_{DX})} \le 2\,\mathrm{TV}(P_t, P_0)\). Hence,
\[
\|\mu_t - \mu_0\|_{L_1(P_0)} \le 4C_Y\,\mathrm{TV}(P_t, P_0) \to 0,
\]
since QMD implies \(\mathrm{TV}(P_t, P_0) \to 0\). It follows that \(\mu_t \to \mu_0\) in \(L_2(P_0)\), and
\[
\bE_0[c\,\mu_t\,s] \to \bE_0[c\,\mu_0\,s].
\]
Taking limits in~\eqref{eq:plm_crt_diff},
\[
\lim_{t \to 0} \bE_0[c\,r_t] = \bE_0[c(Y - \mu_0)s] = \bE_0[cUs].
\]
Now, let \(c\) range over the bounded elements of \(\cA_{\mathrm{plm}}^\perp\). Since \(r_t \in \cA_{\mathrm{plm}}\) for every \(t \neq 0\), we have \(\bE_0[c\,r_t] = 0\) for every such \(c\), and
\[
0 = \bE_0[cUs] = \bE_0[c\,\bE_0[Us \mid D, X]] \qquad \text{for all bounded } c \in \cA_{\mathrm{plm}}^\perp.
\]
By Lemma~\ref{lem:plm_Aperp_density}, bounded elements are dense in \(\cA_{\mathrm{plm}}^\perp\), so the same orthogonality holds for all of \(\cA_{\mathrm{plm}}^\perp\). Therefore
\[
\bE_0[Us \mid D, X] \in \bigl(\cA_{\mathrm{plm}}^\perp\bigr)^\perp = \cA_{\mathrm{plm}},
\]
which proves \(\cT \subseteq \cT_*\) since \(\cT_*\) is a closed linear subspace.

Since \(\bE_0[Us \mid D, X] \in \cA_{\mathrm{plm}}\), there exist \(a \in L_2(P_{0,X})\) and \(b \in \RR\) such that
\[
\bE_0[Us \mid D, X] = a(X) + bV.
\]
From~\eqref{eq:plm_rt_in_A}, write \(r_t = a_t(X) + b_t V\) with \(b_t = (\theta_t - \theta_0)/t\). Using the limit already established with the bounded test function \(c = V\), we obtain
\[
\lim_{t \to 0} \bE_0[V\,r_t] = \bE_0[VUs].
\]
Now \(\bE_0[V a_t(X)] = \bE_0[\bE_0[V \mid X] \cdot a_t(X)] = 0\), so
\[
\bE_0[V\,r_t] = b_t\,\bE_0\bigl[V^{2}\bigr].
\]
Therefore,
\[
\lim_{t \to 0} b_t J_0 = \bE_0[V\,\bE_0[Us \mid D, X]] = \bE_0[V(a(X) + bV)] = bJ_0.
\]
Since \(J_0 > 0\), this shows \(b_t \to b\), i.e.\ \(\dot\theta_s = b\).

By Lemma~\ref{lem:plm_bounded_density}, the bounded class \(\cT_b\) is dense in \(\cT_*\). Proposition~\ref{prop:plm_bounded_submodels} constructs, for every element of \(\cT_b\), an explicit regular submodel inside the partially linear model obtaining that score. Therefore \(\cT_b \subseteq \cT\). Since the tangent space is closed in \(L_2(P_0)\) by definition,
\[
\cT_* = \overline{\cT_b}^{L_2(P_0) } \subseteq \cT.
\]
Combined with the above, this yields \(\cT = \cT_*\).

\end{proof}

\begin{corollary}[Nuisance tangent space]
\label{cor:plm_nuisance_tangent}
The nuisance tangent space is
\begin{equation}
\label{eq:plm_nuisance_tangent}
\Lambda = \left\{s \in L_2^0(P_0) : \bE_0[Us \mid D, X] = a(X) \text{ for some } a \in L_2(P_{0,X})\right\}.
\end{equation}
\end{corollary}

\begin{proof}
By Proposition~\ref{prop:plm_tangent_space}, the derivative of \(\theta\) along any regular submodel with score \(s\) is the coefficient of \(V\) in \(\bE_0[Us \mid D, X]\). Nuisance scores are exactly those tangent directions with derivative zero, i.e.\ those with \(b = 0\).
\end{proof}

\begin{proposition}[Orthogonal complements]
\label{prop:plm_orthogonal_complements}
Under Assumption~\ref{ass:plm_standing},
\begin{align}
\Lambda^\perp &= \left\{Ua(D, X) : a \in L_2(P_{0, DX}),\; \bE_0[a(D, X) \mid X] = 0\right\}, \label{eq:plm_Lambda_perp} \\
\cT^\perp &= \left\{Ua(D, X) : a \in L_2(P_{0, DX}),\; \bE_0[a(D, X) \mid X] = 0,\; \bE_0[a(D, X)V] = 0\right\}. \label{eq:plm_T_perp}
\end{align}
\end{proposition}

\begin{proof}
Take any \(r \in L_2^0(P_0)\). Define
\[
\alpha(X) := \bE_0[r \mid X], \qquad \gamma(D,X) := \bE_0[r \mid D,X] - \bE_0[r \mid X], \qquad c(D,X) := \frac{\bE_0[Ur \mid D,X]}{\sigma_0^2(D,X)},
\]
and
\[
w(Y,D,X) := r - \alpha - \gamma - Uc.
\]
By conditional Jensen's inequality,
\(
\alpha \in L_2(P_{0, X})
\) and \(\gamma \in L_2(P_{0, DX}).\)
Also,
\[
c^2 = \frac{\bE_0[Ur \mid D,X]^2}{\sigma_0^4(D,X)} \le \frac{\bE_0\bigl[U^{2} \mid D,X\bigr]\,\bE_0\bigl[r^{2} \mid D,X\bigr]}{\sigma_0^4(D,X)} = \frac{\bE_0\bigl[r^{2} \mid D,X\bigr]}{\sigma_0^2(D,X)},
\]
so Assumption~\ref{ass:plm_standing} implies
\[
\bE_0[c^2] \le \underline\sigma^{-2}\,\bE_0[r^2] < \infty.
\]
Hence \(c \in L_2(P_{0, DX})\), and therefore \(w \in L_2(P_0)\) as well. Moreover,
\[
\bE_0[\alpha] = \bE_0[r] = 0, \qquad \bE_0[\gamma] = \bE_0[\bE_0[\gamma \mid X]] = 0, \qquad \bE_0[w] = 0,
\]
and 
\begin{equation}
\label{eq:plm_orthcomp_identities}
\bE_0[\gamma \mid X] = 0, \qquad \bE_0[w \mid D,X] = 0, \qquad \bE_0[Uw \mid D,X] = 0.
\end{equation}

\smallskip
\noindent\textbf{Characterization of \(\Lambda^\perp\).}

To start, assume \(r \in \Lambda^\perp\). We first show that \(\alpha = 0\) almost surely. Let
\[
s_\alpha(Z) := \alpha(X).
\]
Since \(\bE_0[s_\alpha] = \bE_0[\alpha] = 0\) and
\[
\bE_0[U s_\alpha \mid D,X] = \alpha(X)\,\bE_0[U \mid D,X] = 0,
\]
Corollary~\ref{cor:plm_nuisance_tangent} gives \(s_\alpha \in \Lambda\). Because \(r \in \Lambda^\perp\),
\[
0 = \bE_0[r\,s_\alpha] = \bE_0[r\,\alpha].
\]
Using the decomposition \(r = \alpha + \gamma + Uc + w\), we obtain
\[
\bE_0[r\,\alpha] = \bE_0\bigl[\alpha^{2}\bigr] + \bE_0[\gamma\,\alpha] + \bE_0[Uc\,\alpha] + \bE_0[w\,\alpha].
\]
The three cross terms vanish as follows,
\begin{align*}
\bE_0[\gamma\,\alpha] &= \bE_0[\alpha(X)\,\bE_0[\gamma \mid X]] = 0, \\
\bE_0[Uc\,\alpha] &= \bE_0[\alpha(X)\,c(D,X)\,\bE_0[U \mid D,X]] = 0, \\
\bE_0[w\,\alpha] &= \bE_0[\alpha(X)\,\bE_0[w \mid D,X]] = 0.
\end{align*}
Hence \(\bE_0[r\,\alpha] = \bE_0\bigl[\alpha^{2}\bigr]\), so \(\bE_0\bigl[\alpha^{2}\bigr] = 0\) and \(\alpha = 0\) \(P_0\text{-a.s.}\) Next, we show that \(\gamma = 0\) almost surely. Let
\[
s_\gamma(Z) := \gamma(D,X).
\]
Since \(\bE_0[s_\gamma] = \bE_0[\gamma] = 0\) and
\[
\bE_0[U s_\gamma \mid D,X] = \gamma(D,X)\,\bE_0[U \mid D,X] = 0,
\]
Corollary~\ref{cor:plm_nuisance_tangent} again gives \(s_\gamma \in \Lambda\). Therefore,
\[
0 = \bE_0[r\,s_\gamma] = \bE_0[r\,\gamma].
\]
Since \(\alpha = 0\), we have
\[
\bE_0[r\,\gamma] = \bE_0\bigl[\gamma^{2}\bigr] + \bE_0[Uc\,\gamma] + \bE_0[w\,\gamma].
\]
The last two terms vanish similarly,
\begin{align*}
\bE_0[Uc\,\gamma] &= \bE_0[\gamma(D,X)\,c(D,X)\,\bE_0[U \mid D,X]] = 0, \\
\bE_0[w\,\gamma] &= \bE_0[\gamma(D,X)\,\bE_0[w \mid D,X]] = 0.
\end{align*}
Hence \(\bE_0[r\,\gamma] = \bE_0\bigl[\gamma^{2}\bigr]\), so \(\bE_0\bigl[\gamma^{2}\bigr] = 0\) and 
\(\gamma = 0\) \(P_0\text{-a.s.}\) We now show that \(w = 0\) almost surely. Let
\[
s_w(Z) := w(Y,D,X).
\]
Since \(\bE_0[s_w] = \bE_0[w] = 0\) and \(\bE_0[U s_w \mid D,X] = \bE_0[Uw \mid D,X] = 0\) by~\eqref{eq:plm_orthcomp_identities}, Corollary~\ref{cor:plm_nuisance_tangent} gives \(s_w \in \Lambda\). Thus,
\[
0 = \bE_0[r\,s_w] = \bE_0[r\,w].
\]
Since \(\alpha = \gamma = 0\), we obtain
\[
\bE_0[r\,w] = \bE_0[Uc\,w] + \bE_0\bigl[w^{2}\bigr].
\]
The first term vanishes as
\[
\bE_0[Uc\,w] = \bE_0[c(D,X)\,\bE_0[Uw \mid D,X]] = 0.
\]
Therefore \(\bE_0[r\,w] = \bE_0\bigl[w^{2}\bigr]\), so \(\bE_0\bigl[w^{2}\bigr] = 0\) and \(w = 0\) \(P_0\text{-a.s.}\) Up to this point, 
\[
r = Uc(D,X).
\]
It remains to identify the constraint on \(c\). Let any \(h \in L_2(P_{0,X})\) be given, and define
\[
s_h(Z) := \frac{U\,h(X)}{\sigma_0^2(D,X)}.
\]
Since
\[
\bE_0\bigl[s_h^{2}\bigr] = \bE_0\!\left[\frac{U^2\,h(X)^2}{\sigma_0^4(D,X)}\right] = \bE_0\!\left[\frac{h(X)^2}{\sigma_0^2(D,X)}\right] \le \underline\sigma^{-2}\,\bE_0\bigl[h(X)^2\bigr] < \infty,
\]
we have \(s_h \in L_2(P_0)\). Also,
\[
\bE_0[s_h] = \bE_0\!\left[\frac{h(X)\,\bE_0[U \mid D,X]}{\sigma_0^2(D,X)}\right] = 0,
\]
and
\[
\bE_0[U s_h \mid D,X] = \frac{h(X)\,\bE_0\bigl[U^{2} \mid D,X\bigr]}{\sigma_0^2(D,X)} = h(X).
\]
Hence Corollary~\ref{cor:plm_nuisance_tangent} implies \(s_h \in \Lambda\). Since \(r \in \Lambda^\perp\),
\[
0 = \bE_0[r\,s_h] = \bE_0\!\left[Uc(D,X) \cdot \frac{U\,h(X)}{\sigma_0^2(D,X)}\right] = \bE_0[c(D,X)\,h(X)] = \bE_0[\bE_0[c(D,X) \mid X]\,h(X)]
\]
for every \(h \in L_2(P_{0,X})\). Choosing
\[
h(X) := \bE_0[c(D,X) \mid X] \in L_2(P_{0,X}),
\]
we arrive at
\[
0 = \bE_0\left[\bE_0[c(D,X) \mid X]^2\right].
\]
Therefore
\[
\bE_0[c(D,X) \mid X] = 0 \qquad P_0\text{-a.s.}
\]
This proves
\[
\Lambda^\perp \subseteq \{Ua(D,X) : a \in L_2(P_{0, DX}),\; \bE_0[a(D,X) \mid X] = 0\}.
\]
To see the reverse inclusion, suppose
\[
r = Ua(D,X) \qquad \text{with} \qquad a \in L_2(P_{0, DX}), \quad \bE_0[a(D,X) \mid X] = 0.
\]
Since
\[
\bE_0\bigl[r^{2}\bigr] = \bE_0\bigl[U^2 a(D,X)^2\bigr] = \bE_0\bigl[\sigma_0^2(D,X)\,a(D,X)^2\bigr] \le \overline\sigma^2\,\bE_0\bigl[a(D,X)^2\bigr] < \infty,
\]
we have \(r \in L_2(P_0)\), and also
\[
\bE_0[r] = \bE_0[a(D,X)\,\bE_0[U \mid D,X]] = 0,
\]
so \(r \in L_2^0(P_0)\). Next, take any \(s \in \Lambda\). By Corollary~\ref{cor:plm_nuisance_tangent}, there exists \(h \in L_2(P_{0,X})\) such that
\[
\bE_0[Us \mid D,X] = h(X).
\]
Then, by iterated expectations,
\begin{align*}
\bE_0[r\,s] &= \bE_0[Ua(D,X)\,s] \\
&= \bE_0[a(D,X)\,\bE_0[Us \mid D,X]] \\
&= \bE_0[a(D,X)\,h(X)] \\
&= \bE_0[h(X)\,\bE_0[a(D,X) \mid X]] \\
&= 0.
\end{align*}
Since this holds for every \(s \in \Lambda\), we conclude that \(r \in \Lambda^\perp\). Therefore
\[
\{Ua(D,X) : a \in L_2(P_{0, DX}),\; \bE_0[a(D,X) \mid X] = 0\} \subseteq \Lambda^\perp.
\]
Combining the two inclusions proves~\eqref{eq:plm_Lambda_perp}.

\smallskip
\noindent\textbf{Characterization of \(\cT^\perp\).}

To start, assume \(r \in \cT^\perp\). Since \(\Lambda \subseteq \cT\), we also have \(r \in \Lambda^\perp\). From the above, we know there exists \(c(D,X) \in L_2(P_{0, DX})\) such that
\[
r = Uc(D,X), \qquad \bE_0[c(D,X) \mid X] = 0.
\]
To obtain the additional restriction, define
\[
s_\theta(Z) := \frac{UV}{\sigma_0^2(D,X)}.
\]
Since
\[
\bE_0\bigl[s_\theta^2\bigr] = \bE_0\!\left[\frac{U^2 V^2}{\sigma_0^4(D,X)}\right] = \bE_0\!\left[\frac{V^2}{\sigma_0^2(D,X)}\right] \le \underline\sigma^{-2}\,\bE_0\bigl[V^2\bigr] < \infty,
\]
we have \(s_\theta \in L_2(P_0)\). Also,
\[
\bE_0[s_\theta] = \bE_0\!\left[\frac{V\,\bE_0[U \mid D,X]}{\sigma_0^2(D,X)}\right] = 0,
\]
and
\[
\bE_0[U s_\theta \mid D,X] = \frac{V\,\bE_0\bigl[U^2 \mid D,X\bigr]}{\sigma_0^2(D,X)} = V = 0 + 1 \cdot V.
\]
Hence Proposition~\ref{prop:plm_tangent_space} implies \(s_\theta \in \cT\). Since \(r \in \cT^\perp\),
\[
0 = \bE_0[r\,s_\theta] = \bE_0\!\left[Uc(D,X) \cdot \frac{UV}{\sigma_0^2(D,X)}\right] = \bE_0[c(D,X)\,V].
\]
Therefore
\[
r = Uc(D,X) \qquad \text{with} \qquad \bE_0[c(D,X) \mid X] = 0, \quad \bE_0[c(D,X)\,V] = 0,
\]
which proves
\[
\cT^\perp \subseteq \{Ua(D,X) : a \in L_2(P_{0, DX}),\; \bE_0[a(D,X) \mid X] = 0,\; \bE_0[a(D,X)\,V] = 0\}.
\]
To see the reverse inclusion, suppose
\[
r = Ua(D,X) \qquad \text{with} \qquad a \in L_2(P_{0, DX}), \quad \bE_0[a(D,X) \mid X] = 0, \quad \bE_0[a(D,X)\,V] = 0.
\]
Take any \(s \in \cT\). By Proposition~\ref{prop:plm_tangent_space}, there exist \(h \in L_2(P_{0,X})\) and \(b \in \RR\) such that
\[
\bE_0[Us \mid D,X] = h(X) + bV.
\]
Then
\begin{align*}
\bE_0[r\,s] &= \bE_0[Ua(D,X)\,s] \\
&= \bE_0[a(D,X)\,\bE_0[Us \mid D,X]] \\
&= \bE_0[a(D,X)\,h(X)] + b\,\bE_0[a(D,X)\,V] \\
&= \bE_0[h(X)\,\bE_0[a(D,X) \mid X]] + b\,\bE_0[a(D,X)\,V] \\
&= 0.
\end{align*}
Thus \(r\) is orthogonal to every element of \(\cT\), so \(r \in \cT^\perp\). Therefore
\[
\{Ua(D,X) : a \in L_2(P_{0, DX}),\; \bE_0[a(D,X) \mid X] = 0,\; \bE_0[a(D,X)\,V] = 0\} \subseteq \cT^\perp.
\]
Combining the two inclusions proves~\eqref{eq:plm_T_perp}.
\end{proof}

\begin{proposition}[All influence functions]
\label{prop:plm_all_IFs}
A function \(\varphi \in L_2^0(P_0)\) is an influence function for \(\theta\) if and only if it can be written as
\begin{equation}
\label{eq:plm_all_IFs}
\varphi(Z) = Ua(D, X) \qquad \text{with} \qquad \bE_0[a(D, X) \mid X] = 0, \qquad \bE_0[a(D, X)V] = 1.
\end{equation}
\end{proposition}

\begin{proof}
If \(\varphi\) is an influence function, then it must be the case that \(\varphi \in \Lambda^\perp\). Proposition~\ref{prop:plm_orthogonal_complements} therefore implies that \(\varphi = Ua(D, X)\) with \(\bE_0[a \mid X] = 0\). To identify the normalization, we test \(\varphi\) against the score \(s_\theta := UV/\sigma_0^2(D, X) \in \cT\) from the proof of Proposition~\ref{prop:plm_orthogonal_complements}. Since the pathwise derivative of \(\theta\) along \(s_\theta\) is \(1\), we have
\[
1 = \bE_0[\varphi\,s_\theta] = \bE_0\!\left[Ua(D, X) \cdot \frac{UV}{\sigma_0^2(D, X)}\right] = \bE_0[a(D, X)V].
\]
Conversely, suppose~\eqref{eq:plm_all_IFs} holds. Let \(s \in \cT\). By Proposition~\ref{prop:plm_tangent_space},
\[
\bE_0[Us \mid D, X] = h(X) + bV
\]
for some \(h \in L_2(P_{0,X})\) and \(b \in \RR\), and the derivative of \(\theta\) along \(s\) is exactly \(b\). Then
\begin{align*}
\bE_0[\varphi\,s] &= \bE_0[a(D, X)\,\bE_0[Us \mid D, X]] = \bE_0[a(D, X)h(X)] + b\,\bE_0[a(D, X)V] = 0 + b = b.
\end{align*}
Thus \(\varphi\) yields the pathwise derivative along every tangent direction and is an influence function.
\end{proof}

\begin{proposition}[Efficient influence function]
\label{prop:plm_EIF}
Define
\[
r_0^\star(X) := \frac{\bE_0\bigl[D / \sigma_0^2(D, X) \mid X\bigr]}{\bE_0\bigl[1 / \sigma_0^2(D, X) \mid X\bigr]},
\]
and
\[
K_0 := \bE_0\!\left[\frac{(D - r_0^\star(X))V}{\sigma_0^2(D, X)}\right] = \bE_0\!\left[\frac{(D - r_0^\star(X))^2}{\sigma_0^2(D, X)}\right].
\]
Then the efficient influence function for \(\theta\) is given by
\begin{equation}
\label{eq:plm_EIF}
\varphi_{\mathrm{eff}}(Z) = K_0^{-1}\,\frac{D - r_0^\star(X)}{\sigma_0^2(D, X)}\,U.
\end{equation}
It is the unique influence function lying in \(\cT\).
\end{proposition}

\begin{proof}
Set \(a_{\mathrm{eff}}(D, X) := K_0^{-1}(D - r_0^\star(X))/\sigma_0^2(D, X)\). Since \(r_0^\star\) is the weighted conditional mean of \(D\), we have
\[
\bE_0\!\left[\frac{D - r_0^\star(X)}{\sigma_0^2(D, X)} \;\middle|\; X\right] = 0,
\]
so \(\bE_0[a_{\mathrm{eff}} \mid X] = 0\). Also, by definition of \(K_0\),
\[
\bE_0[a_{\mathrm{eff}} V] = 1.
\]
Thus, \(\varphi_{\mathrm{eff}} := Ua_{\mathrm{eff}}\) is an influence function. Next, let \(\varphi = Ua\) be any other influence function and write \(\delta := a - a_{\mathrm{eff}}\). Then
\[
\bE_0[\delta \mid X] = 0, \qquad \bE_0[\delta V] = 0.
\]
Expanding the variance,
\begin{align*}
\bE_0\bigl[\sigma_0^2 a^2\bigr] &= \bE_0\bigl[\sigma_0^2 a_{\mathrm{eff}}^2\bigr] + \bE_0\bigl[\sigma_0^2 \delta^2\bigr] + 2\bE_0\bigl[\sigma_0^2 a_{\mathrm{eff}} \delta\bigr] \\
&= \bE_0\bigl[\sigma_0^2 a_{\mathrm{eff}}^2\bigr] + \bE_0\bigl[\sigma_0^2 \delta^2\bigr] + 2K_0^{-1}\bE_0[(D - r_0^\star(X))\delta(D, X)].
\end{align*}
The cross-term vanishes as
\begin{align*}
\bE_0[(D - r_0^\star(X))\delta(D, X)] &= \bE_0[V\delta(D, X)] + \bE_0[(m_0(X) - r_0^\star(X))\delta(D, X)] = 0 + 0,
\end{align*}
where the first term is zero by \(\bE_0[\delta V] = 0\) and the second because \(m_0 - r_0^\star\) is a function of \(X\) and \(\bE_0[\delta \mid X] = 0\). It follows that
\[
\bE_0\bigl[\varphi^2\bigr] = \bE_0\bigl[\varphi_{\mathrm{eff}}^2\bigr] + \bE_0\bigl[\sigma_0^2 \delta^2\bigr] \ge \bE_0\bigl[\varphi_{\mathrm{eff}}^2\bigr],
\]
with equality if and only if \(\delta = 0\) almost surely. Therefore, \eqref{eq:plm_EIF} is the minimum-variance influence function. We also have
\[
\bE_0[U\varphi_{\mathrm{eff}} \mid D, X] = K_0^{-1}(D - r_0^\star(X)) = K_0^{-1}(m_0(X) - r_0^\star(X)) + K_0^{-1}V,
\]
which belongs to \(\cA_{\mathrm{plm}}\). By Proposition~\ref{prop:plm_tangent_space}, we know that \(\varphi_{\mathrm{eff}} \in \cT\). Since every influence function lies in \(\Lambda^\perp\) and the efficient influence function is the unique element of \(\Lambda^\perp \cap \cT\), this completes the proof.
\end{proof}

\subsubsection{The Residual-on-residual Moment}
\label{subsubsec:plm_ror}

The residual-on-residual moment is given by \citep{chernozhukov_doubledebiased_2018}
\begin{equation}
\label{eq:plm_ror}
\psi_{\mathrm{ror}}(Z;\, \theta, \ell, m) := (D - m(X))\left[Y - \ell(X) - \theta(D - m(X))\right].
\end{equation}
We now show that it is orthogonal with respect to the nuisance parameterization \(\eta = (\ell, m)\), identify the influence function it induces, and compare it to the efficient influence function.

\begin{proposition}[Neyman orthogonality and mean-zero]
\label{prop:plm_ror_orthogonality}
The residual-on-residual moment is correctly specified and Neyman orthogonal with respect to the nuisance parameterization \(\eta = (\ell, m)\).
\end{proposition}

\begin{proof}
Let \(P\) be any distribution in the partially linear model, and write
\[
U_P := Y - \ell_P(X) - \theta(P)(D - m_P(X)), \qquad V_P := D - m_P(X).
\]
Then \(\bE_P[U_P \mid D, X] = 0\) and \(\bE_P[V_P \mid X] = 0\). Therefore,
\[
\bE_P[\psi_{\mathrm{ror}}(Z;\, \theta(P), \ell_P, m_P)] = \bE_P[V_P U_P] = \bE_P[V_P\,\bE_P[U_P \mid D, X]] = 0.
\]
Next, we fix \(P_0\) and perturb only the nuisance parameter around \(\eta_0 = (\ell_0, m_0)\):
\[
\ell_t = \ell_0 + th_\ell, \qquad m_t = m_0 + th_m.
\]
At \(\theta = \theta_0\),
\begin{align*}
Y - \ell_t(X) - \theta_0(D - m_t(X)) &= U - th_\ell(X) + t\theta_0 h_m(X),\\
D - m_t(X) &= V - th_m(X).
\end{align*}
Hence, it follows that
\begin{align*}
&\frac{d}{dt}\bE_0[\psi_{\mathrm{ror}}(Z;\, \theta_0, \ell_t, m_t)]\bigg|_{t=0} \\
&= \bE_0[-h_m(X)U + V(-h_\ell(X) + \theta_0 h_m(X))] \\
&= -\bE_0[h_m(X)\,\bE_0[U \mid X]] - \bE_0[h_\ell(X)\,\bE_0[V \mid X]] + \theta_0\,\bE_0[h_m(X)\,\bE_0[V \mid X]] \\
&= 0,
\end{align*}
and the residual-on-residual moment is Neyman orthogonal for every nuisance direction \((h_\ell, h_m)\).
\end{proof}

\begin{proposition}[The induced influence function]
\label{prop:plm_ror_IF}
The Jacobian of the residual-on-residual moment with respect to \(\theta\) is
\[
G_{\mathrm{ror}} = -\bE_0\bigl[V^2\bigr] = -J_0,
\]
and the induced influence function is
\begin{equation}
\label{eq:plm_ror_IF}
\varphi_{\mathrm{ror}}(Z) = \frac{VU}{J_0}.
\end{equation}
This is a valid influence function for \(\theta\).
\end{proposition}

\begin{proof}
Differentiating~\eqref{eq:plm_ror} with respect to \(\theta\) at the truth gives
\[
\partial_\theta \psi_{\mathrm{ror}}(Z;\, \theta_0, \ell_0, m_0) = -V^2, \qquad G_{\mathrm{ror}} = -\bE_0\bigl[V^2\bigr] = -J_0.
\]
Hence, the induced influence function is
\[
-G_{\mathrm{ror}}^{-1}\,\psi_{\mathrm{ror}}(Z;\, \theta_0, \ell_0, m_0) = \frac{VU}{J_0}.
\]
To see that it is a valid influence function, we apply Proposition~\ref{prop:plm_all_IFs} with \(a(D, X) = V/J_0\). Then \(\bE_0[a \mid X] = J_0^{-1}\bE_0[V \mid X] = 0\) and \(\bE_0[aV] = J_0^{-1}\bE_0\bigl[V^2\bigr] = 1\), so~\eqref{eq:plm_ror_IF} is an influence function.
\end{proof}

\begin{proposition}[Decomposition of \(\varphi_{\mathrm{ror}}\) into  \(\varphi_\mathrm{eff}\) plus a \(\cT^\perp\) remainder]
\label{prop:plm_ror_decomposition}
Let
\begin{equation}
\label{eq:plm_Delta_def}
\Delta(Z) := \varphi_{\mathrm{ror}}(Z) - \varphi_{\mathrm{eff}}(Z) = U\Delta_a(D, X), \qquad \Delta_a(D, X) := \frac{V}{J_0} - K_0^{-1}\,\frac{D - r_0^\star(X)}{\sigma_0^2(D, X)}.
\end{equation}
Then \(\Delta \in \cT^\perp\). Consequently,
\begin{equation}
\label{eq:plm_IF_decomposition}
\varphi_{\mathrm{ror}} = \varphi_{\mathrm{eff}} + \Delta, \qquad \varphi_{\mathrm{eff}} \in \cT, \qquad \Delta \in \cT^\perp.
\end{equation}
In particular,
\begin{equation}
\label{eq:plm_Pythagorean}
\bE_0\bigl[\varphi_{\mathrm{ror}}^2\bigr] = \bE_0\bigl[\varphi_{\mathrm{eff}}^2\bigr] + \bE_0\bigl[\Delta^2\bigr] = \bE_0\bigl[\varphi_{\mathrm{eff}}^2\bigr] + \bE_0\bigl[\sigma_0^2(D, X)\,\Delta_a(D, X)^2\bigr].
\end{equation}
Thus, the residual-on-residual moment is generally inefficient, and the efficiency gap is the \(\cT^\perp\) component.
\end{proposition}

\begin{proof}
The decomposition~\eqref{eq:plm_IF_decomposition} is immediate from the definition of \(\Delta_a\). To show \(\Delta \in \cT^\perp\), Proposition~\ref{prop:plm_orthogonal_complements} says it suffices to verify
\[
\bE_0[\Delta_a(D, X) \mid X] = 0, \qquad \bE_0[\Delta_a(D, X)V] = 0.
\]
The first identity holds because
\[
\bE_0\!\left[\frac{V}{J_0} \;\middle|\; X\right] = 0, \qquad \bE_0\!\left[\frac{D - r_0^\star(X)}{\sigma_0^2(D, X)} \;\middle|\; X\right] = 0
\]
by the definition of \(r_0^\star\). The second holds because
\[
\bE_0[\Delta_a V] = \frac{\bE_0\bigl[V^2\bigr]}{J_0} - K_0^{-1}\bE_0\!\left[\frac{(D - r_0^\star(X))V}{\sigma_0^2(D, X)}\right] = 1 - 1 = 0.
\]
Therefore \(\Delta \in \cT^\perp\).

Since \(\varphi_{\mathrm{eff}} \in \cT\) by Proposition~\ref{prop:plm_EIF}, orthogonality gives \(\bE_0[\varphi_{\mathrm{eff}}\,\Delta] = 0\). Expanding \(\bE_0\bigl[(\varphi_{\mathrm{eff}} + \Delta)^2\bigr]\) yields~\eqref{eq:plm_Pythagorean}. The final expression follows from \(\Delta^2 = U^2\Delta_a^2\) and \(\bE_0\bigl[U^2\Delta_a^2\bigr] = \bE_0\bigl[\sigma_0^2(D, X)\,\Delta_a(D, X)^2\bigr]\). If the gap term is nonzero, then the residual-on-residual influence function has strictly larger variance than the efficient one.
\end{proof}

\begin{corollary}[Homoskedastic special case]
\label{cor:plm_homoskedastic}
If \(\sigma_0^2(D, X) \equiv \sigma^2\) is constant, then
\[
r_0^\star(X) = \bE_0[D \mid X] = m_0(X), \qquad K_0 = \frac{J_0}{\sigma^2},
\]
and therefore
\[
\varphi_{\mathrm{eff}}(Z) = \frac{VU}{J_0} = \varphi_{\mathrm{ror}}(Z).
\]
So the residual-on-residual moment is efficient in the homoskedastic partially linear model.
\end{corollary}

\begin{proof}
If \(\sigma_0^2(D, X) \equiv \sigma^2\), then
\[
r_0^\star(X) = \frac{\bE_0\bigl[D/\sigma^2 \mid X\bigr]}{\bE_0\bigl[1/\sigma^2 \mid X\bigr]} = \bE_0[D \mid X] = m_0(X).
\]
Hence
\[
K_0 = \bE_0\!\left[\frac{(D - m_0(X))V}{\sigma^2}\right] = \frac{\bE_0\bigl[V^2\bigr]}{\sigma^2} = \frac{J_0}{\sigma^2}.
\]
Substituting into~\eqref{eq:plm_EIF} gives
\[
\varphi_{\mathrm{eff}} = \frac{VU}{J_0} = \varphi_{\mathrm{ror}}.
\]
\end{proof}

\begin{remark}[A broader class of orthogonal moments]
\label{rem:plm_broader}
The residual-on-residual moment is only one member of a larger class. Let \(a_0(D, X)\) be any square-integrable weight satisfying
\[
\bE_0[a_0(D, X) \mid X] = 0, \qquad \bE_0[a_0(D, X)V] \neq 0.
\]
Then it follows that
\[
\psi_{a_0}(Z;\, \theta, \ell, m) := a_0(D, X)(Y - \ell(X) - \theta(D - m(X)))
\]
is correctly specified, Neyman orthogonal for \(\eta = (\ell, m)\), and induces the influence function
\[
\varphi_{a_0}(Z) = \frac{Ua_0(D, X)}{\bE_0[a_0(D, X)V]}.
\]
The residual-on-residual moment corresponds to the special case of \(a_0 = V\).
\end{remark}

\subsubsection{Forward Direction}
\label{subsubsec:plm_forward}

We now verify Assumptions~\ref{assumption:correct_spec}--\ref{assumption:hellinger} of Theorem~\ref{thm:forward} for the moment \(\psi_{\mathrm{ror}}\). To start, define the dense score class
\(
\cS := \cT_b,
\)
where \(\cT_b\) is the bounded score class from Lemma~\ref{lem:plm_bounded_density}. By the same lemma and Proposition~\ref{prop:plm_tangent_space}, the \(L_2(P_0)\)-closure of \(\cS\) is \(\cT\).

\textbf{Assumption~\ref{assumption:correct_spec}.} This is Proposition~\ref{prop:plm_ror_orthogonality}.

\textbf{Assumption~\ref{assumption:coord_smooth_all}.} Take any \(s \in \cS\). By definition, there exist bounded objects \(s_X, s_D, a, w, b\) satisfying the conditions of Proposition~\ref{prop:plm_bounded_submodels} such that
\[
s = s_X + s_D + \frac{U(a + bV)}{\sigma_0^2} + w.
\]
The proposition furnishes an explicit regular submodel \(t \mapsto P_{t,s}\) with score \(s\) along which
\[
\theta_{t,s} = \theta_0 + tb, \qquad m_{t,s} = m_0 + tk_s, \qquad \ell_{t,s} = \ell_0 + t(a + \theta_0 k_s),
\]
where
\[
k_s(X) := \bE_0[Vs_D(D, X) \mid X].
\]
Hence the induced coordinate path is differentiable at \(0\) with
\[\dot\theta_{0,s} = b, \qquad \dot m_{0,s} = k_s, \qquad \dot\ell_{0,s} = a + \theta_0 k_s,\]
which is exactly the smoothness required in Assumption~\ref{assumption:coord_smooth_all}.

\textbf{Assumption~\ref{assumption:frechet_all}.} Equip \(\RR \times \cH\) with the norm
\[
\|(\delta\theta, \delta\ell, \delta m)\| := |\delta\theta| + \|\delta\ell\|_\infty + \|\delta m\|_\infty.
\]
At \((\theta_0, \ell_0, m_0)\), a direct expansion gives
\begin{align*}
&\psi_{\mathrm{ror}}(Z;\, \theta_0 + \delta\theta, \ell_0 + \delta\ell, m_0 + \delta m) - \psi_{\mathrm{ror}}(Z;\, \theta_0, \ell_0, m_0) \\
&\qquad = -V^2\delta\theta - V\delta\ell(X) + (\theta_0 V - U)\delta m(X) + R_{\mathrm{ror}}(Z;\, \delta\theta, \delta\ell, \delta m),
\end{align*}
where the remainder is
\[
R_{\mathrm{ror}} = \delta m\,\delta\ell + 2\delta\theta\,V\,\delta m - \theta_0(\delta m)^2 - \delta\theta(\delta m)^2.
\]
Therefore the Fr\'echet derivative is the bounded linear map
\begin{align*}
D_\theta\,\psi_{\mathrm{ror},0}(\delta\theta) &= -V^2\,\delta\theta, \\
D_\eta\,\psi_{\mathrm{ror},0}[\delta\ell, \delta m] &= -V\,\delta\ell(X) + (\theta_0 V - U)\,\delta m(X).
\end{align*}
Since \(U\) and \(V\) are bounded,
\[
\|R_{\mathrm{ror}}(\delta\theta, \delta\ell, \delta m)\|_{L_2(P_0)} \le C\|(\delta\theta, \delta\ell, \delta m)\|^2
\]
for some finite constant \(C\). Thus the map \((\theta, \ell, m) \mapsto \psi_{\mathrm{ror}}(\cdot;\, \theta, \ell, m) \in L_2(P_0)\) is Fr\'echet differentiable.

\textbf{Assumption~\ref{assumption:regularity_submodels}.} Fix \(s \in \cS\) and its associated submodel \(t \mapsto P_{t,s}\). Let
\[
f_{t,s}(Z) := \psi_{\mathrm{ror}}(Z;\, \theta_{t,s}, \ell_{t,s}, m_{t,s}).
\]
Since the coordinate paths are affine in \(t\), the same expansion as above gives
\[
f_{t,s}(Z) = VU + t\,\dot f_{0,s}(Z) + t^2 R_{t,s}(Z),
\]
with
\begin{equation}
\label{eq:plm_fdot_ror}
\dot f_{0,s}(Z) = -V^2 b - V(a + \theta_0 k_s) + (\theta_0 V - U)k_s = -V^2 b - Va - Uk_s,
\end{equation}
and with \(R_{t,s}\) uniformly bounded since all elements are bounded. Hence
\[
\left\|\frac{f_{t,s} - f_{0,s}}{t} - \dot f_{0,s}\right\|_{L_2(P_0)} \to 0,
\]
and \(f_{0,s} = VU\) is bounded, and the conditions of Lemma~\ref{lem:differentiation_varying_f} hold for each \(s \in \cS\) and its corresponding submodel from Assumption~\ref{assumption:coord_smooth_all}.

\textbf{Assumption~\ref{assumption:jacobian}.} This is  Proposition~\ref{prop:plm_ror_IF}. 

\textbf{Assumption~\ref{assumption:neyman}.} This is Proposition~\ref{prop:plm_ror_orthogonality}.

\textbf{Assumption~\ref{assumption:hellinger}.} This is Proposition~\ref{prop:plm_Hellinger_Lipschitz}.

Since all the assumptions of Theorem~\ref{thm:forward} are verified, we conclude that \(\theta\) is pathwise differentiable at \(P_0\) with influence function
\[
-G^{-1}\psi_{\mathrm{ror}}(Z;\, \theta_0, \ell_0, m_0) = \frac{VU}{J_0} = \varphi_{\mathrm{ror}}(Z),
\]
which recovers Proposition~\ref{prop:plm_ror_IF}.

\subsubsection{Reverse Direction}
\label{subsubsec:plm_reverse}

To verify the reverse implication, it is most convenient to freeze the efficient weight at \(P_0\) and work with the moment
\begin{equation}
\label{eq:plm_frozen_eff}
\psi_{\mathrm{eff},0}(Z;\, \theta, \ell, m) := a_{\mathrm{eff},0}(D, X)(Y - \ell(X) - \theta(D - m(X))),
\end{equation}
where
\[
a_{\mathrm{eff},0}(D, X) := K_0^{-1}\,\frac{D - r_0^\star(X)}{\sigma_0^2(D, X)}.
\]
At the truth,
\[
\psi_{\mathrm{eff},0}(Z;\, \theta_0, \ell_0, m_0) = \varphi_{\mathrm{eff}}(Z).
\]

\textbf{Assumption~\ref{assumption:correct_spec}.} For any distribution \(P\) in the partially linear model,
\[
\bE_P[\psi_{\mathrm{eff},0}(Z;\, \theta(P), \ell_P, m_P)] = \bE_P[a_{\mathrm{eff},0}(D, X)\,U_P] = \bE_P[a_{\mathrm{eff},0}(D, X)\,\bE_P[U_P \mid D, X]] = 0.
\]
So the moment is correctly specified.

\textbf{Assumption~\ref{assumption:frechet_all}.} At \((\theta_0, \ell_0, m_0)\),
\begin{align*}
&\psi_{\mathrm{eff},0}(Z;\, \theta_0 + \delta\theta, \ell_0 + \delta\ell, m_0 + \delta m) - \psi_{\mathrm{eff},0}(Z;\, \theta_0, \ell_0, m_0) \\
&\qquad = a_{\mathrm{eff},0}(D, X)(-V\delta\theta - \delta\ell(X) + \theta_0\,\delta m(X)) + a_{\mathrm{eff},0}(D, X)\,\delta\theta\,\delta m(X).
\end{align*}
Thus the Fr\'echet derivative is
\begin{align*}
D_\theta\,\psi_{\mathrm{eff},0}(\delta\theta) &= -a_{\mathrm{eff},0}(D,X)\,V\,\delta\theta, \\
D_\eta\,\psi_{\mathrm{eff},0}[\delta\ell,\delta m] &= a_{\mathrm{eff},0}(D,X)(-\delta\ell(X) + \theta_0\,\delta m(X)).
\end{align*}
where the remainder satisfies
\[
\|a_{\mathrm{eff},0}\,\delta\theta\,\delta m\|_{L_2(P_0)} \le \|a_{\mathrm{eff},0}\|_\infty\,|\delta\theta|\,\|\delta m\|_\infty \le C\|(\delta\theta, \delta\ell, \delta m)\|^2,
\]
and Assumption~\ref{assumption:frechet_all} holds.

\textbf{Assumption~\ref{assumption:pathwise}.} By Proposition~\ref{prop:plm_EIF}, \(\theta\) is pathwise differentiable at \(P_0\) with influence function
\[
\varphi_{\mathrm{eff}}(Z) = \psi_{\mathrm{eff},0}(Z;\, \theta_0, \ell_0, m_0).
\]

\textbf{Assumption~\ref{assumption:product_reverse}. }Finally, we construct the coordinate submodels explicitly.

\emph{\(\theta\)-coordinate submodel.} Set
\[
\theta_t := \theta_0 + t, \qquad \ell_t := \ell_0, \qquad m_t := m_0.
\]
Equivalently,
\[
\mu_t(D, X) := \ell_0(X) + \theta_t(D - m_0(X)) = \mu_0(D, X) + tV.
\]
We apply Proposition~\ref{prop:plm_Y_submodel} with \(w \equiv 0\) and \(r(D, X) = V\). This furnishes a regular QMD submodel through \(P_0\) with score
\begin{equation}
\label{eq:plm_s_beta}
s_\beta(Z) = \frac{UV}{\sigma_0^2(D, X)}.
\end{equation}
By construction,
\[
\dot\theta_{0, s_\beta} = 1, \qquad \dot\ell_{0, s_\beta} = 0, \qquad \dot m_{0, s_\beta} = 0,
\]
which is the required \(\theta\)-coordinate submodel.

\emph{\(\eta\)-coordinate submodel.} We fix a nuisance direction \(h = (h_\ell, h_m) \in \dot \cH\) and set
\[
\theta_t := \theta_0, \qquad m_t := m_0 + th_m, \qquad \ell_t := \ell_0 + th_\ell.
\]
It follows that
\[
\mu_t(D, X) = \ell_t(X) + \theta_0(D - m_t(X)) = \mu_0(D, X) + t(h_\ell(X) - \theta_0 h_m(X)).
\]
Now perturb \(D \mid X\) using Proposition~\ref{prop:plm_D_submodel} with score
\[
s_D(D, X) = \frac{h_m(X)V}{\tau_0^2(X)},
\]
so that \(m_t = m_0 + th_m\). We also perturb \(Y \mid D, X\) using Proposition~\ref{prop:plm_Y_submodel} with \(w \equiv 0\) and derivative
\[
r(D, X) = h_\ell(X) - \theta_0 h_m(X).
\]
Equivalently, this is the special case of Proposition~\ref{prop:plm_bounded_submodels} with \(s_X = 0\), \(a = h_\ell - \theta_0 h_m\), \(b = 0\), and \(w = 0\). The resulting score is
\begin{equation}
\label{eq:plm_s_h}
s_h(Z) = \frac{h_m(X)V}{\tau_0^2(X)} + \frac{U(h_\ell(X) - \theta_0 h_m(X))}{\sigma_0^2(D, X)}.
\end{equation}
By construction,
\[
\dot\theta_{0, s_h} = 0, \qquad \dot\ell_{0, s_h} = h_\ell, \qquad \dot m_{0, s_h} = h_m.
\]
Therefore, the local product structure required by Assumption~\ref{assumption:product_reverse} holds.

\textbf{Assumption~\ref{assumption:regularity_coordinate}.} For the \(\theta\)-coordinate submodel,
\[
f_{t,\beta}(Z) := \psi_{\mathrm{eff},0}(Z;\, \theta_t, \ell_t, m_t) = a_{\mathrm{eff},0}(D, X)(U - tV),
\]
so
\[
\frac{f_{t,\beta} - f_{0,\beta}}{t} = -a_{\mathrm{eff},0}(D, X)\,V
\]
exactly. Since \(a_{\mathrm{eff},0}\) and \(V\) are bounded, the conditions of Lemma~\ref{lem:differentiation_varying_f} hold.

For the \(\eta\)-coordinate submodel in direction \(h = (h_\ell, h_m)\),
\[
f_{t,h}(Z) := \psi_{\mathrm{eff},0}(Z;\, \theta_0, \ell_0 + th_\ell, m_0 + th_m) = a_{\mathrm{eff},0}(D, X)(U - th_\ell(X) + t\theta_0 h_m(X)),
\]
so
\[
\frac{f_{t,h} - f_{0,h}}{t} = -a_{\mathrm{eff},0}(D, X)(h_\ell(X) - \theta_0 h_m(X)).
\]
Since all terms are bounded, the conditions of Lemma~\ref{lem:differentiation_varying_f} hold.

Since the assumptions of Theorem~\ref{thm:reverse} are  verified for the frozen efficient moment~\eqref{eq:plm_frozen_eff}, we conclude that
\[
\frac{\partial}{\partial \eta}\bE_0[\psi_{\mathrm{eff},0}(Z;\, \theta_0, \eta)]\bigg|_{\eta = \eta_0}[h] = 0 \qquad \text{for every } h \in \dot \cH,
\]
and 
\[
\bE_0[\partial_\theta \psi_{\mathrm{eff},0}(Z;\, \theta_0, \ell_0, m_0)] = -1.
\]

\subsection{Expected Density}
\label{appendix:squared_density}

Finally, we illustrate Proposition~\ref{prop:neyman_wo_lps} through a worked example in which local product structure fails. The target parameter is the expected density, a classical nonparametric functional for which the target is a known functional of the nuisance. We verify the forward direction (Theorem~\ref{thm:forward}) directly, show that local product structure fails, and then apply Proposition~\ref{prop:neyman_wo_lps} to recover the characterization of Neyman orthogonality.

\paragraph{Setup.} Let \((\cZ, \cA)\) be a measurable space with \(\sigma\)-finite dominating measure \(\nu\) satisfying \(\nu(\cZ) < \infty\). Fix a density \(p_0\) with respect to \(\nu\) satisfying
\begin{equation}\label{eq:sq_density_standing}
  0 < c_p \le p_0(z) \le C_p < \infty \quad \nu\text{-a.e.}, \qquad p_0 \text{ not } \nu\text{-a.e.\ constant}.
\end{equation}
We work in the nonparametric model \(\cP\) consisting of all densities \(p\) with respect to \(\nu\) satisfying
\[
  \operatorname*{ess\,inf}_{\nu} p > c_p/2, \qquad \|p\|_\infty < 2C_p.
\]
In particular, \(P_0 \in \cP\), and by~\eqref{eq:sq_density_standing} we see \(P_0\) lies in the interior of \(\cP\) relative to the affine hyperplane \(\left\{p : \int p \, d\nu = 1\right\}\). Define the target and nuisance functionals
\[
  \beta(P) := \int p(z)^2 \, d\nu(z), \qquad \eta(P) := p,
\]
with \(\beta_0 := \beta(P_0) = \int p_0^2 \, d\nu\) and \(\eta_0 := p_0\). Observe that \(\beta = g(\eta)\) where \(g : \cH \to \RR\) is the functional \(g(p) := \int p^2 \, d\nu\), so the target parameter carries no degrees of freedom beyond those already encoded in the nuisance.

Since \(P_0\) lies in the interior of \(\cP\), for any bounded mean-zero function \(g\), the linear tilt \(p_t := p_0(1 + tg)\) remains in \(\cP\) for sufficiently small \(|t|\) with both density bounds preserved, so the tangent space at \(P_0\) is \(\cT = L_2^0(P_0)\) by the same argument as Corollary~\ref{cor:saturation}. The ambient normed space is \(\cV := L_\infty(\nu)\), and the nuisance parameter set coincides with the densities in the model,
\[
  \cH := \left\{p \in L_\infty(\nu) : \operatorname*{ess\,inf}_\nu\, p > c_p / 2, \; \|p\|_\infty < 2C_p, \; \int p \, d\nu = 1 \right\}.
\]
Since \(\cH\) is open relative to the affine subspace \(\left\{p \in L_\infty(\nu) : \int p \, d\nu = 1\right\}\) and \(p_0 \in \cH\), the admissible perturbation space is given by
\[
  \dot{\cH} = \left\{h \in L_\infty(\nu) : \int h \, d\nu = 0\right\}.
\]

\paragraph{Fr\'echet differentiability of \(g\).} The functional \(g(p) = \int p^2 \, d\nu\) is Fr\'echet differentiable at \(p_0\) with derivative \(Dg(p_0)[h] = 2 \int p_0 h \, d\nu\). Indeed,
\[
  g(p_0 + h) - g(p_0) = 2 \int p_0 h \, d\nu + \int h^2 \, d\nu,
\]
and the remainder satisfies \(\left| \int h^2 \, d\nu \right| \le \|h\|_\infty^2 \, \nu(\cZ)\), which is \(o\bigl(\|h\|_\infty\bigr)\) since \(\nu(\cZ) < \infty\). Since \(p_0\) is not \(\nu\)-a.e.\ constant by~\eqref{eq:sq_density_standing}, there exists \(h^* \in \dot{\cH}\) with \(\int p_0 h^* \, d\nu \neq 0\), so \(Dg(p_0)[h^*] = 2\int p_0 h^* \, d\nu \neq 0\), that is, \(Dg(p_0)\) does not vanish on \(\dot{\cH}\).

\paragraph{Estimating function and influence function.} 
Define the estimating function
\begin{equation}\label{eq:sq_density_m}
  m(z;\, \beta,\, p) := 2p(z) - \int p^2 \, d\nu - \beta,
\end{equation}
and the influence function at the truth, 
\begin{equation}\label{eq:sq_density_phi}
  \varphi(z) = 2(p_0(z) - \beta_0).
\end{equation}
\subsubsection{Forward direction}
\label{subsubsec:squared_density_forward}
We verify Assumptions~\ref{assumption:correct_spec}--\ref{assumption:hellinger} and apply Theorem~\ref{thm:forward} to conclude that \(\beta\) is pathwise differentiable with influence function \(\varphi(Z) = -G^{-1} m(Z;\, \beta_0,\, p_0)\).

\textbf{Assumption~\ref{assumption:correct_spec}.} Let \(P\) be any distribution with \(\beta(P) = \beta\) and \(\eta(P) = p\). Then
\[
  \bE_P[m(Z;\, \beta,\, p)] = 2\bE_P[p(Z)] - \int p^2 \, d\nu - \beta = 2\beta - \beta - \beta = 0.
\]
\textbf{Assumption~\ref{assumption:jacobian}.} Since \(m\) is linear in \(\beta\) with coefficient \(-1\), we have \(\partial_\beta m(z;\, \beta_0,\, p_0) = -1\) identically, so
\[
  G := \bE_0[\partial_\beta m(Z;\, \beta_0,\, p_0)] = -1 \neq 0.
\]
\textbf{Assumption~\ref{assumption:neyman}.} For any \(h \in \dot{\cH}\), substituting \(p_0 + th\) into~\eqref{eq:sq_density_m} with \(\beta = \beta_0\):
\begin{align*}
  \bE_0[m(Z;\, \beta_0,\, p_0 + th)]
  &= 2 \int (p_0 + th)\, p_0 \, d\nu - \int (p_0 + th)^2 \, d\nu - \beta_0 \\
  &= \left(2\beta_0 + 2t \textstyle\int p_0 h \, d\nu\right) - \left(\beta_0 + 2t \textstyle\int p_0 h \, d\nu + t^2 \textstyle\int h^2 \, d\nu\right) - \beta_0 \\
  &= -t^2 \int h^2 \, d\nu.
\end{align*}
It follows that \(\dfrac{d}{dt}\bE_0[m(Z;\, \beta_0,\, p_0 + th)]\bigg|_{t=0} = 0\) for all \(h \in \dot{\cH}\).

\textbf{Assumptions~\ref{assumption:coord_smooth_all}--\ref{assumption:regularity_submodels}.} We take \(S = L_\infty(P_0) \cap L_2^0(P_0)\), which is dense in \(\cT = L_2^0(P_0)\). For each \(s \in S\) we use the linear tilt submodel from Lemma~\ref{lem:linear_tilt_submodel}. The induced coordinate path has \(\beta_t = \int p_t^2 \, d\nu\) and \(\eta_t = p_t = p_0 + t p_0 s\). It follows that \(\dot{\beta}_{0,s} = 2\bE_0[p_0(Z) s(Z)]\) and \(\dot{\eta}_{0,s} = p_0 s \in \dot{\cH}\), the final inclusion because \(\int p_0 s \, d\nu = \bE_0[s] = 0\).

To check Fr\'echet differentiability of \(m\) in \(L_2(P_0)\), we compute
\begin{align*}
  &m(z;\, \beta_0 + \delta\beta,\, p_0 + \delta p) - m(z;\, \beta_0,\, p_0) \\
  &\quad = 2(p_0(z) + \delta p(z)) - \int (p_0 + \delta p)^2 \, d\nu - (\beta_0 + \delta\beta) - \left[2p_0(z) - \int p_0^2 \, d\nu - \beta_0\right] \\
  &\quad = 2\delta p(z) - 2\!\int\! p_0 \,\delta p \, d\nu - \delta\beta - \int (\delta p)^2 \, d\nu.
\end{align*}
We identify the partial derivatives as
\[
  D_\beta m_0(\delta\beta)(z) := -\delta\beta, \qquad D_\eta m_0(\delta p)(z) := 2\delta p(z) - 2\!\int\! p_0 \,\delta p \, d\nu,
\]
with the remainder being \(-\int (\delta p)^2 \, d\nu\). It follows that
\[
  \frac{\left|\int (\delta p)^2 \, d\nu\right|}{|\delta\beta| + \|\delta p\|_\infty} \le \frac{\|\delta p\|_\infty^2 \, \nu(\cZ)}{|\delta\beta| + \|\delta p\|_\infty} \le \|\delta p\|_\infty \, \nu(\cZ) \to 0
\]
as \(|\delta\beta| + \|\delta p\|_\infty \to 0\), since \(\nu(\cZ) < \infty\).
 
We now verify the conditions of Lemma~\ref{lem:differentiation_varying_f}. Define \(f_{t,s}(z) := m(z;\, \beta_t,\, p_t)\). Since \(\beta_t = \int p_t^2 \, d\nu\), we have
\[
  f_{t,s}(z) = 2p_t(z) - \int p_t^2 \, d\nu - \beta_t = 2p_t(z) - 2\beta_t.
\]
For the linear tilt \(p_t = p_0(1 + ts)\), we have \(p_t(z) - p_0(z) = tp_0(z)s(z)\) and
\[
  \beta_t = \int p_0^2(1 + ts)^2 \, d\nu = \beta_0 + 2t\,\bE_0[p_0(Z)s(Z)] + t^2\,\bE_0\left[p_0(Z)s(Z)^2\right],
\]
so that \((\beta_t - \beta_0)/t = \dot{\beta}_{0,s} + t\,\bE_0\left[p_0(Z)s(Z)^2\right]\) where \(\dot{\beta}_{0,s} = 2\,\bE_0[p_0(Z)s(Z)]\). It follows that
\begin{align*}
  \frac{f_{t,s}(z) - f_0(z)}{t}
  &= 2 \cdot \frac{p_t(z) - p_0(z)}{t} - 2 \cdot \frac{\beta_t - \beta_0}{t} \\
  &= 2p_0(z)s(z) - 2\dot{\beta}_{0,s} - 2t\,\bE_0\left[p_0(Z)s(Z)^2\right].
\end{align*}
As \(t \to 0\), the final term vanishes and the quotient converges pointwise to the bounded function \(\dot{f}_0(z) := 2p_0(z)s(z) - 2\dot{\beta}_{0,s}\). Each of the three terms above is bounded uniformly in \(z\) and \(|t| \le 1\):
\[
  |2p_0(z)s(z)| \le 2C_p\|s\|_\infty, \qquad \bigl|2\dot{\beta}_{0,s}\bigr| < \infty, \qquad \left|2t\,\bE_0\left[p_0(Z)s(Z)^2\right]\right| \le 2\,\bE_0\left[p_0(Z)s(Z)^2\right],
\]
so \(\|(f_{t,s} - f_0)/t\|_\infty \le M\) for some constant \(M\) independent of \(t\). The first condition then follows by dominated convergence, and the second condition follows because a function bounded by \(M\) has second moment at most \(M^2\) under any probability measure.

\textbf{Assumption~\ref{assumption:hellinger}.} We show that \(\beta\) is Hellinger Lipschitz. For any \(P_1, P_2 \in \cP\), since \(\|p_i\|_\infty \le 2C_p\),
\begin{align*}
  |\beta(P_1) - \beta(P_2)| &= \left|\int (p_1 - p_2)(p_1 + p_2) \, d\nu\right|
  \le 4C_p \int |p_1 - p_2| \, d\nu.
\end{align*}
Writing \(|p_1 - p_2| = \bigl|\sqrt{p_1} - \sqrt{p_2}\bigr|\,\bigl(\sqrt{p_1} + \sqrt{p_2}\bigr)\) and applying Cauchy--Schwarz,
\[
  \int |p_1 - p_2| \, d\nu \le \bigl\|\sqrt{p_1} - \sqrt{p_2}\bigr\|_{L_2(\nu)} \cdot \bigl\|\sqrt{p_1} + \sqrt{p_2}\bigr\|_{L_2(\nu)} \le 2\sqrt{2}\, H(P_1, P_2),
\]
since \(\bigl\|\sqrt{p_1} + \sqrt{p_2}\bigr\|_{L_2(\nu)}^2 \le 2\int (p_1 + p_2)\,d\nu = 4\). Since \(\|p_i\|_\infty \le 2C_p\) for all \(P_i \in \cP\), we conclude that \(\left|\beta(P_1) - \beta(P_2)\right| \le 8\sqrt{2}\, C_p \, H(P_1, P_2)\).

Since all the assumptions of Theorem~\ref{thm:forward} hold, we conclude that \(\beta\) is pathwise differentiable with influence function
\[
  \varphi(Z) = -G^{-1} m(Z;\, \beta_0,\, p_0) = m(Z;\, \beta_0,\, p_0) = 2p_0(Z) - 2\beta_0.
\]

\subsubsection{Reverse direction}
\label{subsubsec:squared_density_reverse}

Since \(\beta = g(\eta)\) with \(g(p) = \int p^2 \, d\nu\), the chain rule ensures that every regular submodel must satisfy \(\dot{\beta}_{0,s} = Dg(p_0)\bigl[\dot{\eta}_{0,s}\bigr] = 2\int p_0 \dot{\eta}_{0,s} \, d\nu\). A \(\beta\)-coordinate submodel requires \(\dot{\beta}_{0,s} = 1\) and \(\dot{\eta}_{0,s} = 0\). However, \(\dot{\eta}_{0,s} = 0\) forces \(\dot{\beta}_{0,s} = Dg(p_0)[0] = 0 \neq 1\), and no such submodel exists. An \(\eta\)-coordinate submodel in direction \(h\) requires \(\dot{\beta}_{0,s} = 0\) and \(\dot{\eta}_{0,s} = h\). However, \(\dot{\eta}_{0,s} = h\) forces \(\dot{\beta}_{0,s} = 2 \int p_0 h \, d\nu\), which is generally nonzero. Therefore, Assumption~\ref{assumption:product-structure} fails and Theorem~\ref{thm:reverse} does not apply.

We now verify the conditions of Proposition~\ref{prop:neyman_wo_lps} and apply it to recover the characterization of Neyman orthogonality. Correct specification, Fr\'echet differentiability of \(m\), Fr\'echet differentiability of \(g\) with \(Dg(p_0)\) non-vanishing on \(\dot{\cH}\), and the influence function representation \(m(z;\, \beta_0,\, p_0) = 2p_0(z) - 2\beta_0 = \varphi(z)\) have all been established above.

\textbf{Submodel condition.} For every \(h \in \dot{\cH}\), we construct a regular submodel through \(P_0\) with bounded score \(s\), nuisance derivative \(\dot{\eta}_{0,s} = h\), and satisfying the regularity conditions of Lemma~\ref{lem:differentiation_varying_f}.

Define the score \(s := h / p_0\). Since \(p_0 \ge c_p\) by~\eqref{eq:sq_density_standing}, \(|s(z)| = |h(z)| / p_0(z) \le \|h\|_\infty / c_p\), so \(s\) is bounded. Moreover, \(\bE_0[s] = \int (h/p_0) \, p_0 \, d\nu = \int h \, d\nu = 0\) since \(h \in \dot{\cH}\). By Lemma~\ref{lem:linear_tilt_submodel}, the linear tilt \(p_t := p_0(1 + ts) = p_0 + th\) defines a regular QMD submodel through \(P_0\) with score \(s\). Moreover, for \(|t| < c_p/(2\|h\|_\infty)\) the bounds \(c_p/2 < p_t < 2C_p\) hold, so \(p_t \in \cP\). The nuisance derivative is
\[
  \dot{\eta}_{0,s} = \frac{d}{dt} p_t \bigg|_{t = 0} = p_0 s = h,
\]
and the chain rule gives \(\dot{\beta}_{0,s} = Dg(p_0)[h] = 2 \int p_0 h \, d\nu\), so the induced coordinate path is differentiable at \(t = 0\).

For the regularity condition, since \(p_t = p_0 + th\) and \(\beta_t = \beta_0 + 2t\int p_0 h \, d\nu + t^2 \int h^2 \, d\nu\), the function \(f_{t,s}(z) = m(z;\, \beta_t,\, p_t) = 2p_t(z) - 2\beta_t\) satisfies
\[
  \frac{f_{t,s}(z) - f_0(z)}{t} = 2h(z) - 4\!\int\! p_0 h \, d\nu - 2t\!\int\! h^2 \, d\nu,
\]
which converges in \(L_2(P_0)\) to the bounded function \(2h(z) - 4 \int p_0 h \, d\nu\). Since each term in the difference quotient is bounded uniformly in \(z\) and \(|t| \le 1\), the conditions of Lemma~\ref{lem:differentiation_varying_f} hold. 

Finally, since \(Dg(p_0)\) does not vanish on \(\dot{\cH}\), the conditions of Proposition~\ref{prop:neyman_wo_lps} hold. Hence, we conclude that
\[
  m \text{ is Neyman orthogonal} \quad \Longleftrightarrow \quad G = -1.
\]

\end{document}